\documentclass[10pt, twocolumn]{article}
\usepackage[top=1.0in, bottom=1.25in,left=0.75in, right=0.75in]{geometry}

\usepackage{graphicx}
\usepackage[utf8]{inputenc}
\usepackage[none]{hyphenat}
\usepackage[sorting=none,doi=false,url=false,isbn=false,eprint=false]{biblatex}

\usepackage{xcolor}
\usepackage{pagecolor}
\definecolor{cobaltblue}{RGB}{0,60,170}
\definecolor{shadowgrey}{RGB}{20,21,27}

\usepackage{enumitem}

\usepackage[dvipsnames]{xcolor}

\definecolor{cobaltblue}{RGB}{0,60,170}
\definecolor{niceblue}{RGB}{50,150,200}

\usepackage[colorlinks=true, linkcolor=niceblue,
citecolor=ForestGreen, urlcolor=cobaltblue]{hyperref}

\usepackage{caption}
\usepackage{subcaption}
\usepackage{algorithm}
\usepackage{algpseudocode}
\usepackage{authblk}
\usepackage{bm}

\usepackage{circuitikz}
\usepackage{tikz}
\usetikzlibrary{calc, intersections, through, backgrounds, arrows, positioning, automata, shapes, shapes.geometric, fit}

\usepackage{amsmath}
\usepackage{amssymb}

\usepackage{geometry}
\usepackage{array}
\usepackage{booktabs}
\usepackage{graphicx}
\usepackage{listings}
\usepackage{fancyhdr}
\usepackage{yfonts}
\usepackage{mathrsfs}
\usepackage[mathscr]{euscript}
\usepackage{relsize}
\usepackage{bigints}
\usepackage{verbatim}
\usepackage[english]{babel}
\usepackage{slantsc}
\usepackage{amsthm}
\usepackage{bbm}
\usepackage{placeins}
\usepackage{cancel}
\usepackage{setspace}
\usepackage{mathtools}
\usepackage{blkarray}
\usepackage{multirow}
\usepackage{bigdelim}

\usepackage{changepage}
\usepackage{mathtools}
\usepackage{optidef}
\usepackage{physics}
\usepackage{cleveref}
\usepackage{cuted}
\usepackage{csquotes}

\usepackage{xspace}

\usepackage{subcaption}
\DeclareMathAlphabet{\mathpzc}{OT1}{pzc}{m}{it}

\newcommand{\n}{\vspace{\baselineskip}}

\newcommand{\R}{\mathbb{R}}
\newcommand{\C}{\mathbb{C}}

\newcommand{\N}{\mathbb{N}}
\newcommand{\E}{\mathbb{E}}

\newcommand{\oask}{\circledast}

\newcommand{\sm}{\smallskip}

\newcommand{\dg}{\dagger}

\newcommand{\gt}{\rightarrow}
\newcommand{\rgt}{\leftarrow}

\usepackage{mdframed}
\usepackage{etoolbox}

\makeatletter
\newtheoremstyle{mystyle}
  {1em}
  {1em}
  {\raggedright\sffamily}
  {0em}
  {\raggedright\normalsize\bfseries\sffamily}
  {.}
  {1em}
  {}
\makeatother

\theoremstyle{mystyle}

\newtheorem{thm}{Theorem}
\newtheorem{cor}{Corollary}[thm]

\newtheorem{remark}{Remark}

\newtheoremstyle{propositional}
{10pt}
{10pt}
{
	\addtolength{\linewidth}{-2.0em}
	\parshape 1 2.0em \linewidth}
{}
{\sc}
{:}
{.5em}
{}
\makeatother

\theoremstyle{propositional}

\newtheoremstyle{definitive}
{10pt}
{10pt}
{
	\addtolength{\linewidth}{-2.0em}
	\parshape 1 0.0em \linewidth}
{ }
{\sffamily\bfseries}
{:}
{.5em}
{}
\makeatother

\theoremstyle{definitive}
\newtheorem{defi}{Definition}

\newtheoremstyle{SIdefinitive}
{10pt}
{10pt}
{
	\addtolength{\linewidth}{-4.0em}
	\parshape 1 2.0em \linewidth}
{ }
{\sffamily\bfseries}
{:}
{.5em}
{}
\makeatother

\theoremstyle{SIdefinitive}
\newtheorem{sdefi}{Definition}

\makeatletter
\newtheoremstyle{SImystyle}
  {1em}
  {1em}
  {\addtolength{\linewidth}{-4.0em}
	\parshape 1 2.0em \linewidth
    \raggedright\sffamily}
  {0em}
  {\raggedright\normalsize\bfseries\sffamily}
  {.}
  {1em}
  {}
\makeatother

\theoremstyle{SImystyle}
\newtheorem{sthm}{Theorem}
\newtheorem{scor}{Corollary}[sthm]
\newtheorem{slem}[sthm]{Lemma}

\newtheorem{sremark}{Remark}

\newcounter{numb}
\newcounter{bean}
\newcounter{count}
\newcounter{count2}
\newcounter{count3}

\makeatletter

\newlength\tdima
\newcommand\tabfill[1]{
	\setlength\tdima{\linewidth}
	\addtolength\tdima{\@totalleftmargin}
	\addtolength\tdima{-\dimen\@curtab}
	\parbox[t]{\tdima}{#1\ifhmode\strut\fi}}

\newcommand\mytabs{\hspace*{5cm}\=\hspace{1cm}\=\hspace{2cm}}

\makeatother

\usepackage{titlesec}
\titleformat*{\section}{\raggedright\Large\bfseries\sffamily}
\titleformat*{\subsection}{\raggedright\large\bfseries\sffamily}
\titleformat*{\subsubsection}{\raggedright\bfseries\sffamily}

\usepackage{tocloft}

\usepackage{newfloat}
\DeclareFloatingEnvironment[fileext=lof]{suppfigure}
\crefname{suppfigure}{Supplementary Figure}{Supplementary Figures}

\DeclareFloatingEnvironment[fileext=lot]{supptable}

\crefname{supptable}{Supplementary Table}{Supplementary Tables}
\Crefname{supptable}{Supplementary Table}{Supplementary Tables}

\newboolean{showcomments}
\setboolean{showcomments}{true}
\DeclareRobustCommand{\HL}[1]{\ifthenelse{\boolean{showcomments}}{{\color{YellowOrange}{\bf#1}}}{}}

\renewcommand{\l}{\left}
\renewcommand{\r}{\right}
\renewcommand{\t}[1]{\text{#1}}
\renewcommand{\tt}[1]{\texttt{#1}}
\renewcommand{\dg}{\dagger}

\newcommand{\tsbf}[1]{\textbf{\sffamily#1}}
\newcommand{\tsf}[1]{{\sffamily #1}}
\newcommand{\tbf}[1]{\textbf{#1}}

\newcommand{\mf}[1]{\mathfrak{#1}}
\newcommand{\mc}[1]{\mathcal{#1}}

\newcommand{\la}{\left\langle\,}
\newcommand{\ra}{\,\right\rangle}

\usepackage{afterpage}

\renewcommand{\P}{\mathcal{P}}

\renewcommand{\L}{\mathcal{L}}
\newcommand{\G}{\mathcal{G}}
\newcommand{\g}{\mathfrak{g}}
\newcommand{\B}{\mathcal{B}}

\newcommand{\Qc}{\mathcal{Q}}

\newcommand{\I}{\mathbbm{1}}

\newcommand{\bmt}{\bm{\theta}}

\newcommand{\ad}[2]{\text{ad}_{#1}\left(#2\right)}

\newcommand{\orb}{\text{orb}_{\mathfrak{g}}^{iO}}
\newcommand{\orbp}{\left(\text{orb}_{\mathfrak{g}}^{iO}\right)^{\perp}}
\newcommand{\orbj}{\text{orb}_{\mathfrak{g}}^{iO_{j}}}
\newcommand{\dens}{\text{density}}

\renewcommand{\Tr}[1]{\text{Tr}\left( #1 \right)}

\newcommand{\com}[2]{\left[ #1 , #2 \right]}

\newcommand{\Oc}{\mathcal{O}}
\newcommand{\h}{\hspace{-0.15em}}

\newcommand{\tprojs}{\texttt{projs}}

\newcommand{\tadjs}{\texttt{adjs}}

\newcommand{\trho}{\texttt{rho}}

\renewcommand{\trho}{\texttt{p\_rho}}

\newcommand{\tobs}{\texttt{obs}}

\newcommand{\tgens}{\texttt{gens}}

\newcommand{\tsig}{\texttt{sig}}

\newcommand{\peq}{\phantom{=}}

\newcommand{\results}{\hyperlink{result}{Results}\xspace}

\newcommand{\methods}{\hyperlink{methods}{Methods}\xspace}

\newcommand{\vr}{\varrho}
\newcommand{\hs}{\text{HS}}
\newcommand{\dos}{\text{DOS}}

\crefname{equation}{Eq.}{Eqs.}
\crefname{section}{Section}{Sections}
\crefname{figure}{Figure}{Figures}
\crefname{table}{Table}{Tables}
\crefname{appendix}{Supplementary Section}{Supplementary Sections}
\crefname{theorem}{Theorem}{Theorems}
\crefname{thm}{Theorem}{Theorems}
\crefname{defi}{Def.}{Defs.}
\crefname{conjecture}{Conjecture}{Conjectures}
\crefname{proposition}{Prop.}{Props.}
\crefname{lemma}{Lemma}{Lemmas}
\crefname{corollary}{Corollary}{Corollaries}
\crefname{algorithm}{Alg.}{Algs.}

\crefname{sthm}{Theorem}{Theorems}
\crefname{sdefi}{Def.}{Defs.}
\crefname{slem}{Lemma}{Lemmas}

\crefname{part}{Supplementary Part}{Supplementary Parts}
\Crefname{part}{Supplementary Part}{Supplementary Parts}

\usepackage{tcolorbox}
\definecolor{fundamental}{RGB}{55, 110, 111}
\newtcolorbox[auto counter]{pabox}[2][]{fonttitle=\bfseries,
title=Box~\thetcbcounter: #2,#1,colframe=gray}
\newtcolorbox[use counter from=pabox]{mybox}[2][]{
floatplacement=t,float,
colback=fundamental!5!white,colframe=fundamental!75!black,title=Box~\thetcbcounter: #2,#1}

\newlength{\ointextsep}
\DeclareDocumentCommand{\ket}{ m }{\vphantom{\sum}\left| #1 \right\rangle}
\DeclareDocumentCommand{\bra}{ m }{\vphantom{\sum}\left\langle #1 \right|}
\DeclareDocumentCommand{\ketbra}{ m m }{\vphantom{\sum}\left| #1 \middle\rangle\!\middle\langle #2 \right|}

\let\oldra\ra
\renewcommand{\ra}{\vphantom{\sum}\oldra}
\let\oldla\la
\renewcommand{\la}{\oldla\vphantom{\sum}}

\usepackage{lipsum}
\usepackage{fancyhdr}
\makeatletter
\def\fps@figure{!t}
\makeatother

\newtcolorbox{algobox}[1][]{
    colback=gray!20,
    colframe=black,
    arc=5pt,
    boxrule=0.5pt,
    left=0pt, right=0pt, top=0pt, bottom=0pt
}

\let\oldbra\bra
\renewcommand{\bra}[1]{\vphantom{\sum}\oldbra{#1}}
\let\oldket\ket
\renewcommand{\ket}[1]{\vphantom{\sum}\oldket{#1}}
\let\oldketbra\ketbra
\renewcommand{\ketbra}[2]{\vphantom{\sum}\oldketbra{#1}{#2}\vphantom{\sum}}

\newcommand{\vem}{0.25em}

\usepackage[normalem]{ulem}
\definecolor{editblue}{RGB}{30,110,220}
\renewcommand{\HL}[1]{{\color{editblue}#1}}

\makeatletter
\renewenvironment{proof}[1][\proofname]{
  \par\pushQED{\qed}
  \normalfont\topsep6\p@\@plus6\p@\relax
  \list{}{\leftmargin=1.5em\rightmargin=1.5em\itemindent=0pt\labelwidth=0pt\labelsep=\parindent\topsep=0pt}
  \item[\hskip\labelsep\itshape#1\@addpunct{.}]\ignorespaces
}{
  \popQED\endlist\@endpefalse
}
\makeatother

\renewcommand{\results}{\hyperlink{results}{Results}\xspace}

\renewcommand{\methods}{\hyperlink{methods}{Methods}\xspace}

\crefname{cor}{Corollary}{Corollaries}

\title{\huge\bfseries\sffamily Dynamic Observable Subspaces: Fast Simulation of Observable Quantum Dynamics}
\author[,1]{\large\sffamily Hannes Leipold\thanks{Email: \texttt{hleipold@fujitsu.com}}}
\affil[1]{\large\sffamily Fujitsu Research of America, Santa Clara, CA, U.S.A.}
\date{\vspace{-1em}\large\sffamily October, 2026}

\usepackage{titling}
\begin{document}

\maketitle

\begin{strip}
\vspace{-6em}
\begin{abstract}
By decomposing quantum dynamics across the Lie orbits of a list of observables, we find polynomial bounded classical simulations for the dynamics of the quantum system given a polynomial sized dynamic Lie algebra (DLA) for the generators of the quantum system. To do so, we describe how to construct the Dynamic Observable Subspace (DOS) that captures all the relevant dynamics for calculating expectation values for a specific observable for Pauli strings, diffusor mixers, and general local generators. Efficient sparse matrix representation, decoupling nonlinearity from such representation and basis construction with permissible pruning allows us to simulate the closed dynamics of such systems with hundreds of qubits. Moreover, we find that while restricted DOS circuits may not express the space of an associated Hamiltonian, they can dramatically outperform a fully expressive circuit due to the absence of the Barren Plateau. While classically simulatable, such circuits can still exhibit a form of quantum advantage through inference, act as warm starting for more expressive circuits, and find high quality trial states for Quantum Amplitude Amplification or Quantum Phase Estimation. In the simulation of quantum dynamics, we find that our restricted circuit can be simulated in polynomial time while producing high quality guiding state for downstream tasks like Quantum Phase Estimation with significantly lower energy than circuits with full expressivity.
\end{abstract}

\vspace{1em}
\textbf{\tsf{Keywords:}} Quantum Machine Learning, Variational Quantum Algorithms, Simulation of Quantum Systems, Many body Physics, Hybrid Classical and Quantum Systems
\end{strip}

\hypertarget{intro}{}
\section*{Introduction}

Simulating large scale quantum systems has become important in many domains. As digital quantum computers continue to grow, simulating circuits that are executable on such hardware has gained much attention~\cite{preskill_quantum_2018,cerezo_variational_2021}, with approaches ranging from state vector simulation~\cite{intro_qulacs_2021} and tensor network methods~\cite{intro_vidal_2003,intro_markov_shi_2008} to Pauli propagation in the Heisenberg picture~\cite{intro_rall_pauli_2019,intro_rudolph_pauli_2025}.

In this work, we provide a system for simulating quantum dynamics by decomposing them across the Lie orbits of the observables, building on Lie algebra simulation methods for variational quantum algorithms~\cite{goh_lie-algebraic_2023}. Refs.~\cite{somma2005quantum, somma2006efficient} recognized that instead of propagating a state vector or density operator over the entire Hilbert space, it is sufficient to evolve the observables of interest within a low dimensional operator space. Ref.~\cite{goh_lie-algebraic_2023} reformulated and extended this approach for variational quantum algorithms through $\g$-sim, which propagates observables within invariant operator subspaces under the adjoint action of the DLA. More recently, Ref.~\cite{barligea2026enabling} further broadened the practical scope of $\g$-sim by developing symmetry adapted representations and efficient preprocessing for additional polynomial dimensional DLAs.

The starting point is that many quantum algorithms only require the expectation value of a specified observable, rather than a full representation of the state. We therefore identify the matrix subspace that contains the dynamics relevant to that observable in the skew Heisenberg picture. This subspace, which we call the Dynamic Observable Subspace (DOS), captures the part of the density operator or observable that can affect the measured expectation value. Recent independent work arrived at an equivalent foundational structure~\cite{barligea2026lie}, termed the ``reachable operator module'' (Definition 1 of Ref.~\cite{barligea2026lie}).

Because the DOS is invariant under the adjoint action of the DLA, unitary evolution generated by elements of $\g$ restricts to a linear evolution within the DOS~\cite{goh_lie-algebraic_2023}. This evolution may be represented by exponentiating the corresponding adjoint superoperator, while for important classes of generators we derive evolution superoperators based on sparse static operators well defined over this matrix space and trigonometric coefficients. For important classes of gates, including Pauli gates, projectors such as diffusor mixers~\cite{leipold_constructing_2021,leipold2026imposing,hadfield_analytical_2022}, and general Hamiltonians over few qubits, explicit linear updates with parameterized trigonometric coefficients can simulate the dynamics within the DOS.

The dimension of the DOS is bounded by the dimension of the multiplicative sector for that observable over the dynamic Lie algebra (DLA) of the circuit. In the cases we study, this dimension is at most the square of the DLA dimension. Thus, when the DLA is polynomial sized, the DOS is also polynomial sized and the relevant observable dynamics can be simulated tractably. The resulting framework covers observables contained in the DLA as well as observables outside the DLA, provided their induced DOS remains small.

This view centered on observables also changes how we interpret restricted variational circuits. We show that circuits with restricted DLAs can dramatically outperform standard circuit ans{\"a}tze, even when the observable is not represented in the generator set of the circuit. In several Hamiltonian examples, including mixer structures related to the quantum alternating operator ansatz~\cite{wang_x_2020}, these circuits remain tractable to train and begin to outperform more expressive circuits once the number of qubits is above $12$.

These restricted circuits open the door for tractable preparation of high quality trial states for Quantum Phase Estimation (QPE)~\cite{intro_kitaev_1995,intro_abrams_lloyd_1999}, since the overlap of the prepared state with the ground state can be orders of magnitude larger than for an expressive circuit affected by barren plateaus~\cite{mcclean2018barren,holmes_connecting_2022,cerezo_cost_2021}. This suggests a practical workflow in which classical simulation trains a restricted circuit and a quantum processor is used for fast inference or downstream phase estimation.

\begin{figure*}[!t]
\centering
\includegraphics[width=0.6\textwidth]{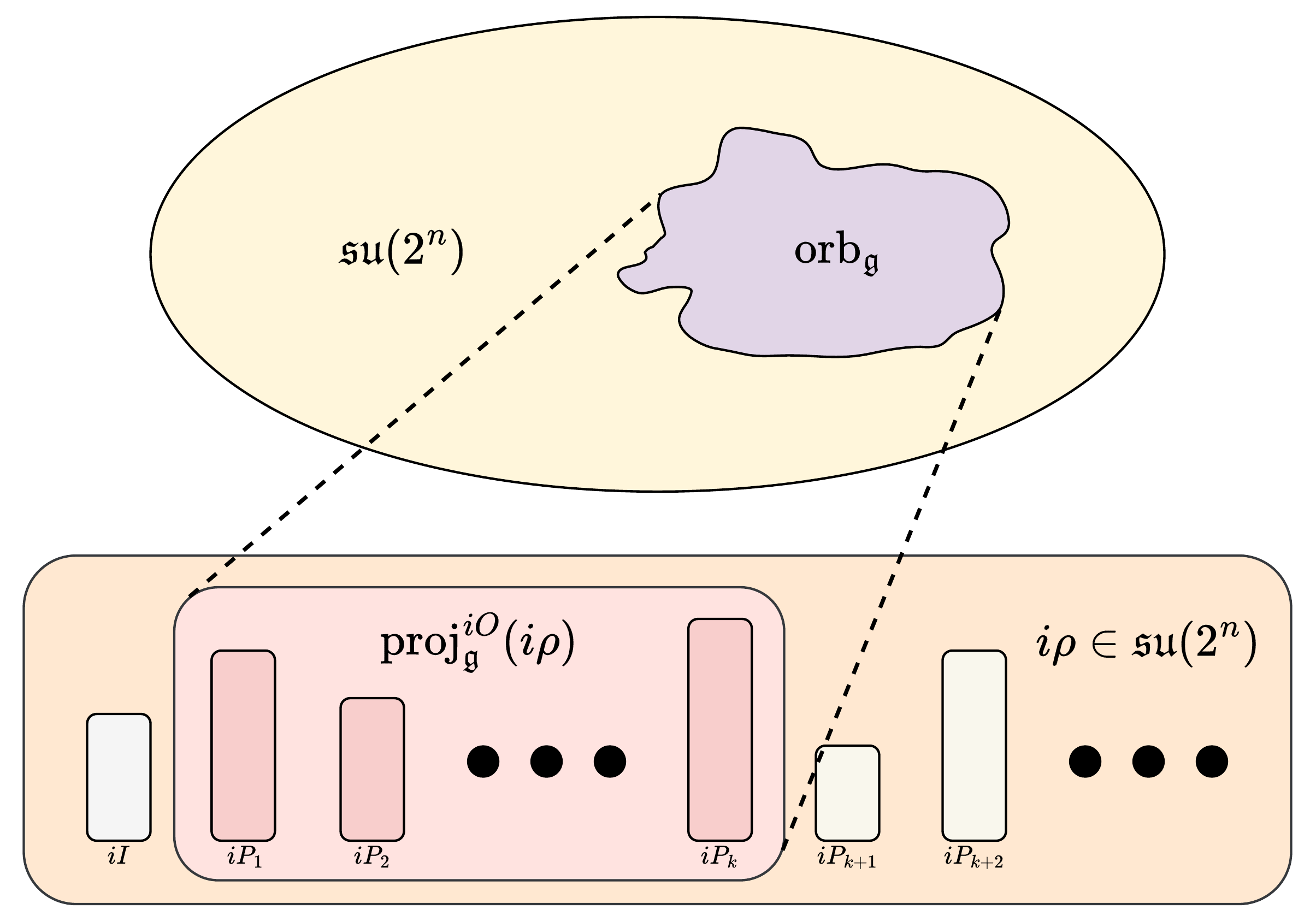}
 \caption{\tsbf{Observable Subspace of Quantum Dynamics.} Given an observable $iO$ and a dynamical Lie algebra $\g$, nested commutation with elements of $\g$ generates the Dynamic Observable Subspace $\orb$, the minimal invariant operator subspace containing $iO$. The expectation value depends only on the projection of the density operator into this subspace, such that $\t{Tr}\l( O \, \rho(T) \r) = \la iO, \t{proj}_{\g}^{iO}(i\rho(T)) \ra_{\dos}$. In the special case $iO\in\g$, it follows that the DOS is contained within the DLA, $\orb\subseteq\g$.}
\label{fig:img_dynsub}
\end{figure*}

\hypertarget{results}{}
\section*{Results}

\subsection*{Skew-Hermitian Model of Schr{\"o}dinger Dynamics}

A variational quantum algorithm (VQA) leverages a parameterized quantum circuit (PQC) to minimize a given loss. We discuss the skew-Hermitian Schr{\"o}dinger picture for how the quantum system evolves under this framework. Recall that in the Schr{\"o}dinger picture, the density operator of the quantum system evolves forward in time until some measurement event as described by the standard postulates of quantum measurement~\cite{intro_nielsen_chuang_2010} and \cref{part1:dos}.

Given a collection of gates $ \G = \{ G_1, \ldots, G_K \} $ with $G_j \in \C^{2^{n} \times 2^{n}}$ such that each is Hermitian $G_{j} = G_{j}^{\dg}$, a variational quantum algorithm (VQA) with individual gate parameters applies the unitary

\begin{align}
\label{def:vqa}
U(\bmt) = \prod_{l=1}^{L} \prod_{k=1}^{K} e^{i \theta_{k+(l-1)K} G_{k}} .
\end{align}
on the initial skew-Hermitian state $ i \rho(0) = i \ketbra{\psi}{\psi} $ for a pure state $ \ket{\psi} $. Specifically, each $G_{j}$ generates a unitary gate $\tt{R}^{G_j}\l( \theta \r) = e^{i \theta G_{j}}$.

Then we define a time slice of $U(\bmt)$ from time $t=j$ to time $t=k$ as

\begin{align}
\label{circ_tslice}
U_{j:k}(\bmt) = \prod_{t=j}^{k} e^{i \theta_{t} G_{(t-1) \t{mod} K + 1}},
\end{align}
which is the unitary applied between $ j $ and $ k $ to the quantum system based on $\bmt$.

In the skew Schr{\"o}dinger picture, we can define the state at a time point $t$ as
\begin{align}
\label{def:schro}
i \rho(t) = U_{1:t}(\bmt) \, i \rho(0) \, U_{1:t}^{\dg}(\bmt) .
\end{align}

Given a target Hermitian observable $ O $, the expected measurement outcomes define a loss landscape over $\bmt$ as
\begin{align}
\label{eq:vqa_loss}
\L(\bmt) &= \t{Tr}\l( O \, U(\bmt) \, \rho(0) \, U^\dg(\bmt) \r) \\
&= -\t{Tr}\l( i O \, U(\bmt) \, i \rho(0) \, U^\dg(\bmt) \r) \\
&= \t{Tr}\l( \l( i O \r)^{\dg} \, i \rho(LK) \r) \\
&= \la i O , i \rho(LK) \ra_{\t{HS}} ,
\end{align}
where the last line is the standard Hilbert-Schmidt inner product for skew-Hermitian operators
\begin{align}
\la i A, i B \ra_{\t{HS}} = \t{Tr}\l( (iA)^{\dg} \, iB \r) = \t{Tr}\l( A^{\dg} \, B \r) = \t{Tr}\l( A \, B \r) .
\end{align}

\subsection*{The Lie Algebra of VQA Dynamics}

Closed quantum system dynamics, via the Schr{\"o}dinger von Neumann equation, is governed by commutation between the instantaneous Hamiltonian $H(t)$ and the density operator $\rho(t)$. Under unitary evolution of idealized quantum circuits, the dynamics are thereby determined by the commutation of the generators of each gate.

In particular, the possible dynamics of a parameterized quantum circuit are encapsulated by the \textit{Dynamic Lie Algebra} (DLA) of the circuit.

\begin{defi}[Dynamic Lie Algebra]\label{def:dla}
Given a collection of Hermitian generators $\G$, the Dynamic Lie Algebra $\g$ (DLA) is the smallest real matrix Lie algebra containing the skew-Hermitian generators $\{ i G : G \in \G \}$:
\begin{align}
\g = \t{span}_{\R}\l( \B^{\g} \r),
\end{align}
where $ \B^\g = \l\{ i B_{1}^\g, \ldots, i B_{\t{dim}\l( \g \r)}^\g \r\} $ is an orthonormal basis constructed via recursive commutation, initialized by
\begin{align}
\B_0^\g = \t{Gram-Schmidt}\l( \l\{ i G : G \in \G \r\} \r),
\end{align}
and updated by
\begin{align}
\Delta_{j+1}^\g &= \t{Gram-Schmidt}\l( \l\{ \l[ i G, i B \r] : G \in \G, i B \in \B_{j}^\g \r\} \r), \nonumber \\
\B_{j+1}^\g &=  \Delta_{j+1}^\g \cup \B_{j}^\g ,
\end{align}
where Gram-Schmidt removes components already in the span of $\B_j^\g$. The construction terminates when no new terms are generated, so $ \B^\g = \B_{k}^\g = \B_{k+1}^\g$ for nesting depth $k$.
\end{defi}

Unitaries expressible by controlling the gate generators $\G$ are members of the \textit{Lie group} associated with the DLA.

\begin{defi}[Lie Group]
The Lie group $e^{\g}$ of a Lie algebra $\g$ is the set of unitaries generated by $\g$:
\begin{align}
e^{\g} = \l\{ \prod_{j=1}^{J} e^{i A_{j}} : i A_{j} \in \g, J \in \N \r\} .
\end{align}
Since each $i A_{j} \in \g$, there exists $\omega^{(j)} \in \R^{\t{dim}\l( \g \r)}$ such that
\begin{align}
i A_j = \sum_{k=1}^{\t{dim}\l( \g \r)} \omega^{(j)}_{k} \, i B_{k} .
\end{align}
Thus each skew-Hermitian generator $i A_j$ is a linear combination of skew-Hermitian operators in the DLA.
\end{defi}

This means that for any parameterized circuit $U(\bmt)$ with $ \bmt \in \R^{LK} $, we have
\begin{align}
U(\bmt) \in e^{\g},
\end{align}
so that $U(\bmt)$ can be expressed as a finite product of exponentials of elements in $\g$. In this sense, the control parameters $\bmt$ induce trajectories within the Lie group $e^{\g}$, whose dimension is governed by $\t{dim}(\g)$~\cite{intro_dalessandro_2007}. As such, the dimension of the DLA is closely related to many important phenomena in quantum control theory. In particular, $\t{dim}(\g)$ characterizes the effective degrees of freedom of the reachable unitary manifold~\cite{intro_dalessandro_2007,larocca2022diagnosing}.

\subsection*{The Lie Orbit of the Skew Heisenberg Observable}

In the Heisenberg picture, as also exploited in Lie-algebraic simulation methods~\cite{somma2006efficient}, the expected measurement outcome of observable $ O $ is still described by \cref{eq:vqa_loss}. Instead of evolving the initial density state to a final state and then computing the trace of $i O$ and $i \rho(LK)$, we can define $ i O(t) $ based on the time $t$ from $T=LK$:
\begin{align}
\label{def:hei}
i O(t) = U_{t+1:LK}(\bmt)^{\dg} \, i O \, U_{t+1:LK}(\bmt) .
\end{align}

In particular, in the skew Heisenberg picture, $i O = i O(LK)$ is evolved \textit{backwards} to the initial time $0$ such that
\begin{align}
\L\l( \bmt \r) &= \t{Tr}\l( \l( O(0) \r)^{\dg} \, \rho(0) \r) = \t{Tr}\l( O(0) \, \rho(0) \r) \\
&= \la i \rho(0), i O(0) \ra_{\hs} .
\end{align}

\begin{figure*}[!t]
\centering
\includegraphics[width=1.0\columnwidth]{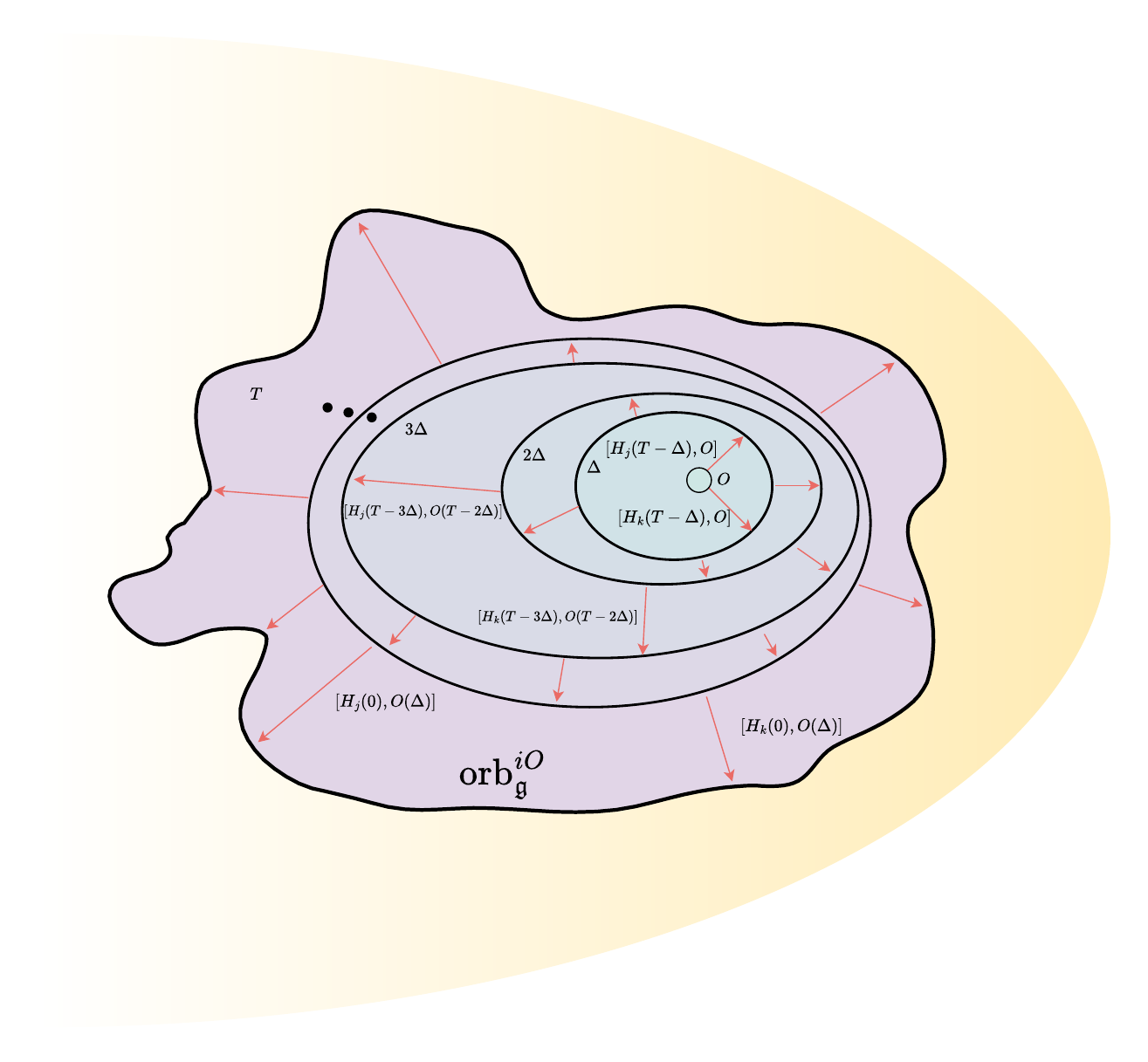}
 \caption{\tsbf{Liouville-von Neumann Dynamics in the Heisenberg Picture.} In the skew Heisenberg picture, an observable $i O$ evolves by a Hamiltonian $H(t)=\sum_k H_k(t)$ in reverse time from $ T $ to $ 0 $ to an operator state $i O(0)$ contained inside $\t{orb}_{\g}^{iO}$. Then the measurement of $i O$ is associated with $\Tr{i O(0) \, \t{proj}_{\g}^{iO}(i \rho(0))} $ in the Heisenberg picture (or $\Tr{i O \, \t{proj}_{\g}^{iO}(i \rho(T))}$ in the Schr{\"o}dinger picture).}
\label{fig:lvn_heisenberg}
\end{figure*}

Moreover, the loss can be seen as the trace of the Schr{\"o}dinger density operator and Heisenberg observable at any time $t \in [0, LK]$:
\begin{align}
\mathcal{L}\l( \bmt \r) &= \t{Tr}\l( O(t) \, \rho(t) \r) \\
&= \la i O(t) , i \rho(t) \ra_{\hs}.
\end{align}

Then the relevant subspace for an observable $i O$ is given by the space spanned through nested commutation with each $i G$ for $G \in \G$. This leads to a definition of the Dynamic Observable Subspace (DOS) analogous to the DLA.

\begin{defi}[Dynamic Observable Subspace]\label{def:dos}
Given a collection of Hermitian generators $\G$ and an observable $O$, the Dynamic Observable Subspace (DOS) is the smallest real matrix subspace containing $i O$ and invariant under commutation with each $i G$ for $G \in \G$:
\begin{align}
\t{orb}_\g^{iO} = \t{span}_{\R}\l( \B \r)
\end{align}
where $ \B = \l\{ i B_{1}, \ldots, i B_{\t{dim}\l( \t{orb}_\g^{iO} \r) } \r\} $ is an orthonormal basis constructed via recursive nested commutation, initialized by
\begin{align}
\B_0 = \t{Gram-Schmidt}\l( \l\{ i O \r\} \r),
\end{align}
and updated by
\begin{align}
\Delta_{k+1} &= \t{Gram-Schmidt}\l( \l\{ \l[ i G, i B \r] : G \in \G, i B \in \B_{k} \r\} \r), \nonumber \\
\B_{k+1} &= \Delta_{k+1} \cup \B_{k},
\end{align}
where Gram-Schmidt removes components already in the span of $\B_k$. The construction terminates when no new terms are generated, so $ \B = \B_{k} = \B_{k+1} $ for some terminal $k$.
\end{defi}

The DOS is an invariant operator subspace of the type considered in $\g$-sim~\cite{goh_lie-algebraic_2023}, with the additional specification that it is the minimal invariant subspace generated by the observable $iO$. See \methods for an algorithmic description of constructing the DOS. \cref{fig:lvn_heisenberg} depicts how through an initial observable $i O$, nested commutation with the generator set defines a space in which the relevant part of the initial density operator evolves.

Recall that the matrix trace over $\C^{2^n \times 2^n}$ is given by the canonical representation $
\t{Tr}(X) = \sum_{j \in \{0,1\}^n} \bra{j} X \ket{j} $
and this defines the standard Hilbert-Schmidt trace inner product.

\begin{defi}[Hilbert-Schmidt Inner Product]
\begin{align}
\la i A, i B \ra_{\hs} := \t{Tr}\l( (i A)^\dg i B \r) = \t{Tr}(A^\dg B) .
\end{align}
\end{defi}

With this inner product, $\C^{2^n \times 2^n}$ is a Hilbert space (see \cref{part1:dos} for further discussion).

Let $V$ be a matrix space of skew-Hermitian operators. Then for a matrix subspace $W \subseteq V$ with an orthonormal basis $\B^W$ (with respect to Hilbert-Schmidt), any $i A \in V$ admits an orthogonal projection onto $W$.

Throughout, a superoperator is a real linear map between matrix spaces of skew-Hermitian operators. In particular, a superoperator on $\mf{u}\l( 2^n \r)$ maps this space to itself. We distinguish such superoperators from operators acting on the Hilbert space $\C^{2^n}$.

\begin{defi}[Projection Superoperator]\label{def:prj_op}
The projection superoperator of subspace $ W $ is a \textit{linear} map $\t{proj}_{W} : V \rightarrow V $ such that $\mathrm{Im}(\t{proj}_{W})=W$, defined by
\begin{align}
\t{proj}_{W}\l( i A \r) = \sum_{j=1}^{\l| \mathcal{B}^{W} \r|} \la i B_{j}^{W}, i A \ra_{\hs} \, i B_{j}^{W} .
\end{align}
\end{defi}

The projection superoperator is \textit{linear} such that $\t{proj}_{W}\l( i C + i D \r) = \t{proj}_{W}\l( i C \r) + \t{proj}_{W}\l( i D \r)$. For the Dynamic Observable Subspace $\orb$, the projection superoperator $\t{proj}_{\g}^{iO}$ isolates the component of an operator that lies within the DOS. The orthogonal complement of $W$ is
\begin{align}\label{eq:spaceperp}
W^\perp = \l\{ i A \in V : \la i A, i B \ra_{\hs} = 0, \; \forall \; i B \in W \r\},
\end{align}
and the complementary projection is given by
\begin{align}\label{eq:projperp}
\t{proj}_{W}^{\perp}(i A) = i A - \t{proj}_{W}(i A).
\end{align}

Then we define the inner product over the DOS, which is a subspace of $\C^{2^{n} \times 2^{n}}$.

\begin{defi}[Inner Product over Dynamic Observable Subspace]
Given a DOS $\t{orb}_{\mf{g}}^{iO}$ as \cref{def:dos} with an orthonormal basis $\B$, define the inner product over this matrix subspace of $\C^{2^{n} \times 2^{n}}$ as:
\begin{align}
\la i C, i D \ra_{\dos} &= \la \t{proj}_{\g}^{iO} \l( i C \r), \t{proj}_{\g}^{iO} \l( i D \r) \ra_{\hs} \\
&= \sum_{i B \in \B} \t{Tr}\l( \l( i B \r)^{\dg} i C \r) \t{Tr}\l( \l( i B \r)^{\dg} i D \r).
\end{align}
\end{defi}

In the Heisenberg picture, $i O(t)$ for any $t \in [0,LK] $ is inside the DOS.

\begin{thm}[DOS representation of $i O$]\label{thm1:dosobs}
\begin{align}
i O(t) \in \t{orb}_{\mf{g}}^{iO} .
\end{align}
\end{thm}

\noindent Then the following theorem captures that the DOS simulation perfectly matches the quantum dynamics.

\begin{thm}[Heisenberg Picture inside the DOS]\label{thm2:heidos}
Given an observable $i O$ and a set of generators $\G$ such that $\mf{g} = \langle i \G \rangle_{\t{Lie}} $, the Dynamic Observable Subspace $\t{orb}_{\mf{g}}^{iO} $ of $i O$ under $\mf{g}$ satisfies
\begin{align}
\t{Tr}\l( \rho(0) \, O(0) \r) = \la i \vr(0) , i O(0) \ra_{\dos} ,
\end{align}
where $i \vr(0) = \t{proj}_{\mf{g}}^{iO}\l( i \rho(0) \r)$ is the projection of $i \rho(0)$ into $\t{orb}_{\mf{g}}^{iO}$.
\end{thm}

We state a stronger theorem, which shows that the DOS is invariant under commutation for any skew-Hermitian operator in the DLA $\g$ as well as unitary evolution for unitaries inside the Lie group $e^\g$.

\begin{thm}[DOS Representation Theorem]\label{thm4:dosrep}
Let $i A \in \orb $. Then for any skew-Hermitian operator in the DLA $i H \in \g$, the commutator is inside the DOS:
\begin{align}
\l[ i H, i A \r] \in \orb .
\end{align}
For any unitary in the Lie group generated by the DLA $V \in e^{\g}$, elements of the DOS remain in the DOS under unitary evolution of $V$:
\begin{align}
V \, i A \, V^{\dg} \in \orb .
\end{align}
\end{thm}

\noindent \cref{thm4:dosrep} has far reaching implications. We define the DOS projection of $i \rho(t)$ for $t \in [0,LK]$ as $i \vr(t)$.
\begin{defi}[DOS representation of $i \rho$]\label{def:dosrho}
\begin{align}
i \vr(t) &= U_{1:t}(\bmt) \t{proj}_{\mf{g}}^{iO} \l( i \rho(0) \r) U_{1:t}^{\dg}(\bmt).
\end{align}
\end{defi}

\noindent Since the DOS $\orb$ is invariant under dynamics generated by $\g$ as stated in \cref{thm4:dosrep}, it follows that $i \vr(t) \in \orb$ since $i \vr(0) \in \orb$ by projection and $U_{1:t}(\bmt) \in e^\g$.

\begin{thm}[DOS representation of Density Operator]\label{thmX:dosrho}
For $t \in [0,LK]$:
\begin{align}
\t{proj}_{\mf{g}}^{iO} \l( U_{1:t}(\bmt) \, i \rho \, U_{1:t}^{\dg}(\bmt) \r) = U_{1:t}(\bmt) \, \t{proj}_{\mf{g}}^{iO} \l( i \rho \r) \, U_{1:t}^{\dg}(\bmt) .
\end{align}
\end{thm}

Then \cref{thm2:heidos} can be extended such that for time $t$, the Schr{\"o}dinger Heisenberg cut at that time is inside the DOS.
\begin{thm}[Heisenberg Cut at $t$ inside DOS]\label{thm4:heicut}
\begin{align}
\la i \rho(t), i O(t) \ra_{\hs} = \la i \vr(t), i O(t) \ra_{\dos} .
\end{align}
\end{thm}

An important consequence of \cref{thm4:dosrep} is that for $i G \in \g$, we can compute gradients inside the DOS. Next we discuss the operationalization inside DOS such that we find efficient computations for $i \vr(t)$ and $i O(t)$ for any $t$.

A fundamental theorem to characterize the DOS when the DLA has more restricted dynamics than arbitrary skew-Hermitian operators over $n$ qubits, for example when the DLA of the VQA is $\t{poly}\l( n \r)$~\cite{goh_lie-algebraic_2023,anschuetz2023efficient}, is that the dimension of the DOS is bounded based on the multiplicative sector of $i O$.

\begin{defi}[Multiplicative Sector of Order $k$]
Given $\g$ with matrix basis $\B^\g$, define the multiplicative sector of order $k$ of $\g$ as
\begin{align}
\Qc_\g^k &= \t{span}\l( \l\{ i \prod_{j=1}^{k} G_{j} : i G_{j} \in \g \r\} \r) \\
&= \t{span}\l( \l\{ i \prod_{j=1}^{k} B_{j} : i B_{j} \in \B^\g \r\} \r) ,
\end{align}
with $\Qc^0 = \t{span}\l( \l\{ i \I \r\} \r)$.
\end{defi}

\begin{thm}[Bounding the Dimension of DOS]\label{thm5:orb_dla_dim_outside}
Given DLA $ \g $, if the observable is inside the multiplicative sector $k$, $ i O \in \Qc_\g^k $, then
\begin{align}
\orb \subseteq \Qc_\g^k,
\end{align}
and so the dimension of the DOS is bounded by that of the multiplicative sector,
\begin{align}
\t{dim}\l( \orb \r) \leq \t{dim}\l( \g \r)^{k} .
\end{align}
\end{thm}

This will play an essential role in much of our numerical analysis later, when we consider VQAs where the circuit ans{\"a}tze have DLA with polynomial size dimension in the number of qubits. In the more limited case where the observable $i O$ is in the DLA, the dimension of the DOS is at most that of the DLA. This includes the DLA supported setting used in some of the $\g$-sim VQA applications~\cite{goh_lie-algebraic_2023}.

\begin{thm}\label{thm6:orb_dla_dim_inside}
Let $ \g $ be the DLA of gate set $ \G $. Given an observable $ i O \in \g $:
\begin{align}
\t{dim}\l( \orb \r) \leq \t{dim}\l( \g \r) .
\end{align}
\end{thm}

After we develop the tools for efficient representation and computation inside the DOS, we will consider several important quantum models where $\orb \subseteq \Qc_\g^k$ with small $k$, such that $\t{dim}\l( \orb \r) \leq \t{dim}\l( \g \r)^k$, which allows efficient simulation.

\subsection*{Efficient Simulation of Evolution in the Dynamic Observable Subspace}

As recognized in \cref{def:dos}, the DOS is a matrix vector space that satisfies all the associated conditions of such a space. Superoperators that preserve the DOS restrict to linear maps on $\orb$. In the basis $\B$, these are represented by real matrices acting on $\R^{\t{dim}\l( \orb \r)}$. We defined the projection superoperator earlier. Another central superoperator on this space is the adjoint superoperator.

\begin{defi}[Adjoint Superoperator]\label{def:adj_op}
Let $ V \subseteq \mf{u}\l( 2^n \r) $ be a matrix space and let $ M \in \mf{u}\l( 2^n \r) $ satisfy $[M,V] \subseteq V$. The adjoint superoperator of $M$ is the real linear map $\t{ad}_M : V \rightarrow V $:
\begin{align}
\t{ad}_{M}(X) = \l[ M, X \r] = M X - X M .
\end{align}
\end{defi}

\begin{figure*}[!t]
\centering
\includegraphics[width=0.7\textwidth]{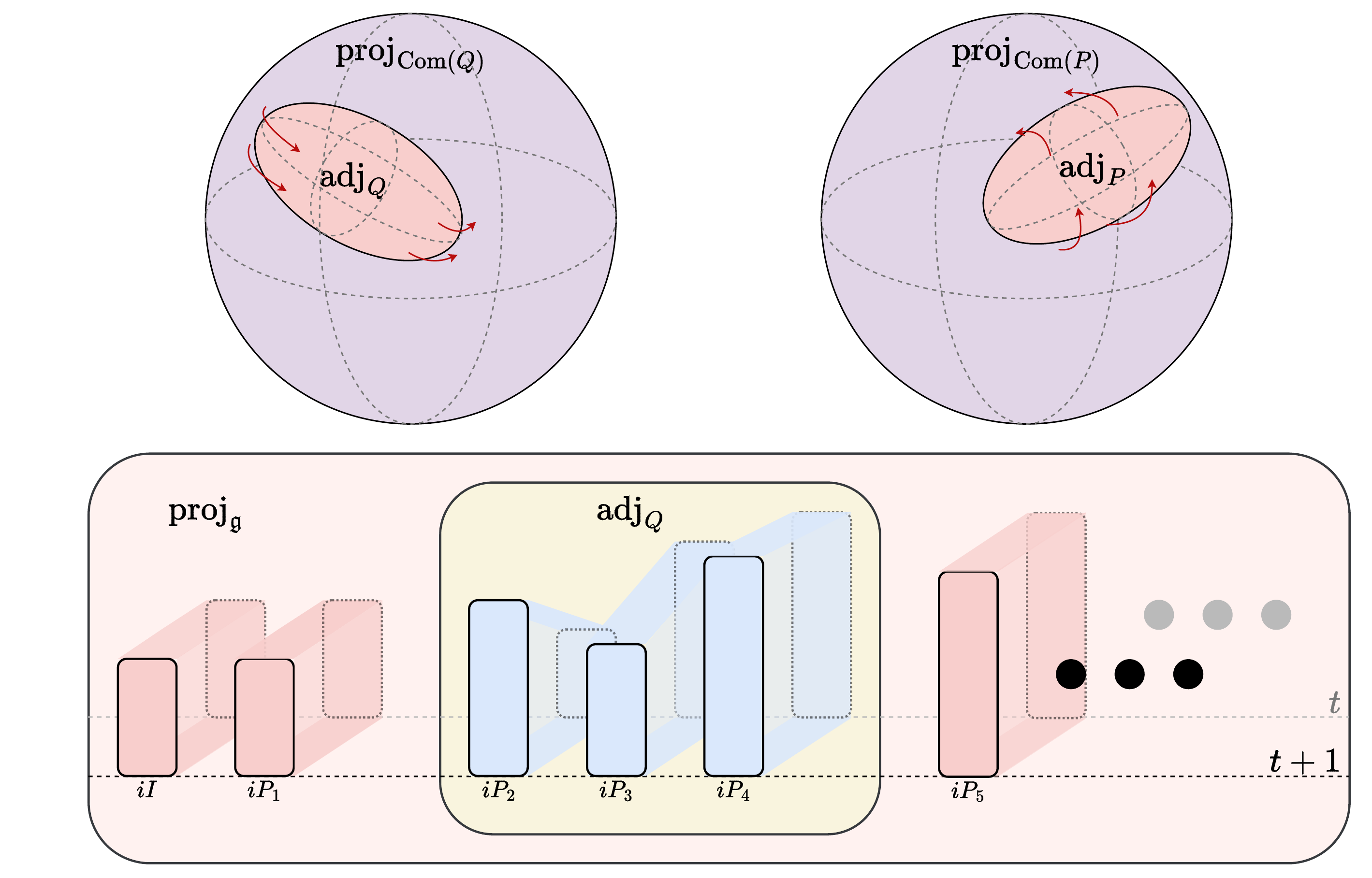}
 \caption{\tsbf{Dynamic Subspace of the Adjoint Superoperator.} Top depicts the dynamic action of two adjoint superoperators $\t{ad}_{iQ}$ and $\t{ad}_{iP}$ with $ \t{proj}_{\t{Com}(i Q)} $ and $ \t{proj}_{\t{Com}(i P)} $ as the outside space that is left stationary respectively. Bottom depicts $e^{i\theta Q}: \rho(t) \gt \rho(t+1) $ where support inside $ \t{proj}_{\t{Com}(Q)}^{\perp} $ colored blue is transformed while support outside (in $\t{proj}_{\t{Com}(i Q)}$) colored red is left stationary.}
\label{fig:img_adjop}
\end{figure*}

Given a construction of the DLA or DOS, we can utilize matrix exponentiation to the find the corresponding evolution in the DLA or DOS respectively. At the core is the correspondence between unitary evolution and the exponential of the adjoint from \cref{def:adj_op}. Under a unitary $U$, the evolution superoperator is given by $U \, i\vr \, U^{\dg}$.
\begin{defi}[Evolution Superoperator]\label{def:evo}
Let $iG$ be a skew-Hermitian operator, then the evolution superoperator under $iG$ is
\begin{align}\label{eq:evo}
\t{evo}_{iG}\l( i\vr; \theta \r) = e^{i \theta G} i \vr e^{-i \theta G} .
\end{align}
\end{defi}

Then it can be shown~\cite{somma2005quantum,goh_lie-algebraic_2023,hall2013lie} (see \methods) that this superoperator is the exponential of the adjoint superoperator, known as Hadamard's Lemma.
\begin{thm}[Hadamard's Lemma]\label{thm7:evoexpadj}
\begin{align}
\t{evo}_{iG}(i\vr; \theta) = e^{\theta \, \t{ad}_{i G} } \l( i \vr \r) .
\end{align}
\end{thm}

It then follows that through exponentiation, we can simulate $ \t{evo}_{iG} $ in $\Oc\h\l( \t{dim}\l( \orb \r)^{2} \r) $. In practice, we often consider gates that are either Pauli or have limited eigenvalues for their eigenspectrums, such as local Hamiltonians. In this case, we can find efficient representations of $\t{evo}_{iG}$ that allow us to reach $\Oc\l( \t{dim}\l( \t{orb}_{\mf{g}} \r) \r) $ runtime scaling. Moreover, a gate will not impact basis terms that it commutes with. As such, the gate may have an even more compact representation, such as we show in the following subsections.

\subsection*{Efficient Representations of Pauli, Projector, and Local Generator Evolution inside DOS}

Pauli operators, or Pauli words, are widely used in VQAs and other studies of quantum dynamics, we discuss the background further in \methods. The phaseless set of all $n$-qubit Pauli strings, $\mathcal{P}$, forms an orthogonal basis for the real vector space of traceless Hermitian operators on $n$ qubits. Evolution under a Pauli generator $ P $ is given by
\begin{align}
\t{evo}_{iP}(i \vr; \theta_t) = e^{i \theta P} \, i \vr \, e^{-i \theta P} .
\end{align}

We show that the evolution inside the DOS is summarized by the following theorem.
\begin{thm}[Pauli Evolution inside the DOS]\label{thm8:paulievo}
Given a Pauli string $P \in \P$, the evolution superoperator as \cref{def:dosrho} is
\begin{align}\label{eq:paulievo}
\t{evo}_{iP}(i \vr; \theta_t) = i \vr &+ (\cos(2\theta)-1) \, \t{proj}_{\t{Com}(i P)}^{\perp}(i \vr) \nonumber \\
&+ \sin(2\theta) \, \t{ad}_{iP}(i \vr) / 2,
\end{align}
with $\t{Com}\l( i P \r) = \t{span}\l( iB : [P, B ] = 0, B \in \B \r) \subseteq \t{orb}_{\mf{g}}^{iO} $ and so $\t{proj}_{\t{Com}(i P)}^{\perp}$ projects onto the active noncommuting subspace inside the DOS.
\end{thm}

In particular, if the orthonormal basis for $\t{orb}_{\mf{g}}^{iO}$ is a subset of the Pauli group $\B \subseteq \P $ such as if the generators and observables are Pauli strings $\G, O \in \P$ since nest commutation of Pauli strings generates Pauli strings, then so is the basis for $\t{Com}(iP)$. In \methods, we discuss fast basis construction and sparse matrix construction for $\t{ad}_{iP}$ and $\t{proj}_{\t{Com}(iP)}^{\perp}$ that allows fast simulation of \cref{eq:paulievo}.

Another important class of generators are those of diffusor mixers~\cite{leipold_constructing_2021,leipold2026imposing,hadfield_analytical_2022} or projectors. A Hermitian operator $F$ satisfying $F = F^2$ is a projector. Its spectrum is contained in $\{0,1\}$, which decomposes $\C^{2^{n}}$ into the corresponding eigenspaces. Evolution through a projector generator is given by
\begin{align}\label{eq:projevo}
\t{evo}_{iF}(i \vr; \theta_t) = e^{i \theta F} \, i \vr \, e^{-i \theta F} .
\end{align}

The evolution is given by trigonometric couplings of the adjoint superoperator and its square:
\begin{thm}[Projector Evolution inside the DOS]\label{thm9:projevo}
\begin{align}
\t{evo}_{iF}(i \vr; \theta_t) &= i\vr(t) + \sin(\theta_t) \, \ad{iF}{i\vr(t)} \nonumber\\
&\peq + (1-\cos(\theta_t)) \ad{iF}{\ad{iF}{i\vr(t)}}
\end{align}
\end{thm}

VQAs generally use local generators, that is, each $G_k \in \G$ acts nontrivially over only a fixed number of qubit indices $ S_k = \{ s_1, \ldots, s_k \} \subseteq [n] $ with $|S_k| = \Oc\h\l( 1 \r)$. Then $ G_{k} $ can be decomposed over $S_k$ as $ \sum_{j} n_{jk} F_{kj} $ with the nonzero entries bounded by $4^{|S_k|}$. Moreover $e^{i \theta (F + Q)} = e^{i \theta F} e^{i \theta Q}$ for any two orthogonal projectors $F$ and $Q$. This means we can apply each  projector sequentially to apply the general Hamiltonian generator.

\begin{thm}[Eigendecomposed Generator Evolution inside the DOS]\label{thm10:eigevo}
Given a Hamiltonian $ H $ and its eigendecomposition $ H = \sum_{j=1}^{J} n_{j}  F_{j} $ with $ F_{j} F_{k} = \delta_{jk} F_j $,
\begin{align}
\t{evo}_{i H}\l( i \vr; \theta \r) &= e^{i \theta H} \, i\vr \, e^{-i \theta H} \nonumber \\
&= \t{evo}_{iF_J} \l( \ldots \t{evo}_{iF_1}\l( \, i \vr ; \theta n_1 \r) \ldots ; \theta n_J \r)
\end{align}
\end{thm}

By decomposing an evolution superoperator as the nonlinearly coupled weighted sum of superoperators over the matrix vector space, we can reduce the computation cost to that of the adjoint superoperator and the projection superoperator. For a Pauli operator acting on a matrix space spanned by skew Pauli words, we can recognize that the adjoint superoperator will map a skew Pauli string to another skew Pauli string, leading to sparse updates of the DOS representation.
\begin{thm}[Linear Computation inside the DOS]\label{thm12:lindens}
Let $iG$ be either a Pauli string $P$ with $\dens\l( \t{ad}_{iP} \r) \leq \t{dim}\l( \orb \r)$ or have decomposition $\sum_{j}^{J} n_{j} F_{j} $ such that $\dens\l( \t{ad}_{i F_{j}} \r) \leq \t{dim}\l( \t{orb}_{\g}^{iO} \r)$ for $J \in \Oc\hm\l( 1 \r)$. Then $\t{evo}_{iG}(i \vr, \theta_t)$ can be computed in $\t{dim}\l( \t{orb}_{\g}^{iO} \r)$.
\end{thm}

By applying $t$ such evolution superoperators on $i \vr(0)$, we can compute $i \vr(t)$.
\begin{cor}[]
$i\vr(t) $ can be computed in $ \Oc(n \, t \, \t{dim}(\orb)) $ given $i\vr(0)$ if conditions of \cref{thm12:lindens} are met by every $G \in \G$.
\end{cor}

\subsection*{Multiple Orbit Decomposition of Dynamics}

\begin{figure*}[!t]
\centering
\includegraphics[width=0.6\textwidth]{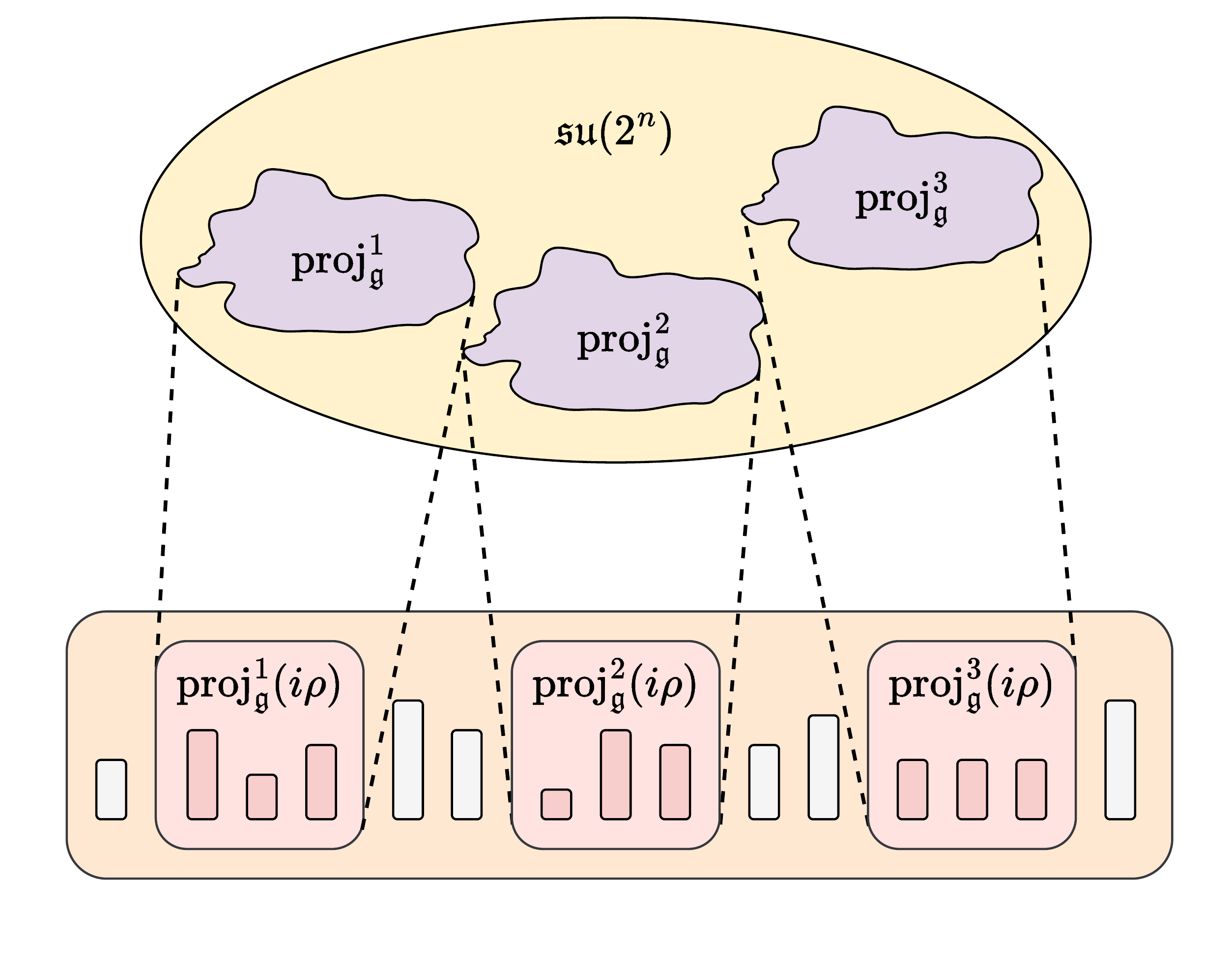}
 \caption{\tsbf{Decomposition into Multiple Dynamic Observable Subspaces.} Given an observable $O = v_1 \, O_1 + v_2 \, O_2 + v_3 \, O_3$, we can define the DOS for each $O_j$ and simulate them separately.}
\label{fig:img_multi_dyn_sub}
\end{figure*}

In many cases, we are interested in a observable that admits this decomposition $H = \sum_{O \in D(H)} v_j \, O_j $ for some set $D(H)$. For example, we will later consider the case that the observable is decomposed into a set of Pauli words, $H = \sum_{O_j \in \P(H)} v_j O_j $ with a subset of Pauli words $\P(H) = \{ P : \t{Tr}\l( P H \r) \neq 0, P \in \P \}$. For example, a Pauli $Z$ Ising spin Hamiltonian decomposable over single Pauli $Z$ and two qubit Pauli $Z$ operators $H = \sum_{j} h_{j} Z_{j} + \sum_{jk} Z_{j} Z_{k}$ with $\P\l( H \r) = \{ Z_j \}_{j=1}^{n} \cup \{ Z_{j} Z_{k} \}_{j<k} $. Then in particular, we can find a $\orbj$ for each $O_{j} \in D(H)$, which may be dramatically smaller than if utilizing $H$ itself. This happens specifically in the case we will study later.

By linearity of the trace, we can decompose expectation values over multiple dynamic observable subspaces.

\begin{thm}[Multiple DOS Representation]\label{thm12:multidosrep}
Given $H = \sum_{O_j \in D(H)} v_j O_j$ and $\orbj$ for each $O_j$, let $\la \cdot, \cdot \ra_{j}$ be the projected trace inner product for orbital $j$ and let $i \vr_{j}(t)$ be equal to $ U_{1:t}(\bmt) \t{proj}_{\g}^{iO_{j}}\l( \rho(0) \r) U_{1:t}^{\dg}(\bmt)$ from \cref{def:prj_op} for orbital $j$. Then for any $V \in e^{\g}$:
\begin{align}
\t{Tr}\l( V \rho(t) V^{\dg} \, H(t) \r) = \sum_{j} v_{j} \, \la \l( V \, i \vr_j(t) \, V^{\dg} \r), i O_j(t) \ra_{j} .
\end{align}
\end{thm}

This means that circuit execution and gradients can be calculated over each DOS as expected. It is possible that two observables $O_j, O_k$ have equivalent subspaces, which occurs if $O_k \in \orbj$. In this case, we only need $\orbj$ and find the representation of $O_k$ inside the matrix subspace.

\subsection*{Dynamic Observable Subspace Adjoint Method for Gradients}

Gradients are a fundamental quantity to calculate for any parameterized system, as they point in the direction of steepiest descent for parameters $\bmt$ inside the loss landscape given by $\L$. We show in \methods that for gate $G_{t}$ at time $t$, the partial derivative can be computed by the trace of the skew Schr{\"o}dinger density operator, the adjoint of gate $G_{t}$, and the skew Heisenberg observable
\begin{align}
\frac{\partial \mathcal{L}}{\partial \theta_t} &= \la i \vr(t) \, \t{ad}_{iG} , i O(t) \ra_\dos \\
&= -\la i \vr(t), \t{ad}_{iG} \, i O(t) \ra_\dos ,
\end{align}
which is also depicted in \cref{fig:schrohei}. The gradient is then defined as $\nabla_{\bmt} \, \mathcal{L} = (\frac{\partial \mathcal{L}}{\partial \theta_1}, \ldots, \frac{\partial \mathcal{L}}{\partial \theta_{LK}}) $.

To compute gradients for VQAs inside the DOS, we provide the DOS Adjoint Method (DOSAM) that is analogous to the self adjoint method for the whole space~\cite{jones2020efficient} and the gradient algorithm of $\g$-sim~\cite{goh_lie-algebraic_2023}. The gradient is computed in reverse order, we evolve until the end to define $i \vr(LK) $ and then evolve $i \vr(t)$ and $i O(t)$ in reverse time while computing the partial derivative as we go. As such, we compute the $LK$ length gradient in only two passes of the circuit execution within the observable generated DOS using the sparse gate actions developed.

\begin{thm}\label{thm13:gradinobs}
There exists an algorithm to compute $ \nabla_{\bmt} \, \L $ in $ \Oc\h\l( L K \t{dim}\l( \orb \r) \r) $.
\end{thm}

In the case that we have a decomposed observable $H = \sum_{O_j \in D(H)} v_j O_j$, by \cref{thm12:multidosrep} we can compute gradients over each DOS.

\begin{align}\label{eq:grad_multidos}
\nabla_{\bmt} \, \L = \sum_{O_j} \l( \frac{\partial \L_{\mathfrak{g}}^{iO_j}}{\partial \theta_1}, \ldots, \frac{\partial \L_{\mathfrak{g}}^{iO_j}}{\partial \theta_T} \r)
\end{align}

Raw gradient updates can then be computed via $\bmt^{\nu+1} = \bmt^{\nu} - \beta \, \nabla_{\bmt} \L $ with a sufficiently small or adaptive $\beta$. We use Adaptive Moment Estimation (ADAM)~\cite{intro_kingma_adam_2014} in our calculations.

Highly expressive circuits can exhibit barren plateaus with exponentially small gradient variance, as shown in Ref.~\cite{mcclean2018barren}. We show that circuits with polynomial-dimensional DLAs avoid this exponential suppression for observables of bounded multiplicative order under the mixing, irreducibility, overlap, and gate-action conditions stated below.

To compare the variances of $\partial\L/\partial\theta_t$ and $\L$ over random initial angles, we cut the circuit at gate $t$ as in Ref.~\cite{mcclean2018barren}, evolving the state forward through gates $1,\ldots,t$ and the observable backward through gates $T,\ldots,t+1$ (\cref{fig:schrohei}). We use Haar averaging over the dynamical group as in Ref.~\cite{ragone2024lie} and assume that the two sides are sampled independently and each reproduces the first two Haar moments of conjugation on $W$, the part of the DOS orthogonal to operators fixed by every gate.

\newcommand{\DOSVarianceStatement}{
Assume the group $e^\g$ generated by the gates is compact. Let $W$ be the orthogonal complement in $\orb$ of the operators fixed under conjugation by every $U \in e^\g$. Suppose $W$ is nonzero and irreducible. Write $d = \t{dim}\l( W \r)$ and let $\t{proj}_{W}$ be the orthogonal projector onto $W$. For fixed $G_t$, suppose the prefix $U_{1:t}$ and inverse suffix $U_{t+1:T}^{\dg}$ are independent and each matches the first two Haar moments of the conjugation action on $W$. Set
\begin{align}
q_\rho &= \l\| \t{proj}_{W}\l( i \rho(0) \r) \r\|_{\hs}^2, \nonumber \\
q_O &= \l\| \t{proj}_{W}\l( iO \r) \r\|_{\hs}^2, \\
\kappa_t &= \l\| \t{proj}_{W} \, \t{ad}_{iG_t} \, \t{proj}_{W} \r\|_{\t{F}}^2/d.
\end{align}
Then
\begin{align}
\t{Var}\l( \L \r)
&= \frac{q_\rho \, q_O}{d}, \\
\t{Var}\l( \frac{\partial \L}{\partial \theta_t} \r)
&= \frac{q_\rho \, q_O \, \kappa_t}{d}.
\end{align}
If $iO \in \Qc_\g^k$, then $d \leq \t{dim}\l( \orb \r) \leq \t{dim}\l( \g \r)^k$ by $W \subseteq \orb$ and \cref{thm5:orb_dla_dim_outside}. Hence
\begin{align}
\t{Var}\l( \frac{\partial \L}{\partial \theta_t} \r) \ge \frac{q_\rho \, q_O \, \kappa_t}{\t{dim}\l( \g \r)^k}.
\end{align}
}

\begin{figure*}[!t]
\centering
\includegraphics[width=0.65\textwidth]{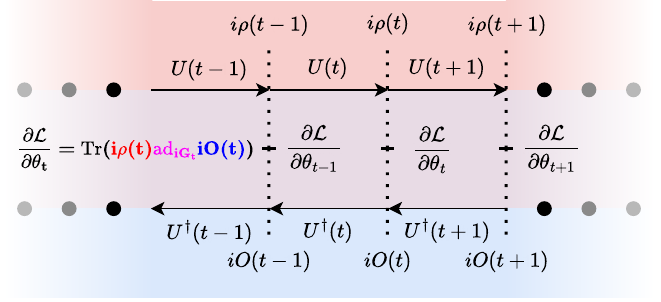}
 \caption{\tsbf{Meeting of Schr{\"o}dinger $ i\rho(t) $ and Heisenberg $ iO(t) $ at time $ t $.} In the mixed Schr{\"o}dinger Heisenberg picture, $ i\rho(t) $ and $ iO(t) $ meet at time $ t $ through the action of generator $ G_t $, allowing us to compute each entry $  \frac{\partial}{\partial \theta_t} \mathcal{L} = \la i \vr(t) \, G_t , i O(t) \ra $ in a single pass of (reverse) evolving $ i \vr $ and $ i O $.}
\label{fig:schrohei}
\end{figure*}

\begin{cor}[Loss and gradient variance in one irreducible sector]\label{cor:grad_var_bound}
\DOSVarianceStatement
\end{cor}

If $ iO \in W $, then by the Cauchy-Schwarz inequality:
\begin{align}
q_\rho \, q_O
&\geq \l| \la \t{proj}_{W}\l( i \rho(0) \r), iO \ra_{\hs} \r|^2 \\
&= \l| \t{Tr}\l( \rho(0) \, O \r) \r|^2 .
\end{align}
For fixed $ k $ and polynomial $ \t{dim}\l( \g \r) $, if both $ \l| \t{Tr}\l( \rho(0) \, O \r) \r| $ and $ \kappa_t $ are at least inverse polynomial, then the gradient variance is at least inverse polynomial.

\subsection*{Pauli Simulation and Training for VQAs in the DLA and DOS}

\newcommand{\esteps}{20}
\newcommand{\nsteps}{200}
\newcommand{\vartol}{10^{-6}}

\renewcommand{\vem}{1.0em}

\begin{figure*}[!t]
\centering

\begin{tabular}{c c}
\begin{subfigure}[t]{0.48\textwidth}
\includegraphics[width=1.0\textwidth]{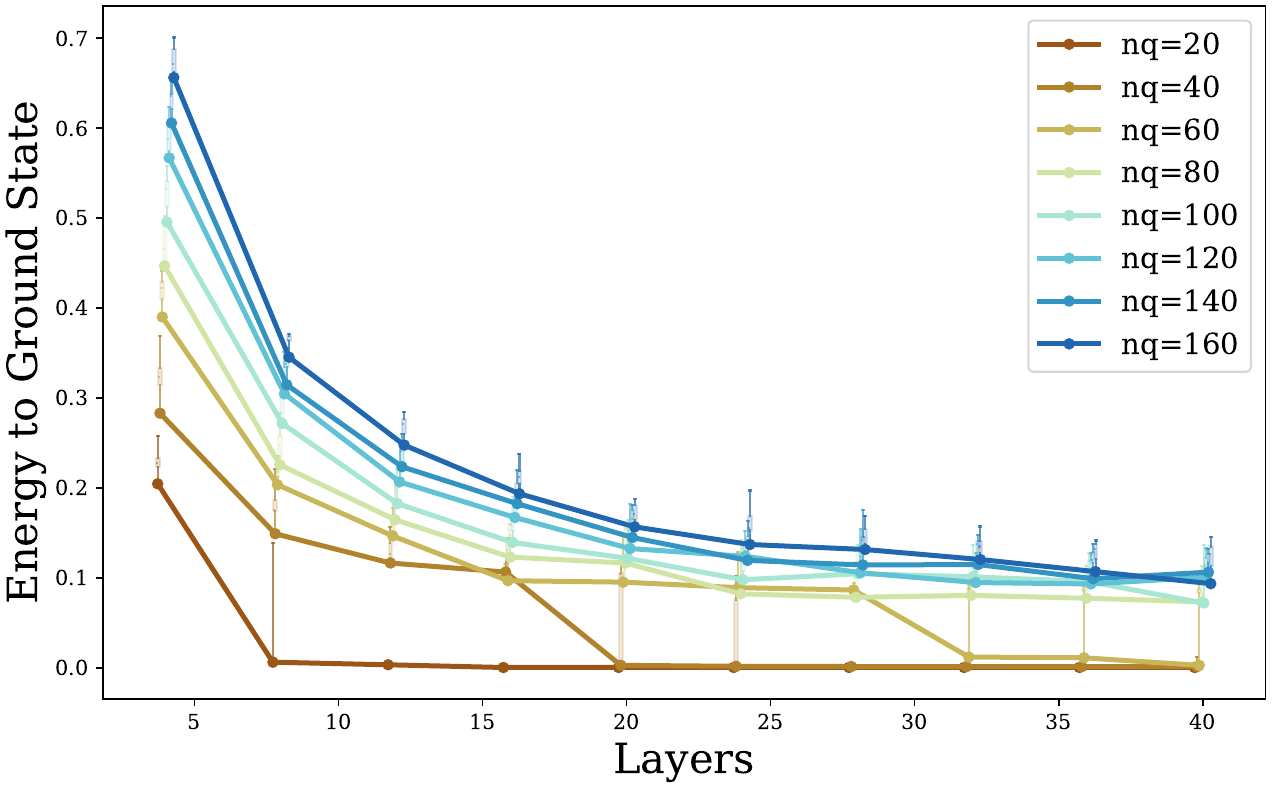}
 \caption{Anisotropic energy gap}
\label{fig:tfxy_b}
\end{subfigure}
&
\begin{subfigure}[t]{0.48\textwidth}
\includegraphics[width=1.0\textwidth]{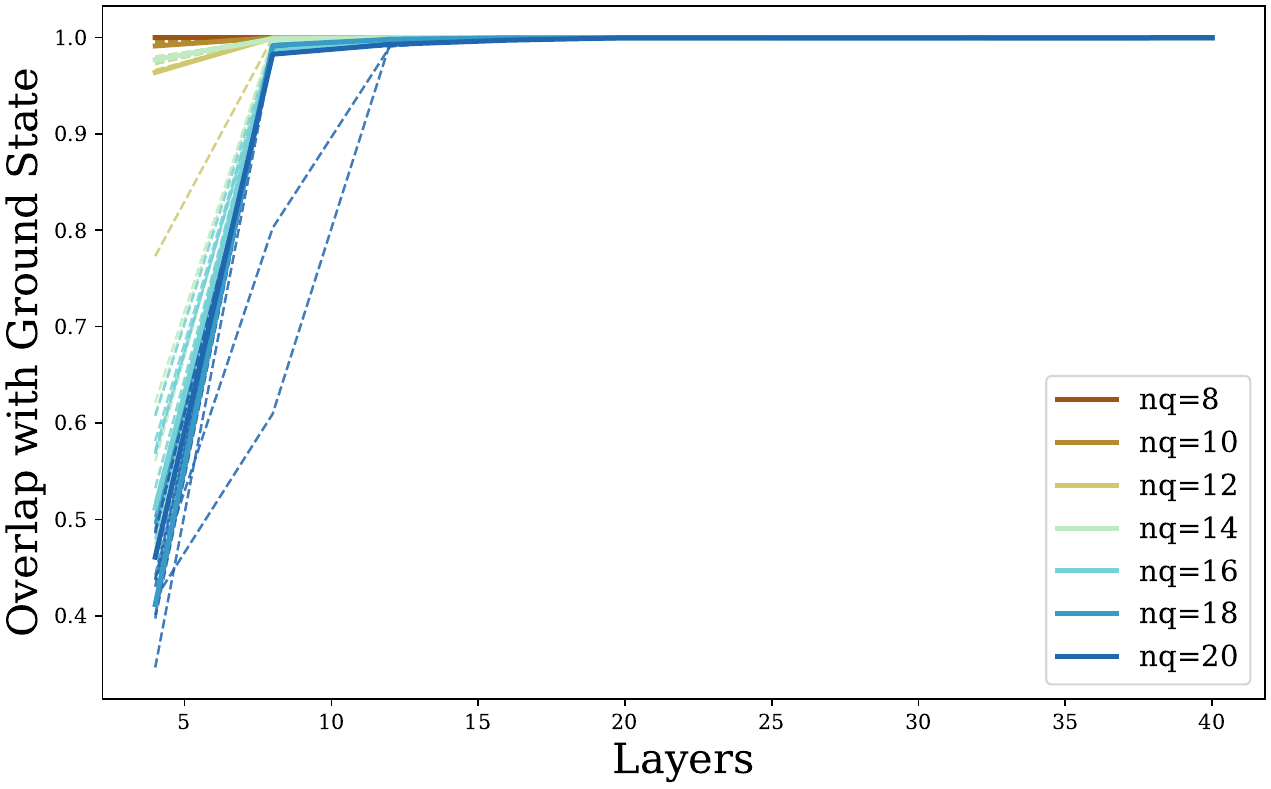}
 \caption{Anisotropic overlap with ground state}
\label{fig:tfxy_c}
\end{subfigure}
\end{tabular}
 \caption{\tsbf{TFXY Anisotropic Robustness.} The left panel shows the anisotropic deformation with couplings $\eta=0.6$, $\nu=0.4$, plotted as energy above the ground state versus depth. The right panel shows the corresponding ground state overlap in the small system regime where direct overlap evaluation is feasible.}
\end{figure*}

\begin{figure*}[!t]
\centering
\begin{tabular}{cc}
\begin{minipage}[t]{0.48\textwidth}
\centering
\includegraphics[width=\textwidth]{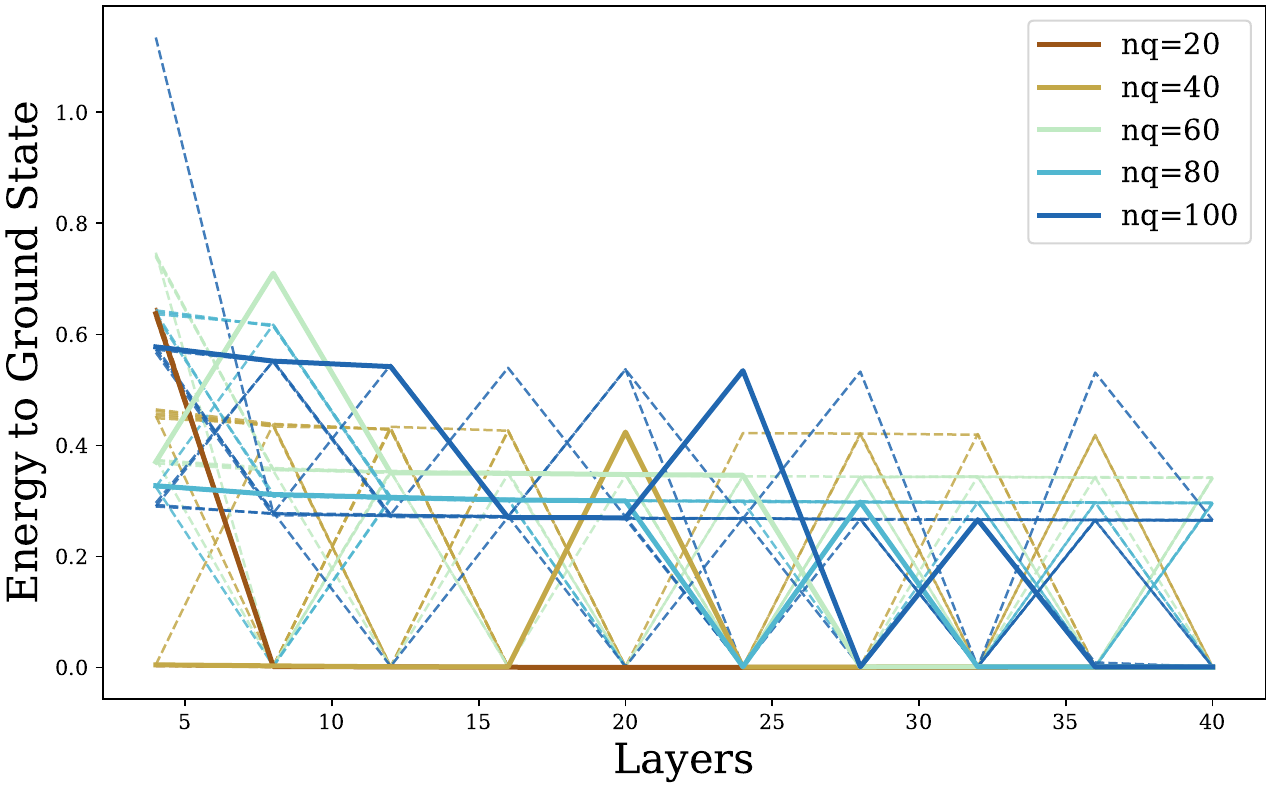}
 \captionof{figure}{\tsbf{TFXY Staggered Field Robustness.} Anisotropic staggered field target with $D=1.0$, using the ansatz $U$. The plotted quantity is the energy above the ground state versus depth.}
\label{fig:tfxy_stag_d1_compare}
\end{minipage}
&
\begin{minipage}[t]{0.48\textwidth}
\centering
\includegraphics[width=\textwidth]{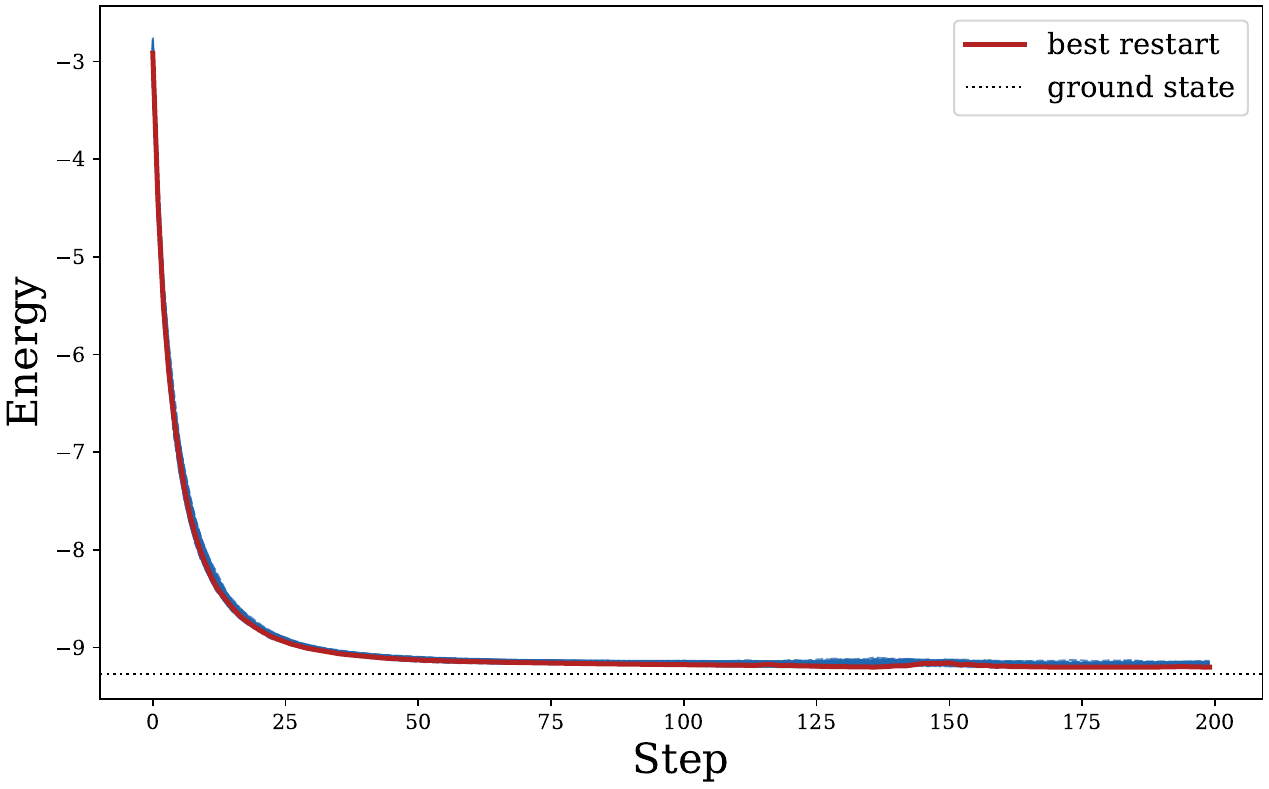}
 \captionof{figure}{\tsbf{TFXY Optimization Traces.} Clean anisotropic TFXY target at $n=100$ and depth $L=40$. We show the optimization energy versus training step for all $100$ restarts, with the best final restart highlighted as a solid line.}
\label{fig:tfxy_aniso_opt_traces}
\end{minipage}
\end{tabular}
\end{figure*}

We now consider the utilization of the DOS for simulating several important quantum circuits with observables associated with important physical quantum systems of interest. See \methods and associated supplementary for development of the standard definitions of common quantum gates and their properties. For all our experiments, we train with gradient updates modified by an ADAM state~\cite{intro_kingma_adam_2014}.

We begin with the TFXY family as a clean Pauli benchmark in which the target Hamiltonian lies inside a DLA of polynomial dimension. In this setting, expressibility mismatch is absent, so the numerical behavior primarily reflects trainability, parameterization, and depth scaling. We consider the generator set $\G_{TFXY}$ together with the ansatz $U$ defined in \methods. We also study the shared parameter circuit $U_S$, which assigns one parameter to each layer. Its relation to Hamiltonian variational ans\"atze and QAOA is described in \cref{sec:tfxy_shared_app,ssec:grad_shared}.

This generator set has a DLA dimension that grows quadratically in the number of qubits.
\begin{thm}
$ \t{dim}\l( \mf{g}_{XX,YY,Z} \r) = \Theta\l( n^2 \r) $
\end{thm}

Take the open chain generated by $X_j \, X_{j+1}$, $Y_j \, Y_{j+1}$, and $Z_j$. For $n > 4$, the sectors generated by $Z_j$ and by $Z_j \, Z_k$ with $j \ne k$ are irreducible~\cite{diaz2023beyond}.

\begin{cor}[Gradient variance for the open chain]\label{cor:tfxy_grad_var}
Initialize the open chain in $\ket{0}^{\otimes n}$. For any individual $XX$, $YY$, or $Z$ gate and under the sampling assumptions of \cref{cor:grad_var_bound},
\begin{align}
\t{Var}\l( \frac{\partial \L}{\partial \theta_t} \r) \in \begin{cases} \Omega(n^{-2}), & O = Z_j, \\
\Omega(n^{-3}), & O = Z_j \, Z_k . \end{cases}
\end{align}
\end{cor}
The calculation in \cref{ssec:tfxy_grad_var} uses a single Pauli observable and the two independent circuit pieces stipulated above.

We study an anisotropic TFXY Hamiltonian with unequal  $X$ and $Y$ spin couplings between nearest neighbors. Its terms are generated by $\g_{TFXY}$, so the target remains inside the same DLA of polynomial dimension and DOS. This gives a clean Pauli test case away from the symmetric $XX+YY$ point while remaining inside the efficiently simulable DLA setting.

\cref{fig:tfxy_b,fig:tfxy_c} report the energy above the ground state and the corresponding ground state overlap for the anisotropic model. In the energy plot, the main line shows the best final restart at each depth and the box and whisker overlay summarizes variability across restarts. The overlap plot reports the small system regime where direct overlap evaluation is feasible. Performance remains strong away from the symmetric point, but the required depth still grows noticeably with system size.

\cref{fig:tfxy_stag_d1_compare} shows a second robustness check in which the on site field is staggered with amplitude $D=1.0$. This model still lies within the same DLA and DOS, and the variational trend is similar. \cref{fig:tfxy_aniso_opt_traces} shows broadly similar convergence across the restarts, with a modest spread in the final energies and one best restart finishing lowest.

\Cref{fig:tfxy_indv_share} compares the individual and shared parameterizations on this staggered field target. The solid individual parameter curves stay near the ground state. The dashed shared parameter curves remain substantially higher for every displayed depth and system size. Parameter sharing reduces the number of variables, but it gives worse preparation quality in this experiment.

A richer class of physically relevant cases arises when the observables lie outside the DLA, so that the DOS extends into a subset of the multiplicative sector as encapsulated in \cref{thm5:orb_dla_dim_outside}. In this setting, circuits with a restricted DLA and observables in a low multiplicative sector can dramatically outperform fully expressive circuits because they remain trainable.

Hamiltonian $H_{SK,XX,YY}$ includes disordered couplings of $Z$ spins together with the $XX$ and $YY$ terms. Hamiltonian $H_{SKZ,XX,YY}$ further includes a disordered $Z$ field. Their explicit forms are given in \methods. Neither Hamiltonian lies in the DLA generated by the restricted ansatz $\G_{TFXY}$, whereas both lie in the DLA generated by the enlarged ansatz $\G_{XX,YY,Z,ZZ}$, which includes $Z$ and $ZZ$ generators. We test whether this added expressibility helps to prepare lower energy states, or whether it instead degrades training due to the onset of a barren plateau and a more rugged optimization landscape. For $H_{SKZ,XX,YY}$, we split the observable into two DOS orbitals and sum the gradient contributions via \cref{eq:grad_multidos}; see \methods for details.

\Cref{fig:skzz_res_gap,fig:skzz_full_gap} compare the restricted and full ans{\"a}tze for $H_{SK,XX,YY}$, while \cref{fig:skzzz_res_gap,fig:skzzz_full_gap} show the same comparison for $H_{SKZ,XX,YY}$. Even though the full ansatz can represent the target Hamiltonian exactly, its optimization quality degrades sharply with system size and depth, whereas the restricted ansatz continues to prepare substantially lower energy states. This separation between the restricted and full ansatz shows that limiting expressive power can improve practical performance because the reduced manifold provides more useful optimization signals and avoids optimization collapse associated with the fully expressive circuit.

\renewcommand{\vem}{1.0em}

\begin{figure*}
\begin{tabular}{c c}
\begin{subfigure}[t]{0.48\textwidth}
\includegraphics[width=1.0\textwidth]{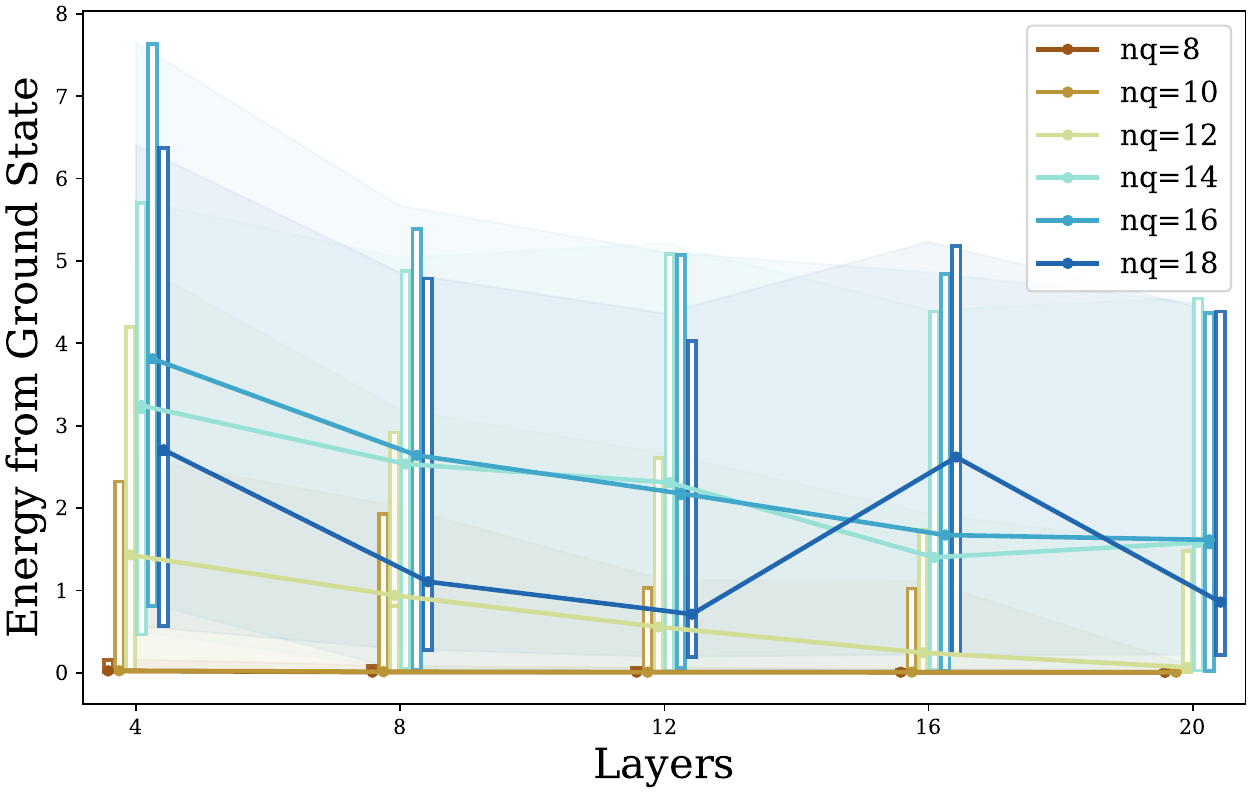}
 \caption{Full Ansatz}
\label{fig:skzz_full_gap}
\end{subfigure}
&
\begin{subfigure}[t]{0.48\textwidth}
\includegraphics[width=1.0\textwidth]{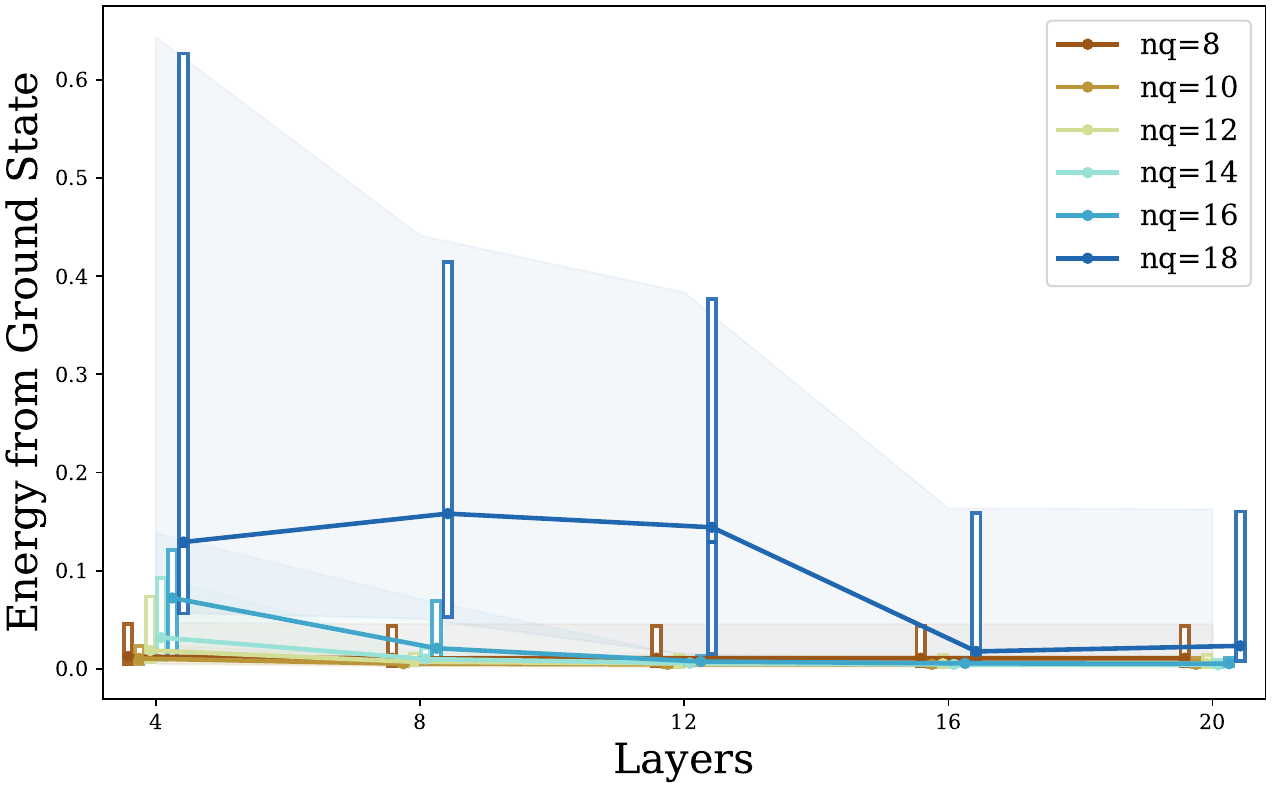}
 \caption{Restricted Ansatz}
\label{fig:skzz_res_gap}
\end{subfigure}
\end{tabular}
\caption{\tsbf{Energy Gap from Ground for XY with disordered ZZ.} (A) and (B) depict the energy above the ground state energy for the full and restricted ansatz respectively. The solid lines are the median while ribbons capture a quarter of instances above and below the median. Note that the two panels use different vertical axis scales because the restricted ansatz reaches substantially lower energies.}
\vspace{0.25em}

\begin{tabular}{c c}
\begin{subfigure}[t]{0.48\textwidth}
\includegraphics[width=1.0\textwidth]{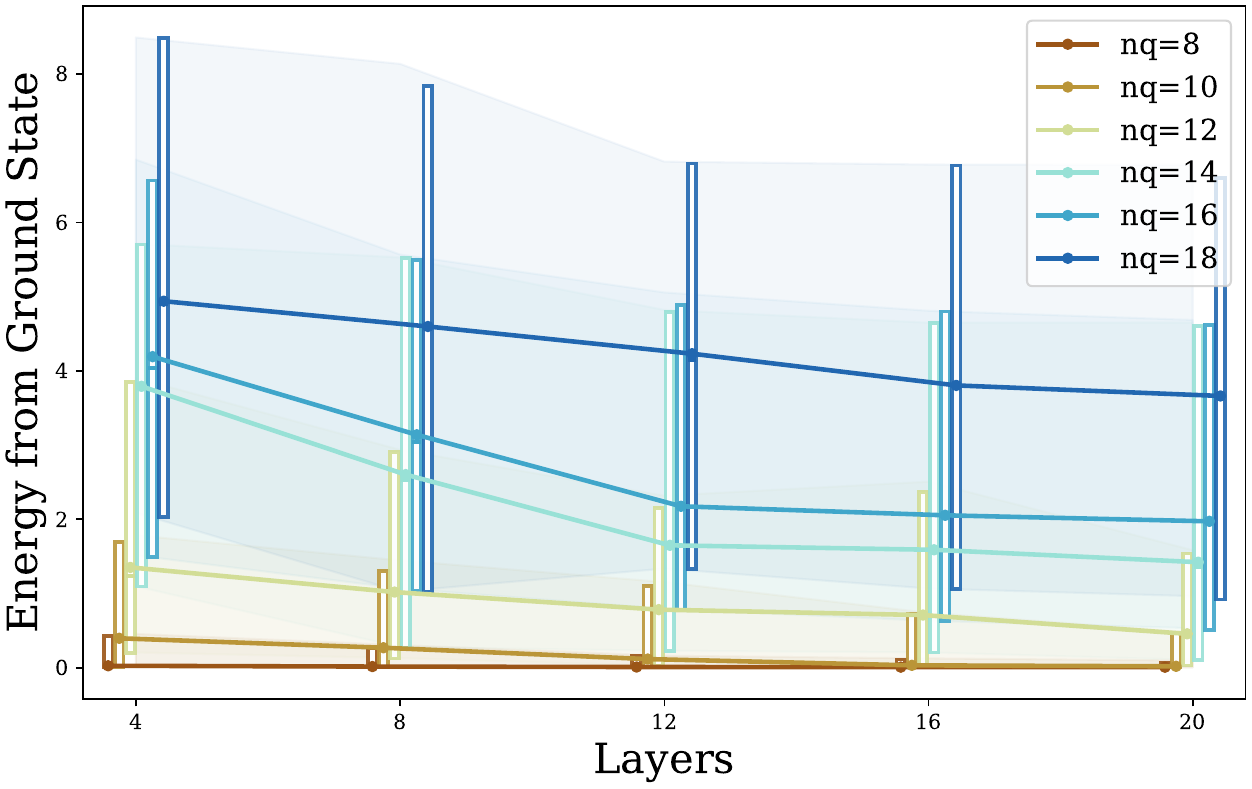}
 \caption{Full Ansatz}
\label{fig:skzzz_full_gap}
\end{subfigure}
&
\begin{subfigure}[t]{0.48\textwidth}
\includegraphics[width=1.0\textwidth]{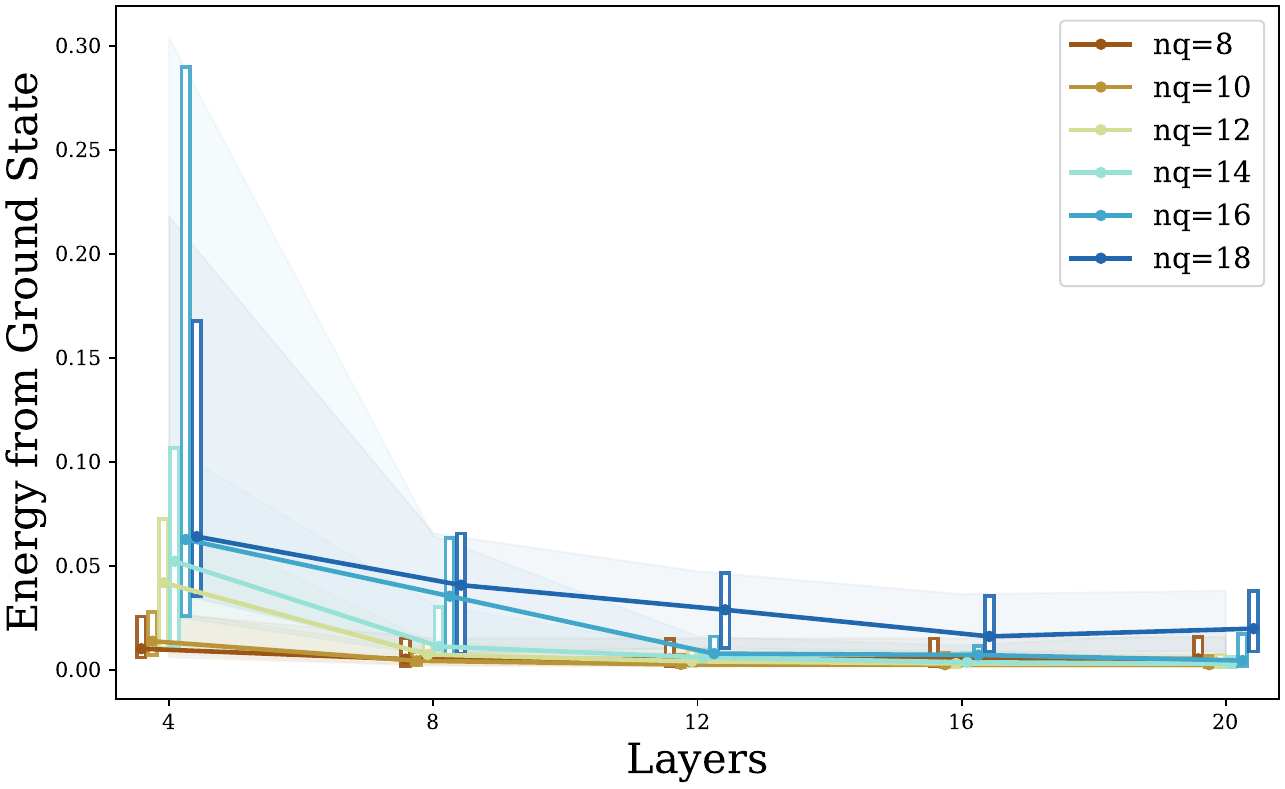}
 \caption{Restricted Ansatz}
\label{fig:skzzz_res_gap}
\end{subfigure}
\end{tabular}
 \caption{\tsbf{Energy Gap from Ground for XY with disordered ZZ and Z.} (A) and (B) depict the energy above the ground state energy for the full and restricted ansatz respectively. The solid lines are the median while ribbons capture a quarter of instances above and below the median. Note that the two panels use different vertical axis scales because the restricted ansatz reaches substantially lower energies.}

\end{figure*}

\subsection*{Eigendecomposed Generator Simulation and Training for VQAs in DOS}\label{subsec:xy}

Many physically relevant generators are sums of Pauli words rather than single Pauli words. This includes constraint preserving mixers for classical feasible sets, which cannot in general be implemented by individual Pauli words alone~\cite{hen_driver_2016,hadfield_quantum_2017,leipold_constructing_2021,leipold2026imposing}. We focus on the XY mixer~\cite{wang_x_2020,hadfield_analytical_2022,brandhofer2022benchmarking,herman_portfolio_2022,he2023alignment,intro_awasthi_xy_2026}, whose definition and generator set are given in \methods. It is the unique two body constraint preserving mixer~\cite{leipold_constructing_2021}.

In the case that the topology of the XY mixer is a path or ring, the dimension of the DLA is quadratic~\cite{kordonowy2025lie}.

The generator set $\G_{XY,Z}$ and its DLA structure are defined in \methods.
\begin{thm}
$ \t{dim}\l( \mf{g}_{XY,Z} \r) = \Theta\l( n^2 \r) $ .
\end{thm}

The XY gate decomposes into two rank one projector steps $P_+$ and $P_-$ associated with the Bell states $\ket{B_3}$ and $\ket{B_4}$; see \methods. This enables exact DOS simulation through \cref{thm10:eigevo}.

The XY ring Hamiltonian $H_{XY}$ is defined in \methods. Since $H_{XY} \in \g_{XY,Z}$, we can prepare its ground state with a PQC of $\G_{XY,Z}$ starting from any fixed Hamming weight state~\cite{wang_x_2020,hadfield_analytical_2022,dicke_prep,intro_bartschi_dicke_2022}.

\cref{fig:xyr_eng,fig:xyr_overlap} report the energy above the ground state and the corresponding ground state overlap for the XY ring target. The same qualitative pattern appears as in the Pauli setting. Fixed depth does not maintain uniform performance as $n$ grows, while increasing the circuit depth substantially improves the preparation quality.

We also study an XY ring target with a staggered $Z$ field, which still lies within $\g_{XY,Z}$. \Cref{fig:xyr_stag_eng,fig:xyr_stag_overlap} show the same overall trend. Larger systems require deeper circuits, while the ansatz remains effective when the target stays inside the algebra.

We now turn to observables built from $Z$ and $ZZ$ terms, which are the objects relevant for the optimization task below. In this setting, $H_{Z} \in \t{orb}_{\g_{XY,Z}}^{Z}$ and $H_{ZZ} \in \t{orb}_{\g_{XY,Z}}^{ZZ}$. From \cref{thm6:orb_dla_dim_inside}, the dimensions of these orbitals remain polynomially bounded.
\begin{thm}
\begin{align}
\t{dim}\l( \t{orb}_{\g_{XY,Z}}^{Z} \r) &= \Oc\l( n^2 \r)   \\
\t{dim}\l( \t{orb}_{\g_{XY,Z}}^{ZZ} \r) &= \Oc\l( n^4 \r)
\end{align}
\end{thm}

An important point is that these DOS orbitals are constructed using $XY$ rather than the decomposed Bell state projectors. The Bell projectors themselves appear to induce an DOS with exponential dimension, which shows that the choice of basis elements for constructing the DOS is a nontrivial computational issue. For our simulations, we therefore expand the orbital space through a \textit{single} application of $i \ketbra{B_4}{B_4}$ and then map it back to the DOS by applying $i \ketbra{B_3}{B_3}$.

Let $H_{PO}$ be the Hamiltonian that embeds the quadratic cost function for Portfolio Optimization:
\begin{align}
H_{PO} = \sum_{j} h_{j} \, Z_{j} + \sum_{jk} J_{jk} \, Z_{j} Z_{k},
\end{align}
where $h_{j}, J_{jk}$ are coefficients associated with the return vector and covariance matrix of the Markowitz Portfolio formulation. We use the risk aversion parameter $q=0.5$ in all portfolio optimization experiments. In \methods, we describe this mapping in further detail.

In Ref.~\cite{kordonowy2025lie}, it was shown that the restricted circuit associated with $\G_{XY,Z}$ outperforms $\G_{XY,Z,ZZ}$ with the initial state as the Dicke state~\cite{kordonowy2025lie}. We consider the same comparison but utilize the ground state of $H_{XY}$ as the initial state. Let $E_{\t{max}}, E_{\t{min}}$ be the maximum (minimum) energy of $H_{\t{PO}}$. Then the approximation ratio is given by
\begin{align}
\frac{\la H_{\t{PO}} \ra - E_{\t{max}}}{E_{\t{min}} - E_{\t{max}}},
\end{align}
which reaches the value $0$ when the final state $\rho$ has energy $E_{max}$ and the value $1$ when $\rho$ has energy $E_{min}$.

In \cref{fig:res_po,fig:full_po} we report the approximation ratio for circuits of varying depth with the ansatz $\G_{XY,Z}$ and $\G_{XY,Z,ZZ}$ respectively. For $\G_{XY,Z,ZZ}$, we see that as the number of qubits $n$ grows, the performance decays for each fixed number of layers. For moderate depth $12$ we see the dramatic drop in performance, suggesting the onset of a Barren Plateau that is somewhat mitigated at lower depth while better parameterization helps performance for higher depth such as $20$. For $\G_{XY,Z}$, we see strong initial performance for low depth and this performance only improves with increased circuit depth. As with the Pauli string experiments, we see much better performance from the restricted circuit than from the circuit that can express the full Hamiltonian.

\renewcommand{\vem}{1.0em}
\begin{figure*}
\centering

\begin{tabular}{c c}
\begin{subfigure}[t]{0.48\textwidth}
\includegraphics[width=1.0\textwidth]{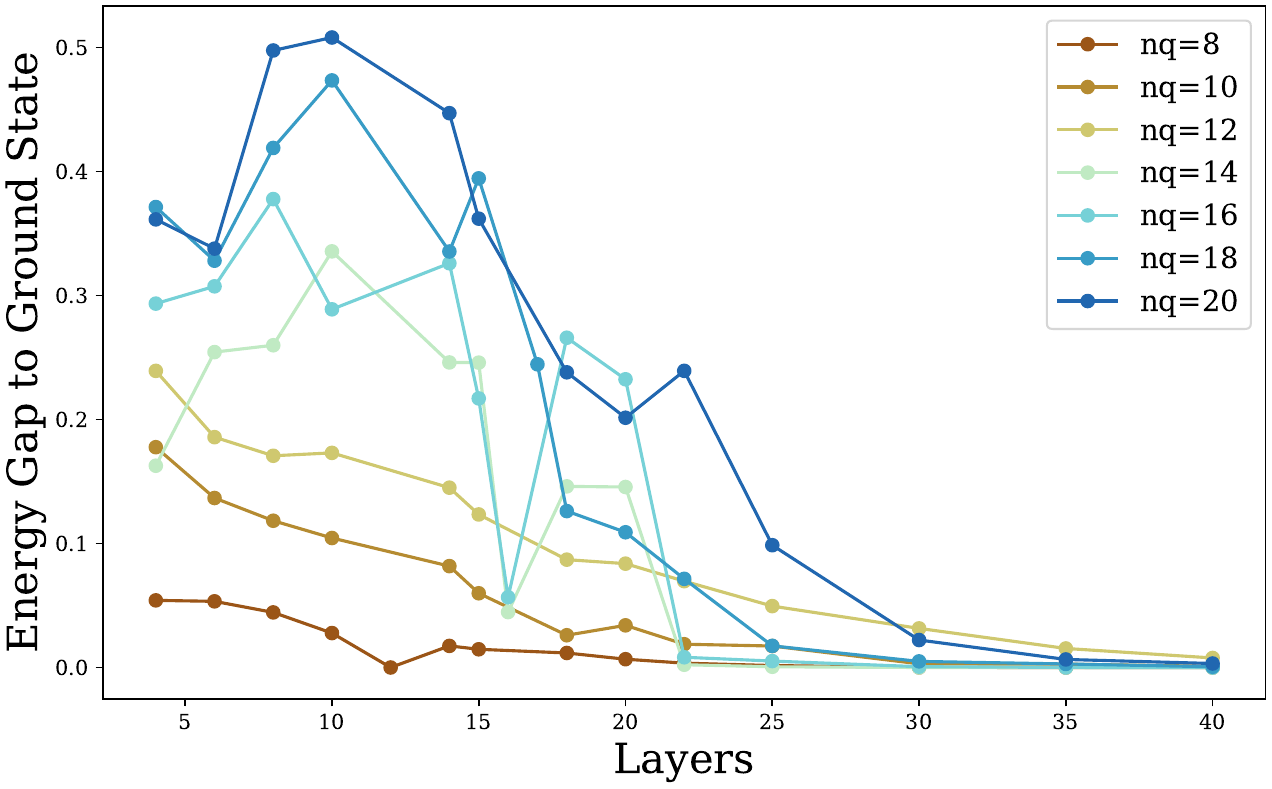}
 \caption{Energy Gap from Ground State}
\label{fig:xyr_eng}
\end{subfigure}
&
\begin{subfigure}[t]{0.48\textwidth}
\includegraphics[width=1.0\textwidth]{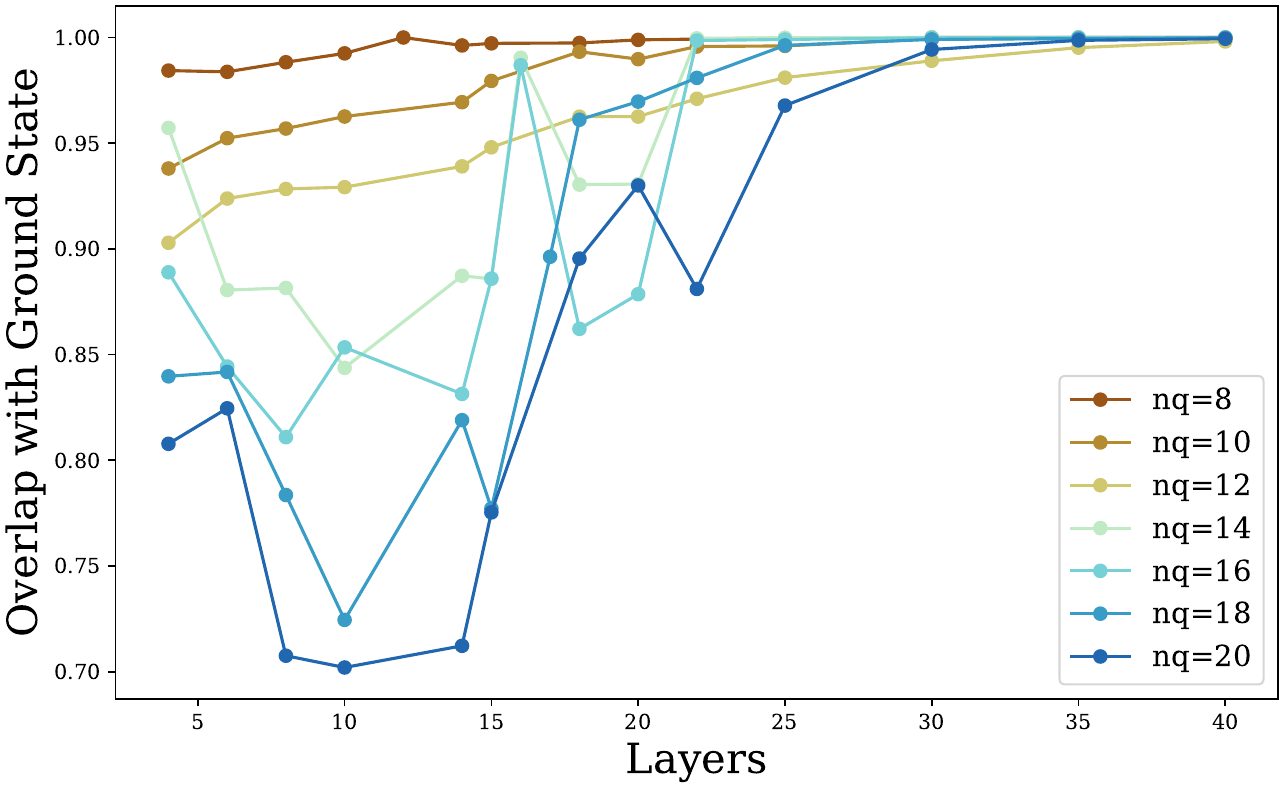}
 \caption{Overlap with Ground State}
\label{fig:xyr_overlap}
\end{subfigure}
\end{tabular}
 \caption{\tsbf{XY Ring Ground State Preparation.}}
\vspace{\vem}

\begin{tabular}{c c}
\begin{subfigure}[t]{0.48\textwidth}
\includegraphics[width=1.0\textwidth]{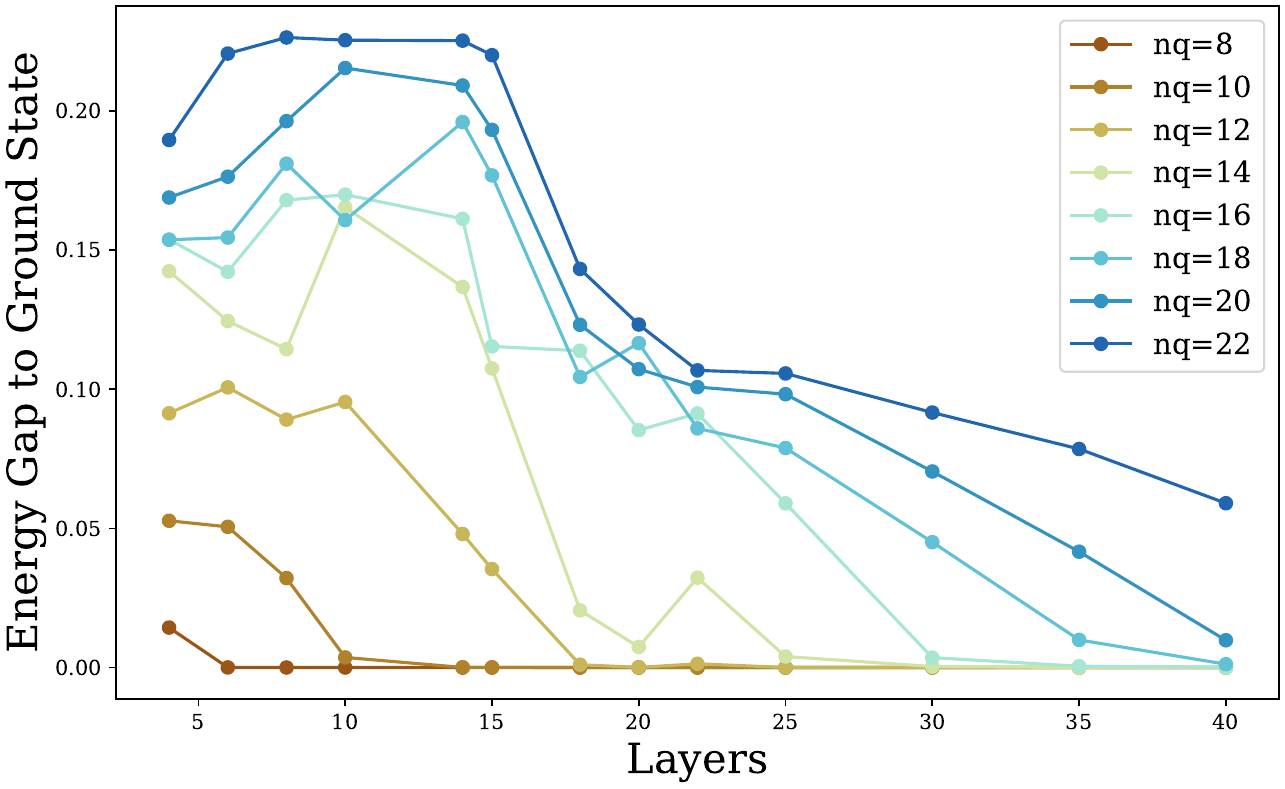}
 \caption{Energy Gap from Ground State}
\label{fig:xyr_stag_eng}
\end{subfigure}
&
\begin{subfigure}[t]{0.48\textwidth}
\includegraphics[width=1.0\textwidth]{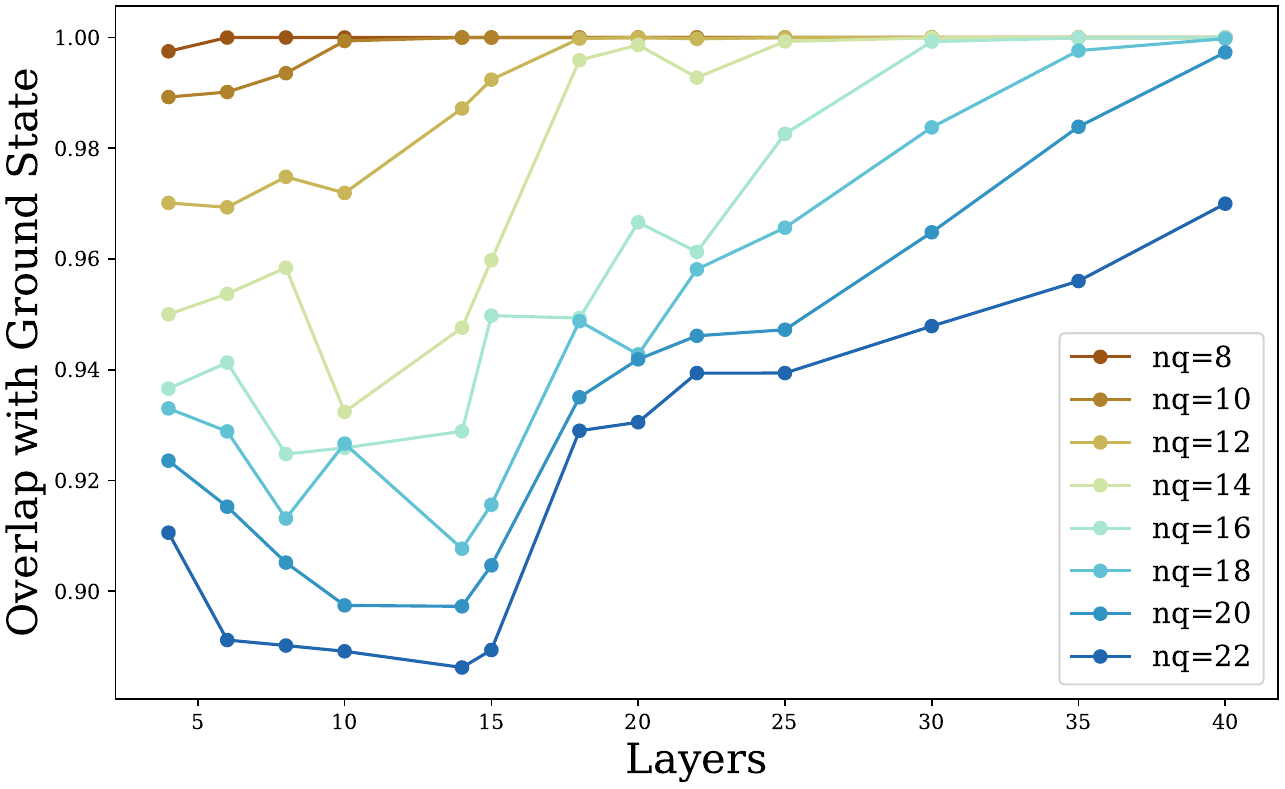}
 \caption{Overlap with Ground State}
\label{fig:xyr_stag_overlap}
\end{subfigure}
\end{tabular}
 \caption{\tsbf{XY Ring with Staggered Z field Ground State Preparation.} }
\vspace{\vem}

\begin{tabular}{c c}
\begin{subfigure}[t]{0.48\textwidth}
\includegraphics[width=1.0\textwidth]{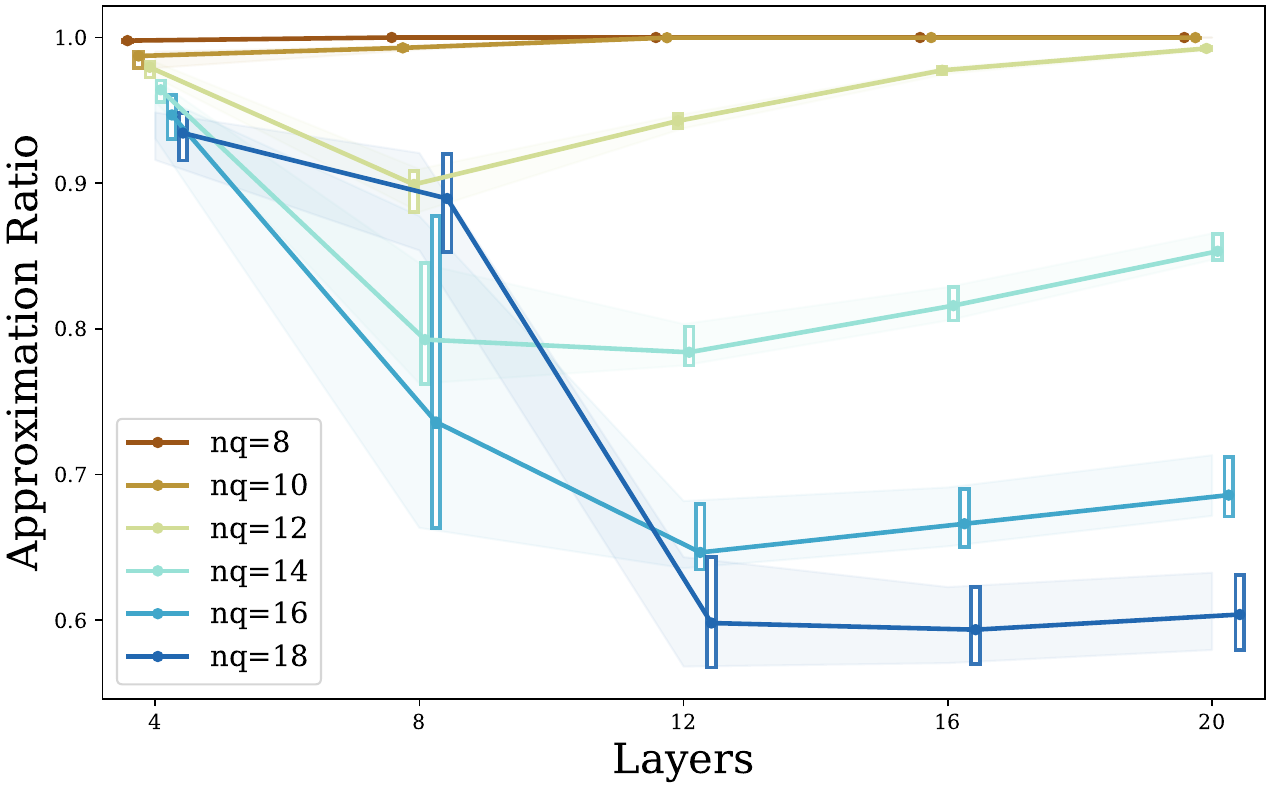}
 \caption{Full Ansatz}
\label{fig:full_po}
\end{subfigure}
&
\begin{subfigure}[t]{0.48\textwidth}
\includegraphics[width=1.0\textwidth]{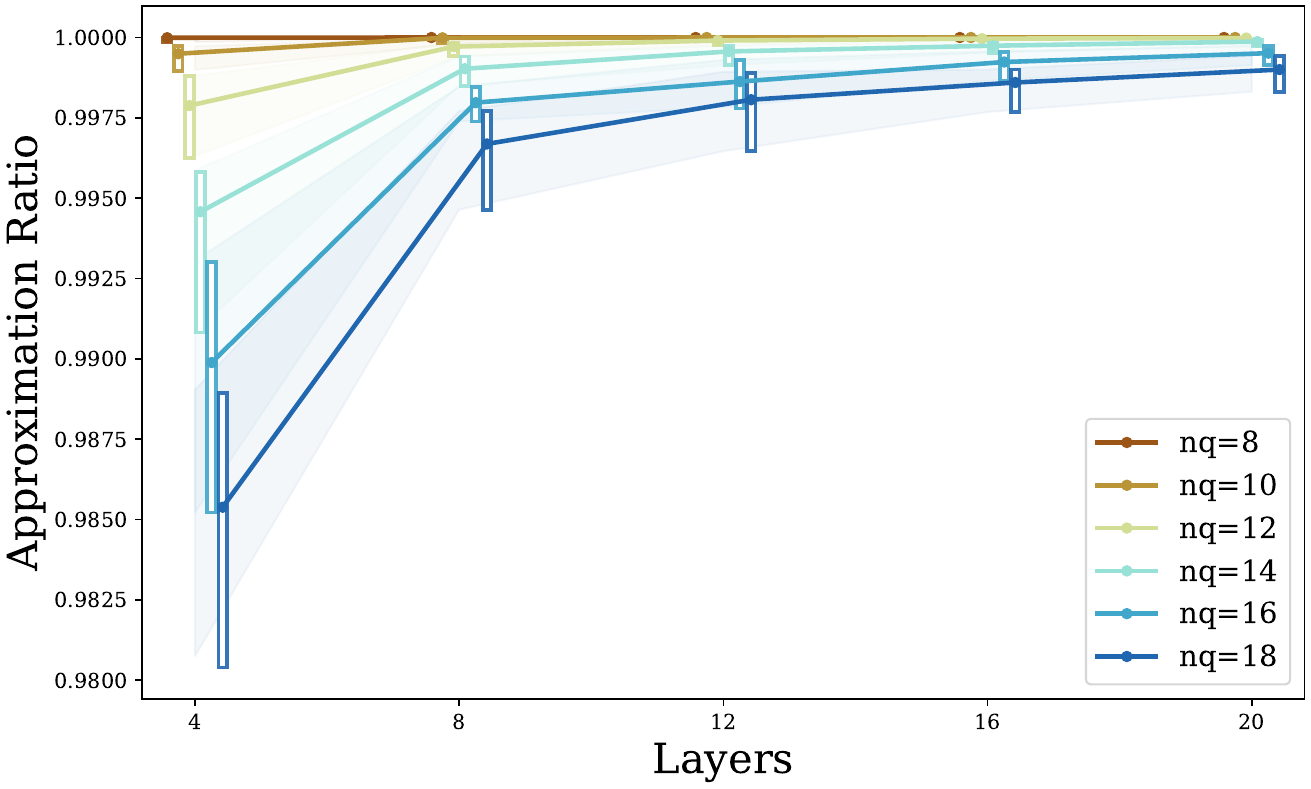}
 \caption{Restricted Ansatz}
\label{fig:res_po}
\end{subfigure}
\end{tabular}
 \caption{\tsbf{Portfolio Optimization with Restricted vs Full Ansatz.} (A) and (B) depict the approximation ratio after training the restricted ansatz and the full ansatz respectively.  }

\end{figure*}

\hypertarget{discuss}{}
\section*{Discussion}

\begin{figure*}[!t]
\centering
\includegraphics[width=0.80\textwidth]{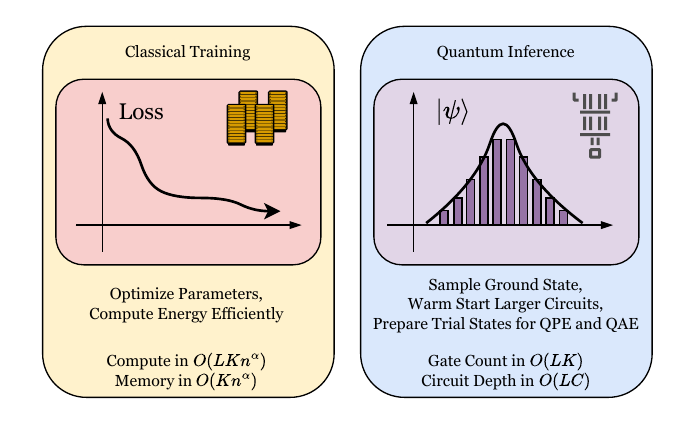}
 \caption{\tbf{Classical Training to Power Quantum Inference.} Given polynomial time training of a parameterized quantum circuit, such as with a polynomial sized DLA ansatz, the quantum circuit can then be utilized for advantage during inference, since the classical overhead for capturing quantum dynamics is alleviated.}
\label{fig:ctrain_qinference}
\end{figure*}

Simulating quantum systems associated with Pauli string rotations play an important role in modern quantum many body physics and quantum computation. In this manuscript, we described a simple framework for simulating such systems within the relevant quantum matrix subspace. \Cref{thm5:orb_dla_dim_outside} shows that when the DLA has dimension $\Oc(n^\ell)$ and the observable lies in the multiplicative sector of order $k$ $\Qc_\g^k$, the DOS dimension is bounded by $\Oc(n^{\ell k})$. In the models we study, the TFXY and XY ring with $Z$ field DLAs each have dimension $\Oc(n^2)$, so observables in the second multiplicative sector yield a DOS of dimension $\Oc(n^4)$, enabling simulation in time $\Oc(n^4 LK)$.

A recent independent work~\cite{barligea2026lie} introduces the reachable operator module, which is equivalent to the DOS up to the Hermitian versus skew-Hermitian convention, and develops a simulation criterion based on the same underlying invariant subspace principle. The two works develop this structure in different directions. Ref.~\cite{barligea2026lie} develops this structure primarily for bosonic and potentially infinite dimensional systems. We develop complementary consequences in finite-dimensional quantum systems, including dimension bounds based on multiplicative sectors of the DLA, efficient simulation for several important classes of generators, efficient gradient computation for these classes, and applications to variational training and quantum inference.

\Cref{thm13:gradinobs} describes the asymptotic gradient cost: the full gradient vector $\nabla_{\bmt}\L$ can be computed in $\Oc(LK\,\t{dim}(\orb))$ via the DOS Adjoint Method, matching the cost of forward simulation up to constant factors.

\cref{fig:ctrain_qinference} depicts how classical training and quantum inference can obtain a synergistic division of labor. Tasks such as simulation, loss evaluation, and gradient computation can run in polynomial classical time for circuits with polynomial DLA. This sidesteps barren plateaus, keeps gradients estimable throughout training, and provides performance via the loss function through evaluation. Then at inference time the trained circuit runs on quantum hardware where polynomial costs of forward simulation are no longer there. For example, if $LK=\l( n \, \log\l( n^2 \r) \r)$ with $\t{dim}\l( \orb \r) = \Oc\l( n^4 \r)$, then this leads to a cubic speedup and even higher speedups if $LK$ grows subquadratically. A further application is QPE, where classically training the ansatz via DOS to maximize ground state overlap can yield dramatic performance boosts, since QPE success probability scales with the overlap squared. Such can also aid in pretraining circuits by dropping out gates that lead to a large DLA. Pretraining with small DLAs was also explored in Refs.~\cite{goh_lie-algebraic_2023,kordonowy2025lie}, where good initialization points were found for the more expressive circuit.

There are many important frontiers related to simulation of quantum dynamics based on our work. There are many types of important systems for which circuits that cannot \textit{express} the observable can play an extraordinary role. Such more limited circuits can be easier to implement and train, leading to ultimately better performance. For Hamiltonians that do not admit a compact eigendecomposition, Trotter Suzuki product formulas~\cite{childs2021theory,yi_spectral_2022,haah2021quantum}, randomized compilers such as qDRIFT~\cite{campbell2019random,babbush2019quantum}, and exact fixed depth Cartan decompositions~\cite{kokcu2022fixed} can decompose the evolution into gates handled exactly by the DOS, with approximation error arising only from the digitalization scheme itself; we treat this in \cref{sec:dyn_digital}.

While the scaling of each adjoint representation is very modest in system size, there are other bottlenecks. Efficient basis construction and adjoint preprocessing can depend strongly on the representation, for example, recent work has shown that symmetry adapted bases can retain efficient preprocessing even when individual operators have exponentially large Pauli expansions~\cite{barligea2026enabling}. Another direction is to consider constructions in other dimensions~\cite{wiersema_classification_2023,kokcu_classification_2024}. In particular, gate sets without exponential DLAs that are well suited for 2D circuits are of special interest.

\hypertarget{methods}{}
\section*{Methods}

We review at a high level the foundations of quantum theory and control as related to our results. We give simple descriptions of the proofs for the theorems of our manuscript, with further details including rigorous proofs given in the Supplementary material.

\subsection*{Generators and Gates over $n$ qubits}

Recall the single qubit Pauli operators:
\begin{align}
\I_{2} &= \ketbra{0}{0} + \ketbra{1}{1} = \begin{pmatrix}
1 & 0 \\
0 & 1
\end{pmatrix}, \\
X &= \ketbra{0}{1} + \ketbra{1}{0} = \begin{pmatrix}
0 & 1 \\
1 & 0
\end{pmatrix}, \\
Y &= -i \ketbra{0}{1} + i \ketbra{1}{0} = \begin{pmatrix}
0 & -i \\
i & 0
\end{pmatrix}, \\
Z &= \ketbra{0}{0} - \ketbra{1}{1} = \begin{pmatrix}
1 & 0 \\
0 & -1
\end{pmatrix} .
\end{align}

\noindent Recall the Kronecker delta
\begin{align}
\delta_{jk} = \begin{cases}
1 & \t{if } j=k \\
0 & \t{else}
\end{cases},
\end{align}
and that $A^{0} = \I_2$. Then any single qubit operator $A$ acting on qubit index $j \in [n]$ can be embedded in the associated space of $n$ qubit operators $\C^{2^{n} \times 2^{n}}$ as
\begin{align}
A_{j} &= \bigotimes_{k=1}^{n} \l( A \r)^{\delta_{jk}} \nonumber \\
&= \underbrace{\I_2 \otimes \cdots \otimes \I_2}_{1:j-1} \, \otimes \, A \, \otimes \, \underbrace{\I_2 \otimes \cdots \otimes \I_2}_{j+1:n} .
\end{align}

In particular, this leads to the standard definitions $\ketbra{0}{0}_{j}, \ketbra{0}{1}_{j}, \ketbra{1}{0}_{j}, \ketbra{1}{1}_{j}, X_{j}, Y_{j}, Z_{j} \in \C^{2^{n} \times 2^{n}} $ and $ \I = \otimes _{j=1}^{n} \I_{2}$. Then the single qubit Pauli rotation gates are given by the simple form
\begin{align}
R^{X}_{j} &= \cos(\theta) \, \I + i \sin(\theta) \, X_j , \\
R^{Y}_{j} &= \cos(\theta) \, \I + i \sin(\theta) \, Y_j , \\
R^{Z}_{j} &= \cos(\theta) \, \I + i \sin(\theta) \, Z_j .
\end{align}

Unitary gates are defined through their generators as $ R^{g}(\theta) = e^{i \theta G_{k}} $. A quantum controlled NOT (CNOT) gate with control qubit $j$ and target qubit $k$ over two qubits is
\begin{align}
\t{CNOT} &= \ketbra{0}{0}_{j} \, \I + \ketbra{1}{1}_{j} \,  X_{k} .
\end{align}

Written over the operator space of those two qubits $\C^{4 \times 4}$ as
\begin{align}
\begin{pmatrix}
1 & 0 & 0 & 0 \\
0 & 0 & 1 & 0 \\
0 & 1 & 0 & 0 \\
0 & 0 & 0 & 1
\end{pmatrix} .
\end{align}

\begin{figure*}
\centering

\includegraphics[width=0.95\textwidth]{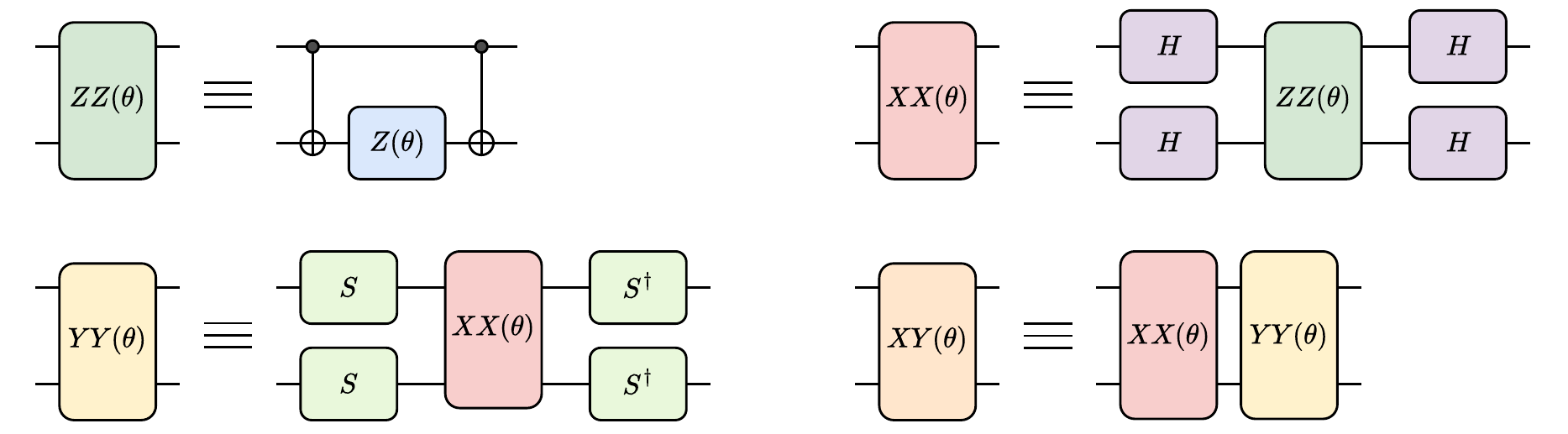}
 \caption{\tsbf{Parameterized ZZ, XX, YY, and XY gates.} Elementary gates $\{ \t{CNOT}, H, S, S^{\dg} \} $ with single spin rotation decomposition for the gates discussed in our experiments. }
\label{fig:gate_reps}
\end{figure*}

Given elementary gate set $\{ \t{CNOT}_{jk} \}_{jk}^{n} \cup \{ H_{j}, S_{j}, S_{j}^{\dg} \}_{j=1}^{n} $ with
\begin{align}
H_{j} = \frac{1}{\sqrt{2}} \l( X_j + Z_j \r) ,
S_{j} = \ketbra{0}{0}_{j} + i \, \ketbra{1}{1}_{j} .
\end{align}

As depicted in \cref{fig:gate_reps}, the parameterized gates we require for our discussion can all be implemented through this elementary set plus the single spin Pauli $Z$ rotation:
\begin{align}
R_{Z}(\theta) = \begin{pmatrix}
e^{i \theta} & 0 \\
0 & e^{-i \theta}
\end{pmatrix} = e^{i\theta} \begin{pmatrix}
1 & 0 \\
0 & e^{-i \, 2 \, \theta}
\end{pmatrix} .
\end{align}

In particular the following decompositions are valid:
\begin{align}
ZZ_{jk}(\theta) &= e^{i \theta Z_j Z_k} = \t{CNOT}_{jk} \, Z_{k}(\theta) \, \t{CNOT}_{jk} \\
XX_{jk}(\theta) &= e^{i \theta X_j X_k} = H_{j} \, H_{k} \, ZZ_{jk}(\theta) \, H_{j} \, H_{k} \\
YY_{jk}(\theta) &= e^{i \theta Y_j Y_k} = S_{j} \, S_{k} \, XX_{jk}(\theta) \, S_{j}^{\dg} \, S_{k}^{\dg} \\
XY_{jk}(\theta) &= e^{i \theta \l( X_j X_k + Y_j Y_k \r)/2} = XX_{jk}(\theta/2) \, YY_{jk}(\theta/2)
\end{align}
Thus, the $ZZ$, $XX$, and $YY$ gate parameters are unchanged, whereas the $XY$ gate uses half of its logical rotation angle for each commuting $XX$ and $YY$ component.

For certain realistic quantum devices, alternative short circuits can be used to implement these fundamental gates well~\cite{barends_digitized_2016}.

\subsection*{Constructing the Dynamic Lie Algebra and Dynamic Observable Subspace}

Both the DLA (\cref{def:dla}) and DOS (\cref{def:dos}) can be constructed through a breadth first search (BFS). In the case of the DOS, we begin with the observable in question. At the beginning of each BFS expansion, we have a list of basis terms that are open and a list of basis terms that are closed. In the first expansion, the open list is just the observable $iO$ and the closed list is empty. In each expansion, we look what new terms can be generated by commuting elements in the open list with each generator and use Gram-Schmidt to find the orthonormal perpendicular matrix to the existing space. We then add the elements of the open list to the closed list and redefine the list of new terms as the new open list. When the new list is empty (the new elements were all inside the existing space), we end the BFS expansion. During this phase, we represent generators and basis elements symbolically, such as weighted sums of Pauli strings, and commutation generates new such symbolic representations with numerical coefficients such that we never have to numerically represent a matrix in $\C^{2^{n} \times 2^{n}}$ explicitly. The formal definitions are in \cref{def:dla,def:dos} and the efficient symbolic construction is detailed in \cref{sec:bfs_construction}.

The importance of such preprocessing, including Lie closure and construction of adjoint data, has also been studied using symmetry adapted representations for structured polynomial dimensional DLAs~\cite{barligea2026enabling}.

\subsection*{Constructing Representations of the Density, the Projection and the Adjoint Superoperators over the Dynamic Observable Subspace}

Once the DOS or DLA has been constructed, we associate each basis element with an index. We build the relevant operators for this space.

For each generator, we construct its numerical adjoint representation by computing the symbolic commutator with every basis element and decomposing the result over the indexed basis using Gram-Schmidt. This leads to a sparse matrix over the indices, $ \t{ad}_{iG} \in \R^{\t{dim}\l( \orb \r) \times \t{dim}\l( \orb \r)}$.

The construction of the projection superoperator follows similarly. We know that for some entries the generator will commute, for example if the generator qubit indices and the basis qubit indices have zero overlap. Basis elements that do not commute with $iG$ span the active subspace of $iG$; the projection superoperator $\t{proj}_{\t{Com}(iG)}^{\perp}$ acts as the identity on this subspace and zero elsewhere. The adjoint $\t{ad}_{iG}$ has zero columns for commuting elements and nonzero columns (from the decomposition of $[iG, iB_j]$ over $\B$) for noncommuting ones. For Pauli generators, these give the evolution $\t{evo}_{iG}(i\vr;\theta) = i\vr + (\cos 2\theta - 1)\,\t{proj}_{\t{Com}(iG)}^{\perp}(i\vr) + \frac{\sin 2\theta}{2}\,\t{ad}_{iG}(i\vr)$ (\cref{thm8:paulievo}). For rank one projector generators, the evolution uses $\t{ad}$ and $\t{ad}^{2}$ directly (\cref{thm9:projevo}). Full details are in \cref{part2:repindos}.

The construction of the density operator inside the space requires a decomposition into the space. For simple or often studied initial states, this can be done analytically as for the initial states we considered in the manuscript. We discuss these projections into the DOS in the relevant sections \cref{part3:exp}.

Because these constructions mix numerical and symbolic representations, initialization warrants runtime consideration. However, gradient computation is more expensive in our experiments, and the ratio of initialization time to the time for 100 gradient steps decreases with system size. Initialization could nevertheless be accelerated with improved data structures, parallelism, or analytic simplifications.

\subsection*{Proof Techniques for Stated Theorems}

\begin{figure*}
\begin{tabular}{cc}
\includegraphics[height=0.35\textheight]{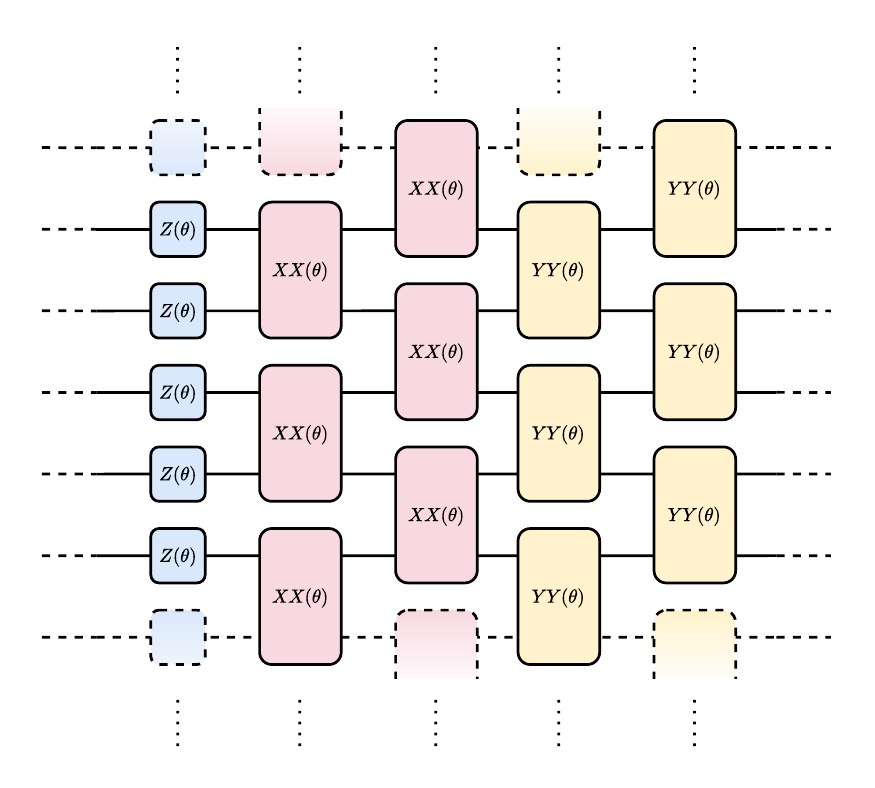} &
\raisebox{-0.5em}{\includegraphics[height=0.368\textheight]{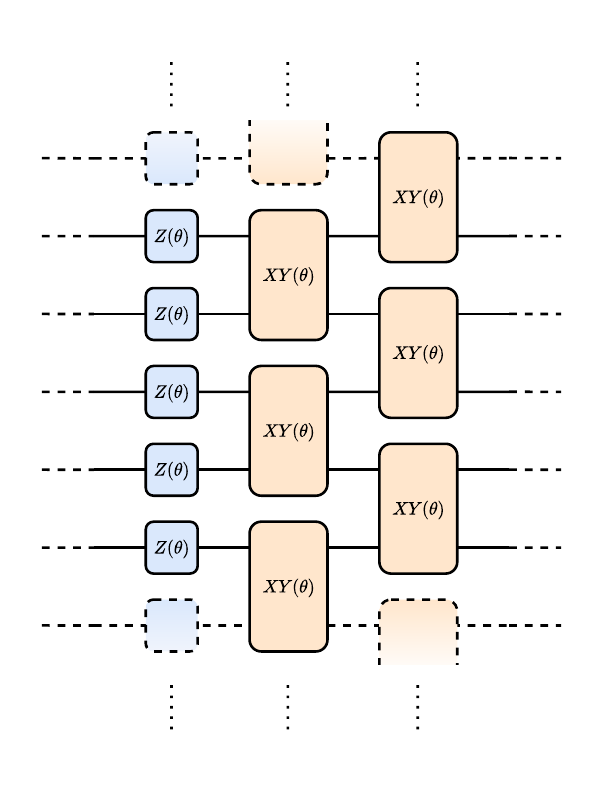}}
\end{tabular}
\caption{\tsbf{Two Circuits with DLAs of Polynomial Dimension.} Left depicts the circuit associated with the generator set $\G_{Z,XX,YY}$ while right depicts the circuit associated with the generator set $\G_{Z,XY}$.}
\label{fig:circ_depict}
\end{figure*}

This section presents proof overviews for the manuscript's fundamental theorems; full proofs appear in \cref{part1:dos,part2:repindos}.

Following from \cref{def:dos}, \cref{thm1:dosobs} is clear as it is a basis element of $\B$, specifically the first one considered in the construction of the orthonormal basis. Then \cref{thm2:heidos} follows by recognizing that we can simulate $iO$ backwards in time to the initial $i\rho$. Since the dynamics are entirely captured in $\orb$, $iO$ has zero support outside this matrix subspace and so only the projection of the initial state $i\rho(0)$ can impact the trace inner product. Moreover, after the projection, the support of two matrices is entirely in the matrix subspace, meaning we define the inner product to be over only the DOS itself.

\cref{thm4:dosrep} establishes the fundamental representation of dynamics inside the DOS and is proved in \cref{sthm:dosrep}. It follows from recognizing that since the infinitesimal changes under the generators are captured by the DOS then for any operator in the Lie group of the generators, we can track its impact on the DOS representations of the density operator and the observable. Then specifically since $U_{j:k}(\bmt) \in e^\g$, \cref{thm4:heicut} follows for $i\vr(t)$ and $iO(t)$ showing that we can capture the dynamics of the system by simulating $i\rho$ (specifically its DOS representation $i \vr$) forward from $0$ to $t$ and $iO$ backward from $T$ to $t$ inside the DLA.

\cref{thm5:orb_dla_dim_outside} bounds the dimension of the DOS by the dimension of the relevant multiplicative sector of the DLA: if the observable $iO$ lies in the multiplicative sector of order $k$ $\Qc_\g^k$, then $\orb \subseteq \Qc_\g^k$ and $\t{dim}(\orb) \leq \t{dim}(\g)^k$. When $k$ is small, this enables simulation with polynomial cost even for systems where the full Hilbert space is exponentially large.

The proof of \cref{thm7:evoexpadj} appears in \cref{sthm:evoexpadj}. \Cref{thm8:paulievo} is derived in \cref{sthm:pauliupdate}, and \cref{thm9:projevo} is derived in \cref{slem:diffevo}. \Cref{thm10:eigevo} follows by applying the projector result to the spectral decomposition of the generator; the derivation appears in \cref{sec:dyn_eigham}.

\cref{thm12:multidosrep} expands the DOS for linear combinations of observables from algebraic manipulation. In particular, we can have multiple orbitals such as occurs in several models we considered. Then we can sum the gradient contributions established by \cref{thm13:gradinobs}, which follows by analyzing DOSAM. The theorem follows from algebraic manipulation, which is depicted visually \cref{fig:schrohei} and proved in \cref{sthm:multidosrep,sthm:grad}.

\subsection*{TFXY Model}

Our experiments with the anisotropic transverse field XY (TFXY) model uses the gate set
\begin{align}
\G_{Z, XX, YY} = \l\{ X_{j} X_{j+1},\, Y_{j} Y_{j+1} \r\}_{j=1}^{n} \cup \l\{ Z_{j} \r\}_{j=1}^{n},
\end{align}
where indices wrap periodically. \cref{fig:circ_depict} on the left shows a single layer of this circuit. We optimize the target Hamiltonian
\begin{align}
H = \sum_{j} h_{j} \, Z_{j} + \eta \sum_{j} X_{j} X_{j+1} + \nu \sum_{j} Y_{j} Y_{j+1},
\end{align}
starting from the initial state $\ket{\psi(0)} = \ket{1}^{ \otimes n}$. Since $H \in \g_{TFXY}$, the DOS equals the DLA and $\t{dim}(\t{orb}) \leq \t{dim}(\g_{TFXY}) = \Theta(n^2)$ by \cref{thm6:orb_dla_dim_inside}. In the main text we use the individually parameterized ansatz
\begin{align}
U(\bmt) = \prod_{\ell=1}^{L} \prod_{j=1}^{n} e^{i \theta_{j\ell} Z_j } e^{i \theta_{j\ell} X_j X_{j+1} } e^{i \theta_{j\ell} Y_j Y_{j+1} } .
\end{align}
A shared parameter Hamiltonian based variant $U_S$ is described in \cref{sec:tfxy_shared_app}; the corresponding shared parameter gradient formulas are developed later in \cref{ssec:grad_shared}. Full derivations are collected in \cref{part3:exp}.

Our baseline clean TFXY benchmark uses the uniform field choice $h_j = 1$ and isotropic couplings $\eta = \nu = 1$. For the anisotropic robustness experiments we keep the same generator support but change the couplings to $\eta = 0.6$ and $\nu = 0.4$, so the target remains in $\g_{TFXY}$. For the staggered field robustness experiments we replace the uniform field by $h_j = D(-1)^j$ with $D=1.0$, again keeping the target inside the same DLA and DOS.

\subsection*{Kitaev with disordered ZZ}

We study a disordered XY spin model with Hamiltonian
\begin{align}
H_{SK,XX,YY} = \eta \sum_{j} X_{j} X_{j+1} + \nu \, Y_{j} Y_{j+1} + \sum_{j<k} J_{jk} \, Z_{j} Z_{k},
\end{align}
where $J_{jk} \sim \mathcal{N}(0,1)$ are disordered couplings and $\eta = 0.6$, $\nu = 0.4$. The $Z_j Z_k$ terms lie outside the DLA $\g_{TFXY}$, so $iH_{SK,XX,YY}$ lies in the second multiplicative sector $\Qc_\g^{2}$. By \cref{thm5:orb_dla_dim_outside}, $\t{dim}(\t{orb}_\g^{ZZ}) = \Oc(n^4)$. We also study an extended model with an additional disordered $Z$ field,
\begin{align}
H_{SKZ,XX,YY} = \sum_{j} h_{j} \, Z_{j} + H_{SK,XX,YY}, \quad h_j \sim \mathcal{N}(0,1),
\end{align}
splitting it into two DOS orbitals $\t{orb}_\g^{Z}$ and $\t{orb}_\g^{ZZ}$ of dimensions $\Oc(n^2)$ and $\Oc(n^4)$ respectively via \cref{thm12:multidosrep}. We compare training with the restricted ansatz $\G_{TFXY}$ against the full ansatz $\G_{XX,YY,Z,ZZ} = \G_{TFXY} \cup \{ Z_j Z_k \}_{j<k}^n$ with $\t{dim}(\g_{XX,YY,Z,ZZ}) = 4^n - 1$.

\subsection*{Portfolio Optimization Cost Hamiltonian}

The binary asset selection formulation for Portfolio Optimization is associated with a return vector $p \in \R^{n}$ and a covariance matrix $C \in \R^{n \times n}$. Given a risk aversion parameter $q$, the optimization task can be formulated as a quadratic form with a cardinality constraint:
\begin{align}
\min_{x} \;  q \, x^{T} \cdot C \cdot x - p \cdot x, \;
\end{align}
which trades off the risk associated with $C$ against the return $p$ for a selection. Then the cost Hamiltonian for this constrained quadratic binary optimization task can be formulated~\cite{herman_portfolio_2022,hadfield_representation_2021} as:
\begin{align}
H_{PO} &= h_{0} \, \I +  \sum_{j=1}^{n} h_{j} \, Z_{j} + \sum_{j<k}^{n} J_{jk} \, Z_{j} Z_{k} \\
h_{0} &= \frac{q}{2} \sum_{j<k} C_{jk} + \frac{q}{2} \sum_j C_{jj} - \frac{1}{2} \sum_j p_j \\
h_{j} &= \frac{p_j}{2} - \frac{q}{2} \l( C_{jj} + \sum_{k \ne j} C_{jk} \r) \\
J_{jk} &= \frac{q}{2} C_{jk}
\end{align}

\subsection*{XY ring with Z Model}

The XY ring with local $Z$ rotations uses the generator set
\begin{align}
\G_{XY,Z} = \l\{ XY_{j,j+1} \r\}_{j=1}^{n} \cup \l\{ Z_{j} \r\}_{j=1}^{n},
\end{align}
where $XY_{jk} = \frac{1}{2}(X_j X_k + Y_j Y_k)$ and indices wrap periodically. The XY mixer preserves Hamming weight and is the natural two body mixer for fixed weight dynamics. The corresponding DLA has dimension $\Oc(n^2)$~\cite{kordonowy2025lie}. We optimize the XY ring Hamiltonian
\begin{align}
H_{XY} = - XY_{1,n} - \sum_{j=1}^{n-1} XY_{j,j+1},
\end{align}
starting from a fixed Hamming weight initial state. Since $H_{XY} \in \g_{XY,Z}$, the DOS has dimension $\Oc(n^2)$ by \cref{thm6:orb_dla_dim_inside}. For the Portfolio Optimization task, the cost Hamiltonian $H_{PO}$ (defined above) decomposes into $\t{orb}_{\g_{XY,Z}}^{Z}$ and $\t{orb}_{\g_{XY,Z}}^{ZZ}$ of dimensions $\Oc(n^2)$ and $\Oc(n^4)$ respectively. Each XY gate decomposes via Bell projectors $P_+ = \ketbra{B_3}{B_3}$, $P_- = \ketbra{B_4}{B_4}$ as $e^{i\theta XY_{jk}} = e^{i\theta P_+} e^{-i\theta P_-}$, enabling the rank one projector simulation of \cref{part3:exp}.

\hypertarget{ack}{}
\section*{Acknowledgements}

\begin{figure}
\begin{algobox}
\setlength{\intextsep}{0.5em}

\captionof{algorithm}{\tsf{Subspace Simulation}}\label{alg:sim}
\vspace{-0.9em}
\hrule
\vspace{0.2em}
\begin{algorithmic}[1]
\Statex \textbf{Inputs: } $ \tt{rho} $, $ \bmt $, $ \tt{ef} $, $ \tt{em} $
\vspace{0.2em}
\State $ t := 0 $
\State $ J := \t{length of } \tt{evomats} $
\While{$ t < LK $}
    \State $k := 0$
    \State $\tt{rho}_{n} \rgt 0$ (zero matrix)
    \While{$ k < K $}
        \State $ t \rgt t + 1 $
        \State $ k \rgt k+1 $
        \State $ j \rgt 0 $
        \State $ \tt{rho}_g \rgt \tt{rho}_n $
        \While{$ j < J $}
            \State $ j \rgt j + 1 $
            \State $\tt{rho}_g \rgt \tt{rho}_g + \tt{ef}[j](\theta) \; \tt{em}[k][j] \, \cdot \, \tt{rho}_n $
        \EndWhile
        \State $\tt{rho}_{n} \rgt \tt{rho}_{g}$
    \EndWhile
    \State $\tt{rho} \rgt \tt{rho}_n$
\EndWhile
\State \textbf{return} $ \tt{rho} $
\end{algorithmic}

\end{algobox}
\vspace{-1.0em}

\captionof*{algorithm}{{Alg.~\ref*{alg:sim}}: \tsbf{Simulating Subspace Dynamics with Evolution Operators.} Given $\rho$ and $ U(\bmt) = \prod_{\ell=1}^{L} \prod_{k=1}^{K} e^{i \theta_{t} G_{t \t{mod} K}} $ of parameterized Pauli gates, this algorithm simulates the dynamics of $ U(\bmt) \, i \rho \, U^{\dg}(\bmt) $ through vector representation $ \tt{rho} $ and matrix representations $ \tt{em} $ with associated $ \theta $ parameterized coefficients $ \tt{ef} $.}

\setlength{\intextsep}{\ointextsep}
\end{figure}

The author thanks Bibhas Adhikari, Sarvagya Upadhyay, Quoc Hoan Tran, Koki Chinzei, Yasuhiro Endo, and Hirotaka Oshima for their insightful comments and discussion. We acknowledge the use of AI models from Anthropic, OpenAI, and Google for literature review, proof checking, and manuscript revisions as well as improving our initial code base.

\section*{Funding Declaration}

This work was supported by Fujitsu Research of America, Inc. No external funding was received for this work.

\section*{Author Contributions}

HL initiated the investigation, developed the original code and methodologies, and wrote the subsequent manuscript including the figures and plots.

\section*{Data Availability}

All data disclosed in this manuscript appear in the plots of specific figures. The associated data points are available in the Supplemental Data as Comma-Separated Values (CSV) files.

\section*{Competing Interest}

The author is an employee of Fujitsu Research of America, Inc., a research subsidiary of Fujitsu Limited. The author declares no other competing financial or non-financial interests.

\printbibliography

\appendix
\pagebreak
\clearpage
\newpage

\setcounter{page}{1}

\onecolumn

\setlength{\cftsecnumwidth}{3em}
\setlength{\cftsubsecnumwidth}{4em}

\renewcommand{\thesection}{S.\arabic{section}}
\renewcommand{\thesubsection}{S.\arabic{section}.\arabic{subsection}}
\crefname{section}{Supplementary Section}{Supplementary Sections}
\renewcommand{\theequation}{S\arabic{equation}}
\setcounter{equation}{0}

\renewcommand{\thesthm}{S\arabic{sthm}}
\renewcommand{\thesdefi}{S\arabic{sdefi}}
\setcounter{sdefi}{0}
\setcounter{sthm}{0}

\begin{center}
\textbf{\sffamily\Huge Supplementary Information}
\end{center}

\begingroup
\small
\tableofcontents
\endgroup

\clearpage
\newgeometry{top=1.5in,bottom=1.75in,left=1.25in,right=1.25in}
\onecolumn
\fancyfoot[C]{\hspace*{-\oddsidemargin}\hspace*{-1in}\thepage}
\pagestyle{fancy}
\newpage

\section*{Glossary and Notation}
\addcontentsline{toc}{section}{Glossary and Notation}
\noindent Matrix spaces and spans of skew-Hermitian operators are taken over $\R$. The following notation is used throughout the manuscript and supplementary information.

\begin{center}
\begin{tabular}{p{0.26\textwidth}p{0.66\textwidth}}
\textbf{Notation or term} & \textbf{Meaning} \\[0.5em]
$\C^{2^n}$ & Hilbert space of $n$ qubits. \\[0.4em]
$\mf{u}\l( 2^n \r)$ & Real linear subspace of $\C^{2^n \times 2^n}$ consisting of skew-Hermitian matrices, with the commutator as its Lie bracket. \\[0.4em]
$\mc{P}$, $i\mc{P}$ & Pauli strings and their skew-Hermitian counterparts, $i\mc{P}=\l\{ iP : P\in\mc{P} \r\}$. \\[0.4em]
$\G$, $\g$ & Hermitian generator set and its Dynamic Lie Algebra (DLA), generated by $iG$ for $G \in \G$. \\[0.4em]
$\orb$ & Dynamic Observable Subspace (DOS): the smallest real matrix subspace containing $iO$ and invariant under commutation with the generators. \\[0.4em]
$\B$, $\B^\g$ & Orthonormal bases of the DOS and DLA, respectively. \\[0.4em]
$\la iA, iB \ra_{\hs}$ & Hilbert-Schmidt inner product, $\t{Tr}\l( (iA)^\dg iB \r)$. \\[0.4em]
$\la iA, iB \ra_{\dos}$ & Hilbert-Schmidt inner product after projection onto the DOS. If $iA,iB \in \orb$, then $\la iA,iB \ra_{\dos}=\la iA,iB \ra_{\hs}$. \\[0.4em]
Superoperator & Real linear map $\mf{u}\l(2^n\r)\to\mf{u}\l(2^n\r)$, or a real linear map between subspaces of $\mf{u}\l(2^n\r)$. \\[0.4em]
$\mc{I}$ & Identity superoperator, equivalent to conjugation by the Hilbert space identity $\I$: $\mc{I}\l( iA \r)=\I \, iA \, \I=iA$. \\[0.4em]
$iB \oask iC$ & Superoperator outer product. It maps $iD$ to $iB \, \la iC, iD \ra_{\hs}$. \\[0.4em]
$\t{ad}_{iG}$ & Adjoint superoperator, $\t{ad}_{iG}\l( iA \r) = [iG,iA]$. \\[0.4em]
$\t{evo}_{iG}\l( \,\cdot\,;\theta \r)$ & Evolution superoperator, $e^{\theta \, \t{ad}_{iG}}\l( i\rho \r) \equiv e^{i\theta G} \, i\rho \, e^{-i\theta G}$. \\[0.4em]
$\t{proj}_{W}$, $\t{proj}_{W}^{\perp}$ & Orthogonal projection superoperators onto $W$ and its orthogonal complement in the stated ambient matrix space. \\[0.4em]
$i\vr$ & Projected density operator, $i\vr = \t{proj}_{\g}^{iO}\l( i\rho \r)$. \\[0.4em]
$F$ & Hilbert space projector satisfying $F^\dg=F=F^2$, distinct from a projection superoperator on a matrix space.
\end{tabular}
\end{center}

\clearpage

\part{\sffamily Observable Dynamics of Quantum Systems}\label{part1:dos}

\section{Preliminaries}

In this section we develop preliminaries, some of which were already covered in the main manuscript, especially in \results and \methods. We develop the basic description of quantum operators, matrix vector spaces, the projection superoperator, and the adjoint superoperator.

\subsection{Tensor Products, Qubit Basis, and Pauli Group}

Recall the standard definition of the Kronecker tensor product for matrices.
\begin{sdefi}[Tensor Product for Matrices]
Given $ A \in \C^{p \times p}, B \in \C^{q \times q} $, the tensor product $ A \otimes B \in \C^{p \, q \, \times p \, q} $ has entries given by:
\begin{align}
(A \otimes B)_{(i-1)\,q+k, (j-1) \, q + \ell} = A_{ij} B_{k \ell} \; \text{ for } 1 \leq i, j \leq p, 1 \leq k, \ell \leq q .
\end{align}
\end{sdefi}

\noindent For example, given two matrices over $\mathbb{C}^{2 \times 2}$
\begin{align}
A = \begin{pmatrix}
a_{11} & a_{12} \\
a_{21} & a_{22} \\
\end{pmatrix}, \\
B = \begin{pmatrix}
b_{11} & b_{12} \\
b_{21} & b_{22}
\end{pmatrix},
\end{align}
the Kronecker product $A \otimes B $ over $\mathbb{C}^{4 \times 4}$ is
\begin{align}
A \otimes B &= \begin{pmatrix}
a_{11} B & a_{21} B \\
a_{21} B & a_{22} B
\end{pmatrix} \\
&= \begin{pmatrix}
a_{11} b_{11} & a_{11} b_{12} & a_{12} b_{11} & a_{12} b_{12} \\
a_{11} b_{21} & a_{11} b_{21} & a_{12} b_{21} & a_{12} b_{22} \\
a_{21} b_{21} & a_{21} b_{21} & a_{22} b_{21} & a_{22} b_{22} \\
a_{21} b_{21} & a_{21} b_{21} & a_{22} b_{21} & a_{22} b_{22} \\
\end{pmatrix} .
\end{align}

The single qubit Pauli operators in $\C^{2 \times 2}$ are:
\begin{align}
X &= \begin{pmatrix}
0 & 1 \\
1 & 0
\end{pmatrix} , \\
Y &= \begin{pmatrix}
0 & -i \\
i & 0
\end{pmatrix} , \\
Z &= \begin{pmatrix}
1 & 0 \\
0 & -1
\end{pmatrix} ,
\end{align}
along with the $2 \times 2$ identity operator:
\begin{align}
\I_2 = \begin{pmatrix}
1 & 0 \\
0 & 1
\end{pmatrix} .
\end{align}

For convenience, we also define standard ketbra matrix basis:
\begin{align}
\ketbra{0}{0} &= \begin{pmatrix}
1 & 0 \\
0 & 0
\end{pmatrix} , \\
\ketbra{0}{1} &= \begin{pmatrix}
0 & 1 \\
0 & 0
\end{pmatrix} , \\
\ketbra{1}{0} &= \begin{pmatrix}
0 & 0 \\
1 & 0
\end{pmatrix} , \\
\ketbra{1}{1} &= \begin{pmatrix}
0 & 0 \\
0 & 1
\end{pmatrix} .
\end{align}

Recall the commutator $[A,B] = AB - BA $. The single qubit Pauli operators satisfy the commutator relationships:
\begin{align}
[ X, Y ] &= -[ Y, X ] =  2 i \, Z, \\
[ Y, Z ] &= -[ Z, Y ] =  2 i \, X, \\
[ Z, X ] &= -[ X, Z ] =  2 i \, Y,
\end{align}
with remaining commutators equivalent to the zero matrix. They satisfy the anticommutator relationships:
\begin{align}
\{ X, X \} = \{ Y, Y \} = \{Z,Z\} = 2 \, \I,
\end{align}
with remaining anticommutators equivalent to the zero matrix. Recall the Kronecker delta
\begin{align}
\delta_{jk} = \begin{cases}
1 & \t{if } j=k \\
0 & \t{else}
\end{cases},
\end{align}
and that $A^{0} = \I_2$. Then any single qubit operator $A$ acting on qubit index $j \in [n]$ can be embedded in the associated $n$ qubit operator space $\C^{2^{n} \times 2^{n}}$ as
\begin{align}
A_{j} &= \bigotimes_{k=1}^{n} \l( A \r)^{\delta_{jk}} \\
&= \underbrace{\I_2 \otimes \cdots \otimes \I_2}_{1:j-1} \, \otimes \, A \, \otimes \, \underbrace{\I_2 \otimes \cdots \otimes \I_2}_{j+1:n} .
\end{align}

\noindent This allows us to define $\ketbra{0}{0}_{j}, \ketbra{0}{1}_{j},\ketbra{1}{0}_{j}, \ketbra{1}{1}_{j}, X_j, Y_j, Z_j \in \C^{2^{n} \times 2^{n}}$.

Then a Pauli string is defined over $n$ qubits as $ P = P_1 \otimes \cdots \otimes P_n $ where $ P_{j} \in \{ \I_2, X, Y, Z \} $. The Pauli group (under standard multiplication) is all possible (nontrivial) Pauli operators with coefficients in $\{\pm 1, \pm i \}$
\begin{sdefi}[Pauli Group]
\begin{align}
\mathcal{P}_{G} = \{ \pm 1, \pm i \} \cdot \{ \I_2, X, Y, Z \}^{ \otimes n} / \{ \I \} .
\end{align}
\end{sdefi}

\noindent For our purposes we just interested in the skew Pauli set or basis.

\begin{sdefi}[Skew Pauli Basis]
\begin{align}
i \P = i \{ \I_2, X, Y, Z \}^{ \otimes n} / \{ \I \} .
\end{align}
\end{sdefi}

\subsection{Skew-Hermitian Matrix Spaces, Lie Algebra, and Linear Operations over Matrix Spaces}

We review the foundations of matrix spaces, linear operators over matrix spaces, and Lie algebras. Recall that a Hermitian matrix $A$ satisfies self equality under negation of the conjugate transpose.
\begin{sdefi}[Skew Hermiticity, Hermiticty, Unitarity]
A matrix $B \in \C^{m \times m}$ is skew-Hermitian if and only if
\begin{align}
B = -B^{\dg} .
\end{align}
A matrix $A \in \C^{m \times m}$ is Hermitian if and only if
\begin{align}
A = A^{\dg} .
\end{align}
A matrix $U \in \C^{m \times m}$ is unitary if and only if
\begin{align}
U U^{\dg} = U^\dg U = \I .
\end{align}
\end{sdefi}

\noindent In particular, every skew-Hermitian matrix $B = i A$ where $A$ is Hermitian. As such, there is no fundamental difference between representing quantum mechanics through skew-Hermitian versus Hermitian representations.

\begin{sthm}[Hermitian Decomposition of Skew-Hermitian Matrices]
If $B = -B^{\dg}$ then $B = i A$ such that $A = A^{\dg}$.
\end{sthm}

Moreover, Hermiticity and skew Hermiticity are preserved under unitary conjugation.

\begin{sthm}
If $iA$ is skew-Hermitian and $U$ is unitary, then $U \, iA \, U^\dg $ is skew-Hermitian.
\end{sthm}

\begin{proof}
\begin{align}
\l( U \, i A \, U^\dg \r)^{\dg} &= \l( U^\dg \r)^\dg \, \l( i A \r)^\dg \, U^\dg \\
&= - U \, i A \, U^\dg.
\end{align}
\end{proof}

Note that in this way Skew-Hermitian operators have an essential equivalence to Hermitian operators. As noted in the main manuscript, the trace $\t{Tr} : \C^{m \times m} \rightarrow \C$ is a linear mapping such that $\t{Tr}(A+B) = \t{Tr}(A) + \t{Tr}(B)$. The trace over this matrix space can be defined as
\begin{align}
\t{Tr}(A) = \sum_{j=1}^{m} \bra{j} A \ket{j} .
\end{align}

Matrix vector spaces are bona fide vector spaces, satisfying all the conditions stated within. The standard inner product for Skew-Hermitian matrices is given by the Hilbert-Schmidt Inner Product.

\begin{figure}[!t]
\begin{algobox}
\setlength{\intextsep}{0.5em}

\captionof{algorithm}{\tsf{Gram-Schmidt}}\label{alg:gram}
\vspace{-0.9em}
\hrule
\begin{algorithmic}[1]
\Statex \textbf{Inputs:} $ iT $, $ \B $
\vspace{0.2em}
\State $iR \rgt iT$
\For{$ iB \in \B $} \Comment{basis is assumed orthonormal}
    \State $r \rgt \t{Tr}\l( (iR)^{\dg} \, iB \r)$
    \State $iR \rgt iR - r \, iB $
\EndFor
\State \textbf{return} $ iR / \sqrt{\t{Tr}\l( (iR)^\dg \, iR \r)} $
\end{algorithmic}
\end{algobox}
\vspace{-1.0em}

\captionof*{algorithm}{{Alg.~\ref*{alg:gram}}: \tsbf{The Gram-Schmidt Algorithm.} Given an orthonormal basis $\B$, the matrix vector Gram-Schmidt algorithm computes the orthonormal perpendicular component of an input operator $iT$. If the result is the null matrix, then $iT$ is inside the span of $\B$.}

\setlength{\intextsep}{\ointextsep}
\end{figure}

\begin{sdefi}[Hilbert-Schmidt Inner Product]
$ \langle A, B \rangle_{HS} = \t{Tr}\l( A^\dg B \r) $
\end{sdefi}

\noindent Note the trace is a linear map $\t{Tr}\l( A + B \r) = \t{Tr}\l( A \r) + \t{Tr}\l( B \r)$. Moreover, a matrix space can be defined as the span over some basis.

\begin{sdefi}[Matrix Basis and Span]
Given a collection of skew-Hermitian matrices $ \B = \l\{ i B_{1}, \ldots, i B_{|\B|} \r\},  i B_{j} \in \C^{2^n \times 2^n} $:
\begin{align}
\t{span}\l( \B \r) = \l\{ \sum_{j} \lambda_{j} \, i B_{j} : i B_{j} \in \B, \lambda_{j} \in \R \r\},
\end{align}
are the linear combinations of the basis terms. The $\t{span}$ of any basis of skew-Hermitian matrices is a matrix space. We use $\t{span} \equiv \t{span}_{\R}$, and specify the field of the span if it is ever different.
\end{sdefi}

\begin{sthm}
$\la iA, iB \ra_{HS} = \la A, B \ra_{HS} \in \R$ for $A,B$ Hermitian.
\end{sthm}

\begin{proof}
By definition $\t{Tr}\l( A \, B \r) \in \C$, then
\begin{align}
\overline{\t{Tr}\l( A \, B \r)} = \t{Tr}\l( \l( A \, B \r)^{\dg} \r) = \t{Tr}\l( B^\dg \, A^\dg \r) = \t{Tr}\l( B \, A \r) = \t{Tr}\l( A \, B \r),
\end{align}
and so the imaginary part must be zero.
\end{proof}

As a result, the set of skew-Hermitian matrices with the Hilbert-Schmidt inner product define a real valued vector space.

\noindent For the purpose of the manuscript, we only discuss matrix spaces, spans, and bases for skew-Hermitian operators. For example, we can use the skew Pauli group as the basis for skew-Hermitian matrices over $n$ qubits.

\begin{sthm}[Skew Pauli Basis]\label{sthm:skewpauli}
If $iA$ is skew-Hermitian and in $\C^{2^{n} \times 2^{n}}$, then $iA \in \t{span}\l( i \mathcal{P} \r)$.
\end{sthm}

Note that \cref{sthm:skewpauli} has $\t{span} \equiv \t{span}_{\R} $ as noted before. The matrix space of skew-Hermitian operators for $\C^{m \times m}$ also happens to form a \textit{Lie Algebra}, which we discuss in the remainder of this section.

\subsection{Superoperators over $\mf{u}\l(2^{n}\r)$}

We define several important superoperators for skew-Hermitian matrix spaces. Here $\mf{u}\l( 2^n \r)$ is the real vector space of skew-Hermitian matrices, and a superoperator on this space is a real linear map $\mf{u}\l( 2^n \r) \rightarrow \mf{u}\l( 2^n \r)$. Its restriction to a subspace $V$ is a superoperator on $V$ whenever it preserves $V$. We first define the identity and composition of superoperators, followed by the commutator.
\begin{sdefi}[Identity Superoperator]\label{sdef:identity_superoperator}
Given a real matrix space $V \subseteq \mf{u}\l( 2^n \r)$, the identity superoperator $\mc{I} : V \rightarrow V$ is defined by
\begin{align}
\mc{I}\l( iX \r) = iX, \qquad iX \in V .
\end{align}
\end{sdefi}
We use $\mc{I}$ for the identity superoperator on a matrix space and $\I$ for the identity operator on the Hilbert space. In particular, the identity superoperator is conjugation by the identity operator:
\begin{align}
\mc{I}\l( i \vr \r) = \I \, i \vr \, \I = i \vr .
\end{align}

\begin{sdefi}[Superoperator Outer Product]\label{sdef:oask}
Let $V$ be a real subspace of skew-Hermitian operators with the Hilbert--Schmidt inner product. For $iB, iC \in V$, define the superoperator outer product $iB \oask iC:V\to V$ by
\begin{equation}
\l(iB \oask iC\r)\l(iD\r)=iB \, \la iC,iD\ra_{\hs},
\qquad iD\in V.
\end{equation}
Its image lies in the span of $iB$, so it has rank one when both $iB$ and $iC$ are nonzero.
\end{sdefi}

\begin{slem}[Trace of a superoperator outer product]\label{slem:oasktrace}
For $iB, iC \in V$, the trace on $V$ satisfies
\begin{equation}
\t{Tr}_V\l(iB \oask iC\r)=\la iC,iB\ra_{\hs}.
\end{equation}
\end{slem}

\begin{proof}
Let $\l\{ iB_a \r\}_{a = 1}^{\t{dim}(V)}$ be a Hilbert--Schmidt orthonormal basis of $V$. Summing the diagonal entries in this basis gives
\begin{align}
\t{Tr}_V\l( iB \oask iC \r)
 &= \sum_a \la iB_a, \l( iB \oask iC \r)\l( iB_a \r)\ra_{\hs} \nonumber\\
 &= \sum_a \la iC, iB_a\ra_{\hs} \, \la iB_a, iB\ra_{\hs} \nonumber\\
 &= \la iC, \sum_a \la iB_a, iB\ra_{\hs} \, iB_a\ra_{\hs} \nonumber\\
 &= \la iC, iB\ra_{\hs}.
\end{align}
The final equality uses the expansion of $iB$ in the orthonormal basis.
\end{proof}

\begin{sdefi}[Superoperator Composition and Powers]\label{sdef:superoperator_composition}
Given superoperators $\mc{A},\mc{B} : V \rightarrow V$ and $a,b \in \R$, sums and products act as
\begin{align}
\l( a\mc{A} + b\mc{B} \r)\l( iX \r) &= a\mc{A}\l( iX \r) + b\mc{B}\l( iX \r), \\
\l( \mc{A}\mc{B} \r)\l( iX \r) &= \mc{A}\l( \mc{B}\l( iX \r) \r) .
\end{align}
Thus products denote composition, with the rightmost superoperator applied first. Powers are defined by $\mc{A}^{0} = \mc{I}$ and $\mc{A}^{k}\l( iX \r) = \mc{A}\l( \mc{A}^{k-1}\l( iX \r) \r)$ for $k \geq 1$.
\end{sdefi}

\begin{sdefi}[Commutator]\label{sdef:com}
For matrices $A,B \in \C^{m \times m}$,
\begin{align}
[A,B] = AB - BA .
\end{align}
For superoperators $\mc{A},\mc{B} : V \rightarrow V$ and $iX \in V$,
\begin{align}
[\mc{A},\mc{B}]\l( iX \r)
= \mc{A}\l( \mc{B}\l( iX \r) \r)
- \mc{B}\l( \mc{A}\l( iX \r) \r) .
\end{align}
Thus the commutator uses matrix multiplication for matrices and composition for superoperators. In particular, $[\mc{A},\mc{B}]=0$ means that the two compositions agree on every $iX \in V$.
\end{sdefi}

Essential for the study of skew-Hermitian dynamics is the adjoint superoperator, which is simply commutation with a specified skew-Hermitian operator.

\begin{sdefi}[Adjoint Superoperator]\label{sdef:adjop}
Given $iM, iA$ skew-Hermitian operators:
\begin{align}
\t{ad}_{iM} \l( iA \r) = \l[ iM, iA \r] = iM \, iA - iA \, iM .
\end{align}
\end{sdefi}

The adjoint superoperator has several important properties. It preserves skew Hermiticity.

\begin{sthm}[Commutation Preserves skew Hermiticity]
If $iM$ is skew-Hermitian and $iA$ is skew-Hermitian, then $\t{ad}_{iM}\l( i A \r) = \l[ iM, iA \r] $ is skew-Hermitian.
\end{sthm}

Commutation therefore maps the real space of skew-Hermitian operators to itself. The adjoint construction also preserves real linear combinations and commutators, as stated in the following lemma.

\begin{slem}[Linearity and Commutation of Adjoint Superoperators]\label{slem:adjoint_superoperator_algebra}
Given skew-Hermitian operators $iA,iB$ and $a,b \in \R$,
\begin{align}
\t{ad}_{a iA+b iB} &= a \, \t{ad}_{iA} + b \, \t{ad}_{iB}, \\
[\t{ad}_{iA},\t{ad}_{iB}] &= \t{ad}_{[iA,iB]} .
\end{align}
In particular, commuting operators have commuting adjoint superoperators.
\end{slem}
\begin{proof}
The first identity follows from (left) linearity of the commutator ($[ a \, i A + b \, i B, i X ] = [ a \, i A, i X ] + [ b \, i B, i X ] $). For the second identity, the Jacobi identity gives
\begin{align}
[ \t{ad}_{iA}, \t{ad}_{iB} ]\l( iX \r)
&= [ iA, [ iB, iX] ] - [ iB, [ iA, iX ] ] \\
&= [ [ iA,iB],iX ] = \t{ad}_{[iA,iB]}\l( iX \r) .
\end{align}
\end{proof}

Since $\t{ad}_{iM}$ is a superoperator over the matrix vector space with $ \t{ad}_{iM} \l( iA + iB \r) = \t{ad}_{iM} \l( iA \r) + \t{ad}_{iM} \l( iB \r) $, we can define powers and the exponential of the adjoint superoperator.

\begin{sdefi}[Exponentiation of the Adjoint Superoperator]
The standard power map for $k>1$ is
\begin{align}
\t{ad}_{iM}^{k} \l( iA \r) = \t{ad}_{iM}^{k-1} \l( \t{ad}_{iM} \l( iA \r) \r),
\end{align}
with $\t{ad}_{iM}^{0} = \mc{I}$, so $ \t{ad}_{iM}^{0}\l( iA \r) = iA $ and so the exponential map is
\begin{align}
e^{\theta \, \t{ad}_{iM}} \l( iA \r) = iA + \sum_{k=1}^{\infty} \frac{\theta^{k}}{k!} \, \t{ad}_{iM}^{k} \l( iA \r) .
\end{align}
\end{sdefi}

More generally, for any superoperator $\mc{A} : V \rightarrow V$ on a finite dimensional matrix space,
\begin{align}
e^{\theta \, \mc{A}} = \mc{I} + \sum_{k=1}^{\infty} \frac{\theta^k}{k!} \, \mc{A}^k .
\end{align}
The exponential identity for commuting linear operators holds on any finite dimensional real or complex vector space (Proposition 3.3(4) of Ref.~\cite{hall2000elementary}). Choosing a basis represents the operators by commuting matrices, so the same identity applies to superoperators on a matrix space.
\begin{slem}[Exponentials of Commuting Superoperators]\label{slem:comm_supop_exp}
Let $\mc{A},\mc{B} : V \rightarrow V$ be linear superoperators on a finite dimensional matrix space $V$. If $[\mc{A},\mc{B}]=0$, then for $a,b \in \R$,
\begin{align}
e^{a\mc{A}+b\mc{B}} = e^{a\mc{A}} e^{b\mc{B}} .
\end{align}
Equivalently, for every $iX \in V$,
\begin{align}
e^{a\mc{A}+b\mc{B}}\l( iX \r)
= e^{a\mc{A}}\l( e^{b\mc{B}}\l( iX \r) \r) .
\end{align}
\end{slem}
\begin{proof}
Expanding the exponential series and using the binomial theorem, which applies because $\mc{A}$ and $\mc{B}$ commute, gives
\begin{align}
e^{a\mc{A}+b\mc{B}}
&= \sum_{k=0}^{\infty} \frac{1}{k!} \l( a\mc{A}+b\mc{B} \r)^k \\
&= \sum_{k=0}^{\infty} \frac{1}{k!} \sum_{r=0}^{k}
\binom{k}{r} a^r b^{k-r} \, \mc{A}^r\mc{B}^{k-r} \\
&= \sum_{k=0}^{\infty} \sum_{r=0}^{k}
\frac{a^r b^{k-r}}{r!(k-r)!} \, \mc{A}^r\mc{B}^{k-r} \\
&= \sum_{r=0}^{\infty} \sum_{s=0}^{\infty}
\frac{a^r b^s}{r!s!} \, \mc{A}^r\mc{B}^s
= e^{a\mc{A}} e^{b\mc{B}} .
\end{align}
Here $s=k-r$, so each pair $r,s \geq 0$ corresponds uniquely to $k=r+s$. The two double sums therefore contain the same terms, and absolute convergence of the exponential series permits their rearrangement. Applying both sides to any $iX \in V$ gives the stated action.
\end{proof}

The representation of unitary conjugation through the exponential of the adjoint superoperator is standard in Lie-algebraic simulation~\cite{somma2005quantum,goh_lie-algebraic_2023,hall2013lie} and dates back to the turn of the 20\textsuperscript{th} century. Unitary evolution of an operator is equivalent to applying the exponential of the adjoint superoperator associated with the generator.

\begin{sthm}[Hadamard's Lemma, \cref{thm7:evoexpadj}]\label{sthm:evoexpadj}
Given skew-Hermitian operators $ iA $ and $ iB $:
\begin{align}
e^{i \theta A} \, i B \, e^{-i \theta A} = e^{\theta \, \t{ad}_{i A}} \l( i B \r) .
\end{align}
\end{sthm}

\begin{proof}
Let $ f(t) := e^{i t A}\, iB \,e^{-i t A} $. Then $ f(0) = iB $. Differentiating with respect to $ t $,
\begin{align}
\frac{d}{dt} f(t) &= i A \, e^{i t A} \, i B \, e^{-i t A} - e^{i t A} \, i B \, e^{-i t A} i A \\
&= \l[ i A, e^{i t A} \, i B \, e^{-i t A} \r] \\
&= \t{ad}_{iA}\l( f(t) \r).
\end{align}
Given the initial condition $f(0) = i B$, the unique solution is:
\begin{align}
f(t) = e^{t \, \t{ad}_{iA}}(iB).
\end{align}
Evaluating at $ t = \theta $ yields
\begin{align}
e^{i \theta A} \, iB \, e^{-i \theta A} = e^{\theta \, \t{ad}_{iA}}(iB) .
\end{align}
\end{proof}

Given matrix space of skew-Hermitian operators $V$ and a matrix subspace $W \subset V$, we define the projection into this subspace.

\begin{sdefi}[Projection Superoperator]\label{def:proj}
Given matrix space $V = \t{span}\l( \B^V \r)$ and matrix subspace $W = \t{span}\l( \B^W \r) \subset V$, we can define the projection superoperator $\t{proj} : V \rightarrow V$
\begin{align}
\t{proj}_{W}\l( A \r) = \sum_{iB \in \B^W} \la iB , iA \ra_{\hs} \, iB
\end{align}
\end{sdefi}

In superoperator notation, the complementary projection is $\t{proj}_{W}^{\perp} = \mc{I} - \t{proj}_{W}$, where $\mc{I}$ acts on the ambient matrix space $V$.

\noindent The projection $\t{proj}$ is a linear map such that $\t{proj}\l( A + B \r) = \t{proj}\l( A \r) + \t{proj}\l( B \r) $. We can define the perpendicular projection as $\t{proj}_{W}^{\perp}\l( A \r) = A - \t{proj}_{W}\l( A \r)$.

\begin{suppfigure}[!t]
\centering
\includegraphics[width=0.6\textwidth]{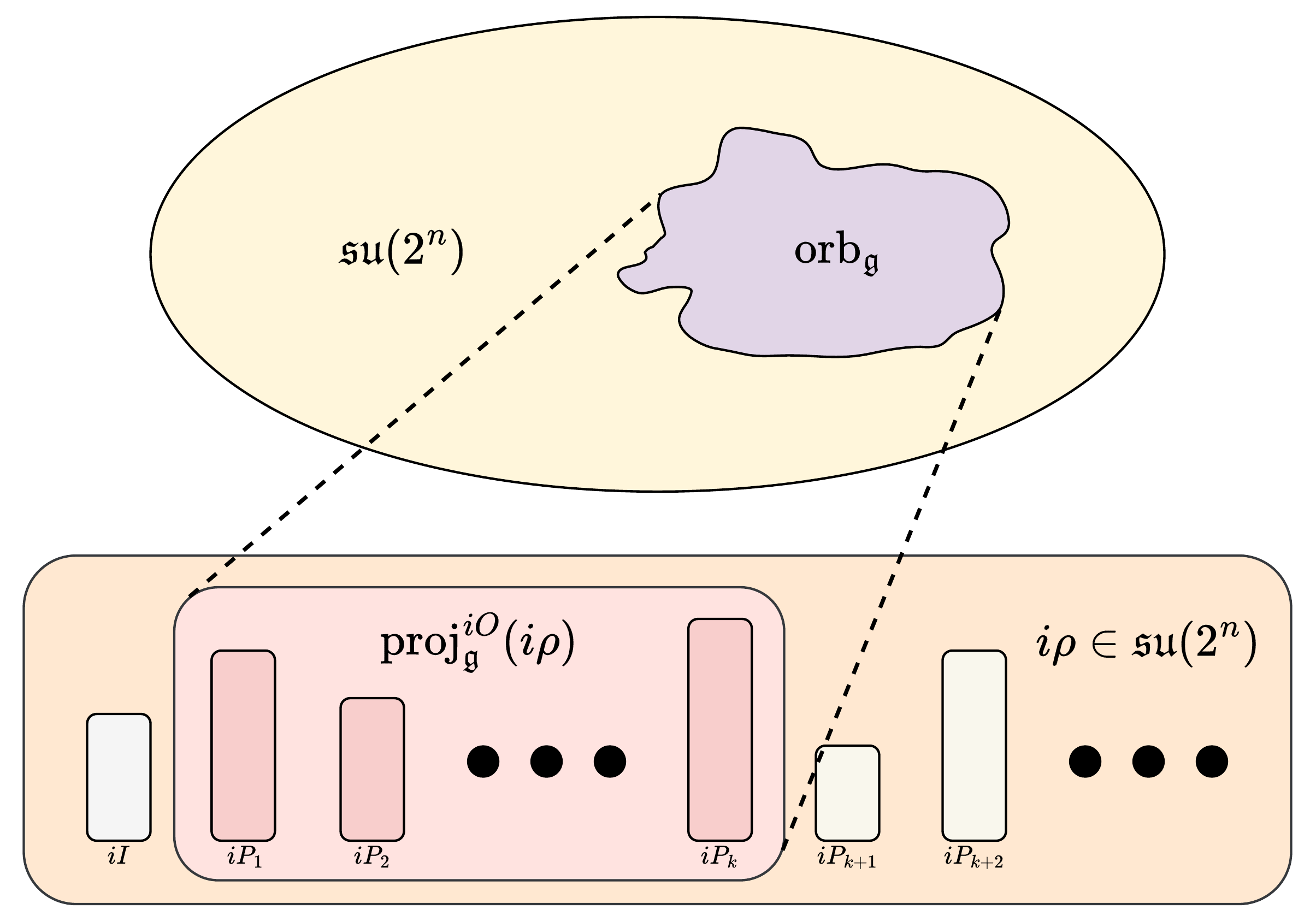}
 \caption{\tsbf{Observable Subspace of Quantum Dynamics} (reprint of \cref{fig:img_dynsub}). Given an observable $iO$ and a dynamical Lie algebra $\g$, nested commutation with elements of $\g$ generates the Dynamic Observable Subspace $\orb$, the minimal invariant operator subspace containing $iO$. The expectation value depends only on the projection of the density operator into this subspace, such that $\t{Tr}\l( O \, \rho(T) \r) = \la iO, \t{proj}_{\g}^{iO}(i\rho(T)) \ra_{\dos}$. In the special case $iO\in\g$, it follows that the DOS is contained within the DLA, $\orb\subseteq\g$.}
\label{suppfig:img_dynsub}
\end{suppfigure}

\section{The Dynamic Lie Algebra of Variational Quantum Algorithms and its Construction}

In this section we discuss Lie algebras, the Dynamic Lie Algebra (DLA) as discussed in the main manuscript, and its construction as well as fundamental features.

Lie algebras are an important study in both classical and quantum control theory. Fundamentally, rotations over real and complex vector spaces (such as Hilbert spaces) are compact matrix Lie groups as noted in the main manuscript. Lie algebras describe the tangent space or the infinitesimal change of these rotation groups.

\begin{sdefi}[Lie Algebra]
A Lie algebra $\g$ of skew-Hermitian operators is a matrix vector space over $\R$ with the standard commutator, \cref{sdef:com}, as a bilinear operator defined over it. Specifically, the standard commutator satisfies
\begin{enumerate}[label=(LA\arabic*), left=1.0em, labelsep=1.0em, itemsep=0.5em]
    \item[] Skew symmetry: $[a,b] = -[b,a]$ for all $a,b \in \g$.
    \item[] Jacobi identity: $[a,[b,c]] = [[a,b],c] + [b,[a,c]]$ for all $a,b,c \in \mf{g}$.
\end{enumerate}
\end{sdefi}

\begin{sdefi}[General Linear Lie Algebra]
The \textit{general linear Lie algebra}, $\mf{gl}\l( m \r)$, is $\C^{m \times m}$ with the standard commutator. The standard basis for $\mf{gl}\l( m \r)$ is
\begin{align}
    e_{j,k} := \ketbra{j}{k} ,
\end{align}
for $j,k \in [m]$ such that
\begin{align}
\mf{gl}(m) = \t{span}_{\C} \l( \{ e_{j,k} \}_{j,k}^{m} \r) .
\end{align}
\end{sdefi}

Then $\t{dim}\l( \mf{gl}(m) \r) = m^2$. In fact $\mf{gl}\l( m \r)$ is $\C^{m \times m}$ with the additional structure of the commutator as a bilinear operator over its elements.

\begin{figure}[!t]
\centering
\includegraphics[width=0.95\textwidth]{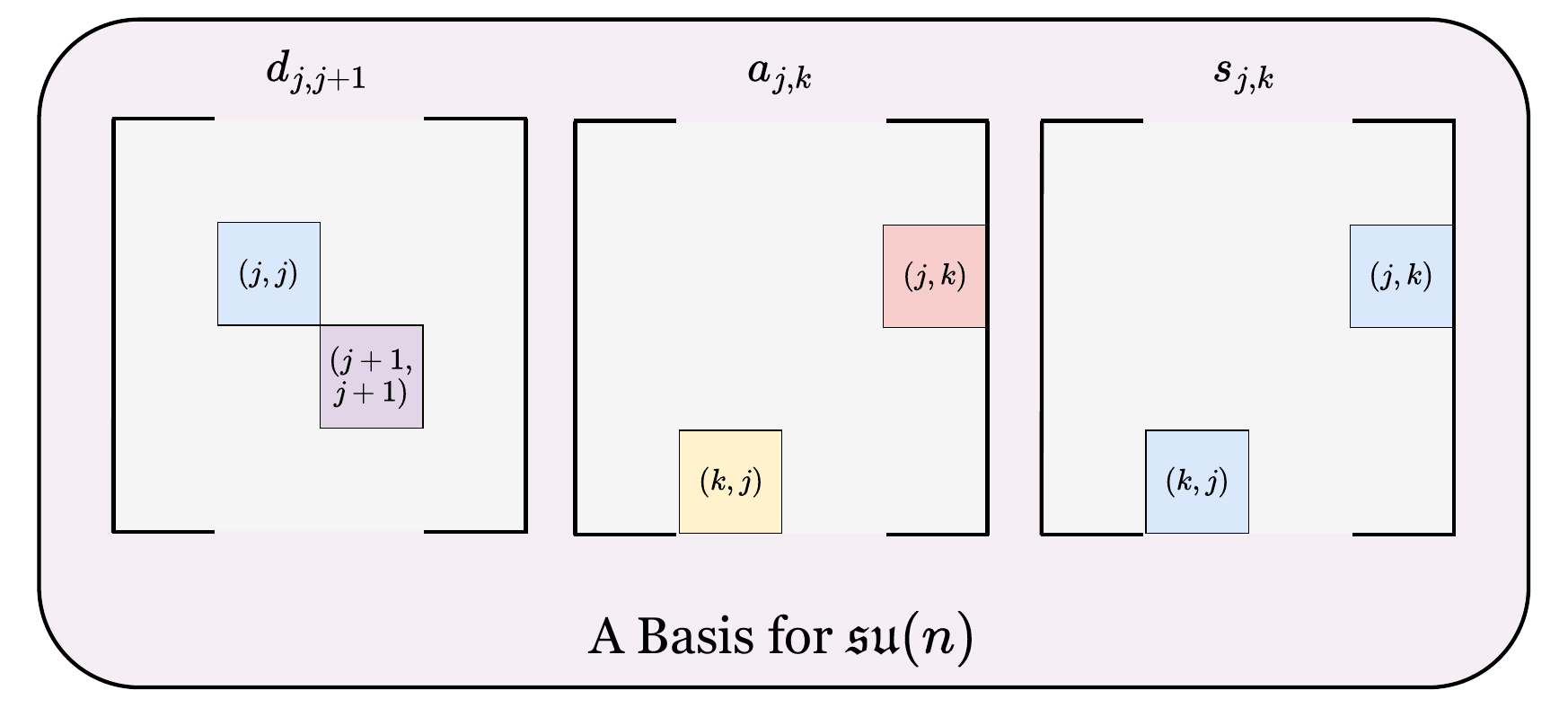}
 \caption{\tsbf{A Basis for $\mf{su}\l( n \r)$.} Here we visually represent the basis elements for $\mf{su}$. The terms $a_{j,k}$ form the canonical basis for $\mf{so}$.}\label{fig:subasis}
\end{figure}

The Lie algebras of skew-Hermitian matrices generates the possible unitary operators over a wavefunction state.

\begin{sdefi}[Unitary and Special Unitary Lie Algebra]\label{sdef:su}
The unitary Lie algebra is defined
\begin{align}
\mf{u}\l( m \r) := \l\{ iH \in \C^{m \times m} : H = H^{\dg} \r\} .
\end{align}
The special unitary Lie algebra is defined
\begin{align}
\mf{su}\l( m \r) := \l\{ iH \in \mf{u}\l( m \r) : \t{Tr}\l( iH \r) = 0 \r\} .
\end{align}

Then $\mf{u}\l( m \r) = \t{span}\l( i \I \r) \oplus \mf{su}\l( m \r)$ and $\t{dim}\l( \mf{su} \r) = \t{dim}\l( \mf{u} \r) - 1$. In particular, $\mf{u}$ has a simple standard basis defined through elements
\begin{align}
a_{j,k} &:= \ketbra{j}{k} - \ketbra{k}{j} , \\
s_{j,k} &:= i \, \ketbra{j}{k} + i \, \ketbra{k}{j} .
\end{align}
Then $ \mf{u}\l( m \r) = \t{span}_{\R}\l( \l\{ a_{j,k} , s_{j,k} \r\}_{j<k}^{m} \cup \l\{ i \, e_{j,j} \r\}_{j=1}^{m} \r) $ which leads to $\t{dim}\l( \mf{u} \r) = \frac{1}{2} n (n+1) $. For $\mf{su}\l( m \r)$, we need exclude $iI = \sum_{j=1}^{m} i \, e_{j,j} $. Different bases exist, such as the Pauli basis
\begin{align}
\mf{su}\l( 2^n \r) = \t{span}_{\R}\l( i \P \r) ,
\end{align}
when $m = 2^{n}$. Another option for any $m$ is to use $d_{j,j+1} = i e_{j,j} - i e_{j+1,j+1}$ since the trace inner product with $iI$ is $0$ such that
\begin{align}
\mf{su}\l( m \r) = \t{span}_{\R}\l( \{ d_{j,j+1} \}_{j=1}^{m-1} \cup \{ a_{j,k} \}_{j<k}^{m} \cup \{ s_{j,k} \}_{j<k}^{m} \r) ,
\end{align}
although this basis is not orthonormal. Fig.~\cref{fig:subasis} depicts what these skew-Hermitian basis elements look like. A standard orthonormal basis is the skew Gell Mann basis.
\end{sdefi}

Two elements $a, b \in \g$ \textit{commute} if $[a,b] = 0$. By skew symmetry, every element commutes with itself $[a,a] = a^2 - a^2 = 0$. A Lie algebra $\g$ is \textit{abelian} if and only if every element commutes with one another.
\begin{sdefi}[Center of Lie Algebra]
The \textit{center} of $\g$, denoted $Z(\g)$, is the set of operators that commute with every other element:
\begin{align}
    Z(\g) := \{g \in \g : [g,h] = 0\ \forall \, h \in \g\} .
\end{align}
\end{sdefi}

\begin{sdefi}[Lie subalgebra]
A Lie subalgebra $\mf{h} \leq \mf{g}$ is a linear subspace closed under the commutator: $[h_1, h_2] \in \mf{h}$ for any $h_1, h_2 \in \mathfrak{h}$. A Lie subalgebra $\mf{h}$ is an \textit{ideal} since $[h,g] \in \mf{h}$ for any $h \in \mf{h}$ and $g \in \mf{g}$.
\end{sdefi}

The center $Z(\mf{g})$ is always an ideal since for any $g \in \mf{g}$ and $h \in Z(\mf{g})$, $[h,g] = 0 \in Z(\mf{g})$. An ideal is said to be \textit{trivial} if it is the empty span $\{0\}$ or the full Lie algebra $\mf{g}$.

\noindent Many Lie algebras can be decomposed into sums of smaller Lie subalgebras. $\mf{g}$ is a direct sum of Lie subalgebras $\mf{g}_1, \dots, \mf{g}_n$, written as $\mf{g} = \mf{g}_1 \oplus \cdots \mf{g}_n$, if $\mf{g}$ is the direct sum of the $\mf{g}_i$ as a vector space and $\l[ \mf{g}_i, \mf{g}_j \r] = \{ 0 \}$ for all $i \neq j$. A Lie algebra $\mf{g}$ is \textit{simple} if it is nonabelian and contains no nontrivial ideals. $\mf{g}$ is \textit{semisimple} if it is a direct sum of simple Lie subalgebras. Many Lie algebras of interest are semisimple.

We state the structure theorem for compact Lie algebras relevant to quantum dynamics, although we will not formally classify any Lie algebras in this paper. In this setting, the dynamical Lie algebra is isomorphic to a direct sum of an abelian factor $\mf{u}\l( 1 \r)^p$ together with classical Lie algebras of type $\mf{so}, \mf{su}, \mf{sp}$ of various dimensions, assuming no appearance of exceptional Lie algebras. For most realistic systems, such exceptional algebras do not arise, though the theorem can be readily extended to include them. We refer to this decomposition as the \emph{SOUP theorem} as a mnemonic for $\mf{so}, \mf{su}, \mf{sp}$.

\begin{sthm}[SOUP theorem]
Let $\mf{g} \subseteq \mf{u}\l( m \r)$ be a Lie subalgebra. Then $\mf{g}$ is
reductive and decomposes as
\begin{align}
\mf{g} = Z\l( \mf{g} \r) \oplus \mf{g}_1 \oplus \cdots \oplus \mf{g}_k,
\end{align}
where $Z\l( \mf{g} \r)$ is the center and each $\mf{g}_i$ is a compact simple
Lie algebra.

If no \textit{exceptional} Lie algebra appears among the $\mf{g}_i$, then
\begin{align}
\mf{g} \cong \mf{u}\l( 1 \r)^p
 \oplus \l( \bigoplus_{j} \mf{so}\l( f_j \r) \r)
 \oplus \l( \bigoplus_{k} \mf{su}\l( h_k \r) \r)
 \oplus \l( \bigoplus_{l} \mf{sp}\l( g_l \r) \r),
\end{align}
such that
\begin{align}
\t{dim}\l( \mf{g} \r)
= p
+ \l( \sum_{j} \frac{f_j\l( f_j-1 \r)}{2} \r)
+ \l( \sum_{k} h_k^2 - 1 \r)
+ \l( \sum_{l} g_l\l( 2g_l+1 \r) \r).
\end{align}
\end{sthm}

\noindent This decomposition theorem allows for the classification of Lie algebras defined over skew-Hermitian operators as the direct sum of known subalgebras.

The relevant DLA classifications for this manuscript are listed below. The first line is for the open $XX,YY,Z$ chain.
\begin{align}
\la \mc{G}_{XX,YY,Z} \ra_{\t{Lie}} &\cong \mf{so}(2n) \label{eq:opendla}\\
\la \mc{G}_{XY,Z} \ra_{\t{Lie}} &\cong \mf{u}(1) \oplus \mf{su}(n) \oplus \mf{su}(n) \\
\la \mc{G}_{XX,YY,ZZ,Z} \ra_{\t{Lie}} &\cong \mf{su}\l( 2^n \r) \\
\la \mc{G}_{XY,ZZ,Z} \ra_{\t{Lie}} &\cong \mf{u}(1)^{ \oplus n} \oplus \bigoplus_{k=1}^{n-1} \mf{su}\l( \binom{n}{k} \r)
\end{align}
The corresponding dimensions are
\begin{align}
\t{dim}\l( \la \mc{G}_{XX,YY,Z} \ra_{\t{Lie}} \r) &= n \, (2n-1) \\
\t{dim}\l( \la \mc{G}_{XY,Z} \ra_{\t{Lie}} \r) &= 2 \, n^2 - 1 \\
\t{dim}\l( \la \mc{G}_{XX,YY,ZZ,Z} \ra_{\t{Lie}} \r) &= 4^n - 1 \\
\t{dim}\l( \la \mc{G}_{XY,ZZ,Z} \ra_{\t{Lie}} \r) &= \binom{2n}{n} - 1 .
\end{align}
The final equality is conditional on the stated conjectural classification.

\section{Skew Quantum Dynamics in the Schr{\"o}dinger and Heisenberg Pictures}

The Schr{\"o}dinger and Heisenberg pictures are two equivalent formulations of closed quantum dynamics. In the Schr{\"o}dinger picture, the state $i\rho$ evolves in time while the observable $iO$ remains fixed, while in the Heisenberg picture the observable evolves and the initial state remain fixed. If there are multiple observables, they evolve independently. In the density matrix formulation, the Schr{\"o}dinger equation becomes the von Neumann equation.

\begin{sdefi}[Skew von Neumann Equation]
Closed system infinitesimal quantum dynamics of a density operator $\rho$ under a Hamiltonian $H(t)$ are determined at time $t$ by
\begin{align}
\frac{\partial \l( i \rho \r) }{\partial t} &=  - \l[ i H, i \rho \r] .
\end{align}
\end{sdefi}

In both the wavefunction or density matrix formulation, the dynamics of the observable $iO$ are given by the Heisenberg equation.

\begin{sdefi}[Skew Heisenberg Equation]
Closed system infinitesimal quantum dynamics of an observable $O$ under a Hamiltonian $H(s)$ are determined at time $s$ by
\begin{align}
\frac{\partial \l( i O \r) }{\partial s} &= \l[ i H, i O \r] .
\end{align}
\end{sdefi}

The two equations are equivalent by cyclicity of the trace, since the expectation value
\begin{align}
\t{Tr}\l( O\,\rho(t) \r) = -\t{Tr}\l( iO \, i\rho(t) \r)
\end{align}
may equivalently be computed by evolving $i\rho$ forward in time (Schr{\"o}dinger picture) or $iO$ backward in time (Heisenberg picture).

For a time independent Hamiltonian $H$, both equations admit closed form solutions via the unitary $U(t) = e^{-iHt}$:
\begin{align}
i\rho(t) &= U(t)\, i\rho(0)\, U^\dagger(t) , \\
iO(t) &= U^\dagger(t)\, iO(0)\, U(t) .
\end{align}

This bridges the differential picture to the matrix multiplication picture used in the main text. In the Schr{\"o}dinger picture each gate $\t{R}^G(\theta) = e^{i\theta G}$ steps the state \emph{forward} through the circuit
\begin{align}
    i\rho \to \t{R}^G(\theta)\,i\rho\,\t{R}^G(\theta)^\dagger.
\end{align}

In the Heisenberg picture the observable is instead stepped \emph{backward} through the circuit using the adjoint gate,
\begin{align}
iO \to \t{R}^G(\theta)^\dagger\,iO\,\t{R}^G(\theta),
\end{align}
to arrive at an effective observable to be evaluated against the initial state $i\rho(0)$. For a time dependent Hamiltonian $H(t)$, the unitary is instead the time ordered exponential $U(t) = \mathcal{T}\exp\!\l( -i\int_0^t H(s)\,ds \r)$; the VQA circuit recovers this as a product of piecewise constant steps, each of the simple matrix exponential form above.

\begin{remark}
The orbit of the density operator and the observables can lead to more restricted dynamic observable subspaces. It is perhaps more clear that there are valuable situations where the observable has a more restricted subspace than the density operator. However there are highly valuable such subspaces from the perspective of the density operator.\footnote{For example, the DLA with cardinality constraint symmetry can be seen as such a case.}
\end{remark}

\begin{remark}
Suppose the initial density operator is pure and unentangled, then it is uniquely defined by the tensor product of the Bloch sphere which involves an exponential number of Pauli strings. As such, the DOS of $i\rho(0)$ has dimension exponential in $n$. Yet if $iO$ is a local observable, its DOS may remain of polynomial dimension, which is precisely the scenario where building the orbitals of each skew Heisenberg observable yields an exponential advantage over building the orbital of the density operator.
\end{remark}

\section{Dynamic Lie Algebra of Variational Quantum Algorithms}

In the case of a variational quantum algorithm, we are given a set of generators $\G = \{ g_{1}, \ldots, g_{k} \}$ over $n$ qubits such that each generator is skew-Hermitian $ig_{j} \in \mf{su}\l( 2^n \r) $ and our circuit is defined by applying a gate $R^g(\theta) = e^{i \theta g}$ generated by $g$ and parameterized by $\theta \in \R$. Then by repeating the $K$ generators for $L$ layers yields

\begin{align}
U(\bmt) = \prod_{\ell=1}^{L} \prod_{k=1}^{K} e^{i \theta_{k} g_{k}} ,
\end{align}
with $\bmt \in \R^{LK}$. Given a skew-pictured initial density operator $i\rho \in \mf{su}\l(2^n\r)$ with $\Tr(\rho)=1$ and a skew-pictured observable $iO \in \mf{su}\l(2^n\r)$, the expectation value of $O$ after applying $U(\bmt)$ is
\begin{align}
\mathcal{L}(\bmt)
&= \t{Tr}\l( O \, U(\bmt) \, \rho \, U^\dg(\bmt) \r) \\
&= -\t{Tr}\l( iO \, U(\bmt) \, i \rho \, U^\dg(\bmt) \r).
\end{align}
Thus, the dynamics relevant to evaluating $\mathcal{L}(\bmt)$ are determined by the action of the circuit on the observable and density operator.

\begin{sdefi}[Dynamic Lie Algebra]\label{sdef:dla}
Given a collection of Hermitian generators $\G$, the Dynamic Lie Algebra $\g$ (DLA) is the matrix span closed under nested commutation:
\begin{align}
\g = \t{span}_{\R}\l( \B^{\g} \r),
\end{align}
where $ \B^\g = \l\{ iB_{1}^\g, \ldots, iB_{\t{dim}\l( \g \r)}^\g \r\} $ is the orthonormal basis constructed via recursive commutation with each element of $\G$ starting with $B_{1}^{\g} = i\G$:
\begin{align}
\Delta_{j+1}^\g &= \t{Gram-Schmidt}\l( \l\{ \l[ iG, iB \r] : G \in \G, iB^\g \in \B_{j}^\g \r\} \r), \nonumber \\
\B_{j+1}^\g &=  \Delta_{j+1}^\g \cup \B_{j}^\g ,
\end{align}
until no new terms are generated with $ \B^\g = \B_{k}^\g = \B_{k+1}^\g$ for nesting depth $k$.
\end{sdefi}

\begin{sdefi}[Lie Group]\label{sdef:liegroup}
The Lie group $e^{\g}$ of a Lie algebra $\g$ is the set of unitaries generated by $\g$:
\begin{align}
e^{\g} = \l\{ \prod_{j=1}^{J} e^{i G_{j}} : iG_{j} \in \g, J \in \N \r\} .
\end{align}
Since $iG \in \g$, we can find $\omega \in \R^{\t{dim}\l( \g \r)}$ such that $iG = \sum_{j=1}^{\t{dim}\l( g \r)} w_{j} \, i B_{j} $.
\end{sdefi}

\section{The Dynamic Observable Subspace of Variational Quantum Algorithms and its Construction}

\begin{suppfigure}[!t]
\centering
\includegraphics[width=0.7\textwidth]{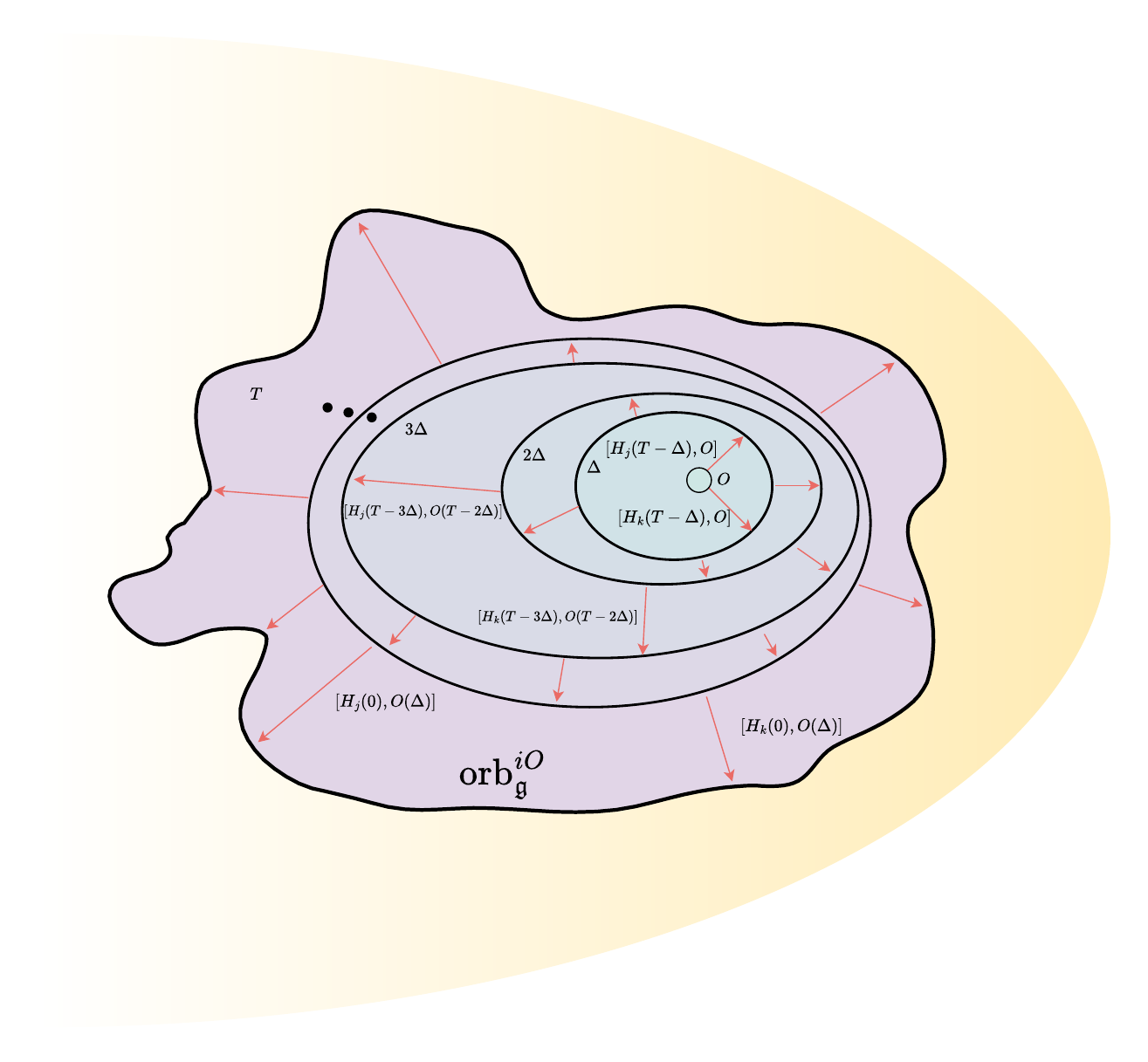}
\caption{\tsbf{Liouville-von Neumann Dynamics in the Heisenberg Picture} (reprint of \cref{fig:lvn_heisenberg}). In the skew Heisenberg picture, an observable $i O$ evolves by a Hamiltonian $H(t)=\sum_k H_k(t)$ in reverse time from $ T $ to $ 0 $ to an operator state $i O(0)$ contained inside $\t{orb}_{\g}^{iO}$. Then the measurement of $i O$ is associated with $\Tr{i O(0) \, \t{proj}_{\g}^{iO}(i \rho(0))} $ in the Heisenberg picture (or $\Tr{i O \, \t{proj}_{\g}^{iO}(i \rho(T))}$ in the Schr{\"o}dinger picture).}
\label{suppfig:lvn_heisenberg}
\end{suppfigure}

As established in the main text, the VQA circuit, Schr{\"o}dinger picture state $i\rho(t) = U_{1:t}(\bmt)\,i\rho(0)\,U_{1:t}^\dg(\bmt)$, and Heisenberg picture observable $iO(t) = U_{t+1:L}^\dg(\bmt)\,iO\,U_{t+1:L}(\bmt)$ are defined as in \cref{def:vqa,def:schro,def:hei}. In particular,
\begin{align}
\t{Tr}\l( \l( O \r)^\dg \rho(LK) \r) &= \t{Tr}(O\,\rho(LK)) \qquad \t{(hermiticity)} \\
&= \t{Tr}((-i)O\,i\rho(LK)) \\
&= \t{Tr}( \l( iO \r)^{\dg} \,i\rho(LK)).
\end{align}
We restate the central definition of the Dynamical Observable Subspace (DOS) and develop its properties rigorously.

\begin{sdefi}[Dynamic Observable Subspace]\label{sdef:dos}
Given a collection of Hermitian generators $\G$ and an observable $O$, the Dynamic Observable Subspace (DOS) is the smallest real matrix subspace containing $iO$ and invariant under commutation with each $iG$ for $G \in \G$:
\begin{align}
\t{orb}_\g^{iO} = \t{span}_{\R}\l( \B \r)
\end{align}
where $ \B = \l\{ iB_{1}, \ldots, iB_{\t{dim}\l( \t{orb}_\g^{iO} \r) } \r\} $ is the orthonormal basis constructed via recursive nested commutation with each element of $\G$ starting with $\B_{1} = \{ iO \}$:
\begin{align}
\B_{k+1} = \t{Gram-Schmidt}\l( \l\{ \l[ iG, iB \r] : G \in \G, iB \in \B_{k} \r\} \cup \B_{k} \r),
\end{align}
until no new terms are generated with $ \B = \B_{k} = \B_{k+1} $ for some terminal $k$.
\end{sdefi}

By its construction, the DOS is a matrix subspace of the skew-Hermitian matrices over $n$ qubits.

\begin{slem}\label{slem:orbsubspace}
$\orb$ is a real linear subspace of $\mf{u}\l( 2^n \r)$. If $iA \in \orb$, then $\t{ad}_{iH}(iA) \in \orb$ for $iG \in \t{span}\l( i\G\r)$.
\end{slem}

\begin{proof}
Since $iO$ and each $iG \in i \G$ is skew-Hermitian and this property is preserved under commutation, element of $\orb$ is skew-Hermitian, so $\orb \subseteq \mf{u}\l( 2^n \r)$. Since $\orb$ is defined as a real span it is a real linear subspace, although in general it is a not a Lie algebra. Specifically, it is not necessarily closed under commutation with $iO$, even though $iO \in \orb$.

Let $iA = \sum_{j=1}^{|\B|} i w_{j} \, B_{j}$, then
\begin{align}
\t{ad}_{iG}\l( iA \r) &= \sum_{j} w_{j} \t{ad}_{iG}\l( i B \r) \\
&\in \orb,
\end{align}
since every basis term $iB \in \orb$ was expanded if $[iA,iB] \notin \B^{k}$ for some $k$ (the space is closed under commutation with operators $i\G$).
\end{proof}

Since $\orb$ is a real linear subspace of $\mf{u}(2^n)$, we can define its orthogonal complement, which is a real linear subspace as well.

\begin{sdefi}[Orthogonal Complement of the DOS]\label{sdef:orthdos}
The orthogonal complement of the DOS $\orb$ with respect to the Hilbert-Schmidt trace inner product is
\begin{align}
\orbp := \l\{ iC \in \mf{u}\l( 2^n \r) : \la iC, iB \ra_{\hs} = 0 \;\; \forall\, iB \in \orb \r\}.
\end{align}
\end{sdefi}

Then as in the main manuscript, the relevant inner product is over the DOS orbit.

\begin{sdefi}[Inner Product over Dynamic Observable Subspace]\label{sdef:dostrace}
Given a DOS $\t{orb}_{\mf{g}}^{iO}$ as \cref{def:dos} with an orthonormal basis $\B$, define the inner product over this matrix subspace of $\C^{2^{n} \times 2^{n}}$ as:
\begin{align}
\la iY,  iX \ra_{\t{DOS}} &= \la \t{proj}_{\g}^{iO}\l( i Y \r) , \t{proj}_\g^{iO}\l( i X \r) \ra_{HS} \\
&= \sum_{iB \in \B} \t{Tr}\l( \l( i B \r)^{\dg} iY \r) \, \t{Tr}\l( \l( i B \r)^{\dg} iX \r).
\end{align}
\end{sdefi}

With the DOS and its inner product in hand, we can now show that the Heisenberg picture observable $iO(t)$ remains inside the DOS for all $t$.

\begin{sthm}[DOS representation of Observable, \cref{thm1:dosobs}]\label{sthm:dosobs}
\begin{align}
iO(t) \in \t{orb}_{\mf{g}}^{iO} .
\end{align}
\end{sthm}

\begin{proof}
Since $iO \in \B$, it follows immediately that $iO(LK) \in \orb \equiv \t{span}\l( \B \r)$. Suppose $iO(t) \in \orb$, then we show that $iO(t-1) \in \orb$ as well.
\begin{align}
iO(t-1) &= U_{t-1:t}^{\dg}(\bmt) \, iO(t) \, U_{t-1:t}(\bmt) \\
&= e^{i \theta_{t} G_{t}} \, iO(t) \, e^{-i \theta_{t} G_{t}} \\
&= e^{\theta \, \t{ad}_{i G_{t}}} \l( iO(t) \r)  \quad \l( \t{from \cref{sthm:evoexpadj}} \r) \\
&= iO(t) + \sum_{k=1}^{\infty} \frac{\theta^{k}}{k!} \t{ad}_{i G_{t}}^{k} \l( iO(t) \r) \in \orb ,
\end{align}
where the last line follows from $\t{ad}_{i G_{t}}\l( i A \r) \in \orb $ for any $iA \in \orb$ (by \cref{slem:orbsubspace}).
\end{proof}

Since $iO(t)$ stays inside the DOS, the expectation value depends only on the projection of $i\rho(0)$ onto the DOS, as the Heisenberg picture demands. As such, the component of $i\rho(0)$ outside the DOS (in the complement) contributes nothing.

\begin{sthm}[DOS Heisenberg, \cref{thm2:heidos}]
Given an observable $iO$ and a set of generators $\G$ such that $\mf{g} = \langle i \G \rangle_{\t{Lie}} $, the Dynamic Observable Subspace $\t{orb}_{\mf{g}}^{iO} $ of $iO$ under $\mf{g}$ satisfies
\begin{align}
\t{Tr}\l( \rho(0) \, O(0) \r) = \la i \vr(0) , i O(0) \ra_{\dos}
\end{align}
where $\t{proj}_{\mf{g}}^{iO}\l( i \rho(0) \r) = i \vr(t)$ is the projection of $i\rho(0)$ into $\t{orb}_{\mf{g}}^{iO}$.
\end{sthm}

\begin{proof}
From \cref{sthm:dosobs}, we know $iO(0) \in \orb$. Decompose $i \rho$ into its components in $\orb$ and $\orbp$:
\begin{align}
i\rho = \t{proj}_{\g}^{iO}\l( i\rho \r) + \t{proj}_{\g}^{iO,\perp}\l( i\rho \r) .
\end{align}

Then
\begin{align}
\t{Tr}\l( \rho(0) O(0) \r) &= -\t{Tr}\l( \l( i \rho(0) \r)^\dg \, i O(0) \r) \\
&= \t{Tr}\l( \l( \t{proj}_{\g}^{iO} \l( i \rho(0) \r) + \t{proj}_{\g}^{iO, \perp} \l( i \rho(0) \r) \r)^{\dg} \, iO(0) \r) \\
&= \t{Tr}\l( \t{proj}_{\g}^{iO} \l( i \rho(0) \r)^{\dg} \, i O(0) \r) \qquad (\t{from \cref{sthm:dosobs}}) \\
&= \la \t{proj}_{\g}^{iO} \l( i \rho(0) \r) , i O(0) \ra_{\dos} .
\end{align}
\end{proof}

\noindent To reason about dynamics inside the DOS, we prove the representation theorem of the main manuscript.

\begin{sthm}[DOS Representation Theorem, \cref{thm4:dosrep}]\label{sthm:dosrep}
Let $iA \in \orb $. Then for any skew-Hermitian operator in the DLA $iH \in \g$, the commutator is inside the DOS:
\begin{align}
\t{ad}_{iH}\l( i A \r) = \l[ iH, iA \r] \in \orb .
\end{align}
Moreover, if $iC \in \orbp$, then $\t{ad}_{iH}\l( iC \r) \in \orbp$.

For any unitary in the Lie Group generated by the DLA $V \in e^{\g}$, elements of the DOS remain is the DOS under unitary evolution of $V$:
\begin{align}
V \, iA \, V^{\dg} \in \orb .
\end{align}

Moreover, if $iC \in \orbp$, then $V \, iC \, V^\dg \in \orbp$.
\end{sthm}

\begin{proof}
Let
\begin{align}
\mf{h} := \l\{ i D : \t{ad}_{iD}(\orb) \subseteq \orb \r\}.
\end{align}

By linearity of $\t{ad}$ in its label, $\mf{h}$ is a real linear subspace. Moreover,
if $iD_1,iD_2 \in \mf{h}$, then
\begin{align}
\t{ad}_{[iD_1,iD_2]} = [\t{ad}_{iD_1},\t{ad}_{iD_2}],
\end{align}
so $\t{ad}_{[iD_1,iD_2]}(\orb) \subseteq \orb$, since both $\t{ad}_{iD_1}$ and $\t{ad}_{iD_2}$ preserve $\orb$. Hence $\mf{h}$ is a Lie subalgebra.

By definition of the DOS, $\t{ad}_{iG}(\orb) \subseteq \orb$ for every generator $iG \in i\G$, so $i\G \subseteq \mf{h}$. Since $\g= \la i \G \ra_{\t{lie}}$, it follows that $\g \subseteq \mf{h}$. Therefore for any $iD \in \g$ and $iC \in \orb$,
\begin{align}
\l[ iD, iC \r] = \t{ad}_{iD}(iC) \in \orb ,
\end{align}
which proves the first statement of the theorem for $\orb$.

Now consider $iC \in \orbp$. From definition, for any $iB \in \orb$, we have that $\la iC, iB \ra_{\hs} = 0$. Then
\begin{align}
\la [iH, iC], iB \ra_{\hs} &= \t{Tr} \l( \l[ iH , iC \r]^{\dg} iB \r) \\
&= -\t{Tr}\l( \l[ iH, iC \r] \, iB \r) \\
&= -\t{Tr}\l( iH \, iC \, iB - iC \, iH \, iB \r) \\
&= -\t{Tr}\l( iH \, iC \, iB \r) + \t{Tr}\l( iC \, iH \, iB \r) \\
&= -\t{Tr}\l( iC \, iB \, iH - iC \, iH \, iB \r) \\
&= -\t{Tr}\l( iC \, \l[ iB, iH \r] \r) \\
&= 0,
\end{align}
where the last line follows from $\l[ iB, iH \r] \in \orb$.

We now prove the second statement of the theorem. Since $V \in e^{\g} $, then $V = \prod_{j=1}^{J} e^{i H_{j}} $ with $i H_{j} \in \g$ for some integer $J$. If $V$ is sufficiently close to identity, then there is a single such $H_{j}$. Then we can utilize \cref{thm10:eigevo} through induction. We have
\begin{align}
e^{iH_1} \, i A \, e^{-iH_1} &= e^{\t{ad}_{iH_1}} \l( i A \r)  \\
&= i A + \sum_{k=1}^{\infty} \frac{1}{k!} \t{ad}_{iH_1}^{k} \l( i A \r) \\
&\in \orb ,
\end{align}
where the last line follows from each $\t{ad}_{iH_1}^{k}\l( i A \r)$ being contained in the $\orb$ since we proved that $\t{ad}_{iH_1}\l( iC \r) \in \orb$ for any $iC \in \orb$ and $iH_1 \in \g$. Suppose
\begin{align}
i A_{J-1} = e^{iH_{J-1}} \ldots e^{iH_{1}} \, i A \, e^{-iH_{1}} \ldots e^{-iH_{J-1}} \in \orb,
\end{align}
then by the same argument $V \, i A \, V^\dg = e^{i H_{J}} i A_{J-1} e^{-i H_{J}} \in \orb$.

Now consider $iC \in \orbp$. Then $\la iC, iB \ra_{\hs} = 0 $ for any $iB \in \orb$ by definition and so
\begin{align}
\la V \, iA \, V^\dg , iB \ra_{\hs} &= \t{Tr}\l( V \, i A \, V^\dg \, iB \r)  \\
&= \t{Tr}\l( i A \, V^\dg \, i B \, V \r)  \\
&= \la i A, V^\dg \, i B \, V \ra_{\hs} \\
&= 0,
\end{align}
where the last line follows from $V^\dg \in e^\g$ such that $V^\dg \, iB \, V \in \orb$.
\end{proof}

As such, $\orb$ is an invariant operator subspace under the adjoint action of $\g$. In representation theoretic language, it is the cyclic $\g$-module generated by $iO$ and need not itself be irreducible. This places the DOS within the invariant subspace framework considered in $\g$-sim~\cite{goh_lie-algebraic_2023}, with the additional specification that the DOS is the minimal invariant subspace containing the observable $iO$.

As in the main manuscript, we define the DOS projection of $i \rho(t)$ for $t \in [0,LK]$ as $i \vr(t)$.

\begin{sdefi}[DOS representation of $i \rho$, \cref{def:dosrho}]\label{sdef:dosrho}
\begin{align}
i \vr(t) &= U_{1:t}(\bmt) \t{proj}_{\mf{g}}^{iO} \l( i\rho(0) \r) U_{1:t}^{\dg}(\bmt),
\end{align}
\end{sdefi}

\noindent From \cref{sthm:dosrep}, it follows that $i\vr(t) \in \orb$ since $i\vr(0) \in \orb$ by projection and $U_{1:t}(\bmt) \in e^\g$.

\begin{sthm}[DOS representation of Density Operator]\label{sthm:dosrho}
For $t \in [0,LK]$:
\begin{align}
\t{proj}_{\mf{g}}^{iO} \, \l( U_{1:t}(\bmt) i\rho \, U_{1:t}^{\dg}(\bmt) \r) = U_{1:t}(\bmt) \, \t{proj}_{\mf{g}}^{iO} \l( i\rho \r) \, U_{1:t}^{\dg}(\bmt) , \\
\t{proj}_{\mf{g}}^{\perp, iO} \, \l( U_{1:t}(\bmt) i\rho \, U_{1:t}^{\dg}(\bmt) \r) = U_{1:t}(\bmt) \, \t{proj}_{\mf{g}}^{\perp, iO} \l( i\rho \r) \, U_{1:t}^{\dg}(\bmt)
\end{align}
\end{sthm}

\begin{proof}
Decompose $i \rho$ into its components in $\orb$ and $\orbp$:
\begin{align}
i\rho = \t{proj}_{\g}^{iO}\l( i\rho \r) + \t{proj}_{\g}^{iO,\perp}\l( i\rho \r) .
\end{align}

Then under the unitary evolution $U_{1:t}(\bmt)$
\begin{align}
U_{1:t}(\bmt)\, i\rho \, U_{1:t}^{\dg}(\bmt) &= U_{1:t}(\bmt)\, \t{proj}_{\g}^{iO}\l( i\rho \r)\, U_{1:t}^{\dg}(\bmt) \\
&\peq + U_{1:t}(\bmt)\, \t{proj}_{\g}^{iO,\perp}\l( i\rho \r)\, U_{1:t}^{\dg}(\bmt) .
\end{align}

By \cref{sthm:dosrep}, the DOS $\orb$ is invariant under the adjoint action of
$U_{1:t}(\bmt)$, so
\begin{align}
U_{1:t}(\bmt)\, \t{proj}_{\g}^{iO}\l( i\rho \r)\, U_{1:t}^{\dg}(\bmt) \in \orb .
\end{align}
Similarly, since $\orbp$ is invariant under the adjoint action, we have
\begin{align}
U_{1:t}(\bmt)\, \t{proj}_{\g}^{iO,\perp}\l( i\rho \r)\, U_{1:t}^{\dg}(\bmt) \in \orbp .
\end{align}

Then applying $\t{proj}_{\g}^{iO}$ removes the component in $\orbp$ and retains that in $\orb$:
\begin{align}
\t{proj}_{\g}^{iO}\l( U_{1:t}(\bmt)\, i\rho \, U_{1:t}^{\dg}(\bmt) \r) = U_{1:t}(\bmt)\, \t{proj}_{\g}^{iO}\l( i\rho \r)\, U_{1:t}^{\dg}(\bmt) ,
\end{align}
and the same logic follows for $\t{proj}_{\g}^{\perp, iO}$.
\end{proof}

Then any Heisenberg cut $t$, where $i \rho(t)$ meets $i O(t)$ is inside the DOS.

\begin{sthm}[Heisenberg Cut at $t$ inside DOS, \cref{thm4:heicut}]\label{sthm:heicut}
\begin{align}
\la i \rho(t) , i O(t) \ra_{\hs} = \la i \vr(t) , i O(t) \ra_{\dos}
\end{align}
\end{sthm}

\begin{proof}
\begin{align}
\la i \rho(t) , i O(t) \ra_{\hs} &= -\t{Tr}\l( i \rho(t) \, i O(t) \r) \\
&= -\t{Tr}\l( \t{proj}_{\g}^{iO}\l( i \rho(t) \r) iO(t) \r) - \t{Tr}\l( \t{proj}_{\g}^{iO,\perp}\l( i \rho(t) \r) iO(t) \r)  \\
&= -\t{Tr}\l( i \vr(t) \, i O(t) \r) \qquad (\t{from \cref{sthm:dosobs,sthm:dosrho}}) \\
&= \la i \vr(t) , i O(t) \ra_{\dos}
\end{align}
\end{proof}

We now prove the results on the dimensionality of the DOS in comparison to the DLA based on the multiplicative sector.

\begin{sdefi}[Multiplicative Sector of Order $k$]
Given $\g$ with matrix basis $\B^\g$, define the multiplicative sector of order $k$ of $\g$ as
\begin{align}
\Qc_\g^k &= \t{span}\l( \l\{ i\prod_{j=1}^{k} G_{j} : iG_{j} \in \g \r\} \r) \\
&= \t{span}\l( \l\{ i\prod_{j=1}^{k} B_{j} : iB_{j} \in \B^\g \r\} \r) ,
\end{align}
with $\Qc^0 = \t{span}\l( \l\{ i \I \r\} \r)$.
\end{sdefi}

\noindent Given that $iO$ is a sum of products of $k$ elements in $\g$, we show that $\orb \subseteq \Qc_\g^k$, which bounds the dimension. For example $iO = i Z_j Z_k \in \Qc_{\g_{XX,YY,Z}}^2$ since $Z_j, Z_k \in \{ Z_{\ell} \}_{\ell=1}^{n}$. Similarly, then $\sum_{j<k} J_{jk} Z_{j} Z_{k} \in \Qc_{\g_{XX,YY,Z}}^2$.

\begin{sthm}[Dimension of DOS, \cref{thm5:orb_dla_dim_outside}]
Given DLA $ \g $, for any observable $ iO \in \Qc_\g^k $:
\begin{align}
\orb \subseteq \Qc_\g^k,
\end{align}
and so
\begin{align}
\t{dim}\l( \orb \r) \leq \t{dim}\l( \g \r)^{k} .
\end{align}
\end{sthm}

\begin{proof}
From its definition, it follows that the product of $k$ terms in the basis $\B^\g$ is bound by $\t{dim}\l( g \r)^k$ elements, such that
\begin{align}
\t{dim}\l( \Qc_k(\g) \r) \leq d^k.
\end{align}

We show that $\Qc_g^k$ is invariant under the adjoint action of $\g$. Let $iH \in \g$ and let $iA_{r} = iB_{r,1},\dots,iB_{r,k} \in \B^\g$. Then
\begin{align}
\l[ iH,\, i B_{r,1} \cdots B_{r,k} \r]
&= \sum_{\ell=1}^{k} B_{r,1} \cdots B_{r,\ell-1}[iH, iB_{r,\ell}]B_{r,\ell+1}\cdots B_{r,k},
\end{align}
based on the Leibniz rule for commutators. Since $\g$ is a Lie algebra and $iH,iB_\ell \in \g$, for each term we have $ [ iH, iB_\ell ] \in \g $ and so $iA_{r} \in \g$. Then for any $iA = \sum_{r} i w_{r} A_{r}$, by linearity
\begin{align}
[ iH, iA ] &= \l[ iH, \sum_{r} w_{r} \, i A_{r} \r] \\
&= \sum_{r} w_{r} \l[ iH, iA_{r} \r] \\
&\in \Qc_\g^k .
\end{align}

Since $iO \in \Qc_k(\g)$ and $\orb$ is the smallest $\g$-invariant matrix subspace containing $iO$, it follows that
\begin{align}
\orb \subseteq \Qc_\g^k.
\end{align}
Therefore
\begin{align}
\t{dim}\l( \orb \r) \le \t{dim}\l( \Qc_k(\g) \r) \le \t{dim}\l( \g \r)^k.
\end{align}
\end{proof}

In the case that $iO$ is a sum over different sectors, we can extend this over this collection. Let
\begin{align}
\Qc_{0:k} = \bigcup_{j=0}^{k} \Qc_{k} .
\end{align}

Then by linearity, we immediately get the following corollary.

\begin{scor}
Given DLA $\g$, for any observable $iO \in \Qc_{1:k}$:
\begin{align}
\t{dim}\l( \orb \r) &\leq \sum_{j=1}^{k} \t{dim}\l( \g \r)^{j} \\
&\leq \Oc\l( \t{dim}\l( \g \r)^{k} \r)
\end{align}
\end{scor}

\section{Decomposition of Observable Subspaces through Dynamics}

\begin{suppfigure*}[!t]
\centering
\includegraphics[width=0.7\textwidth]{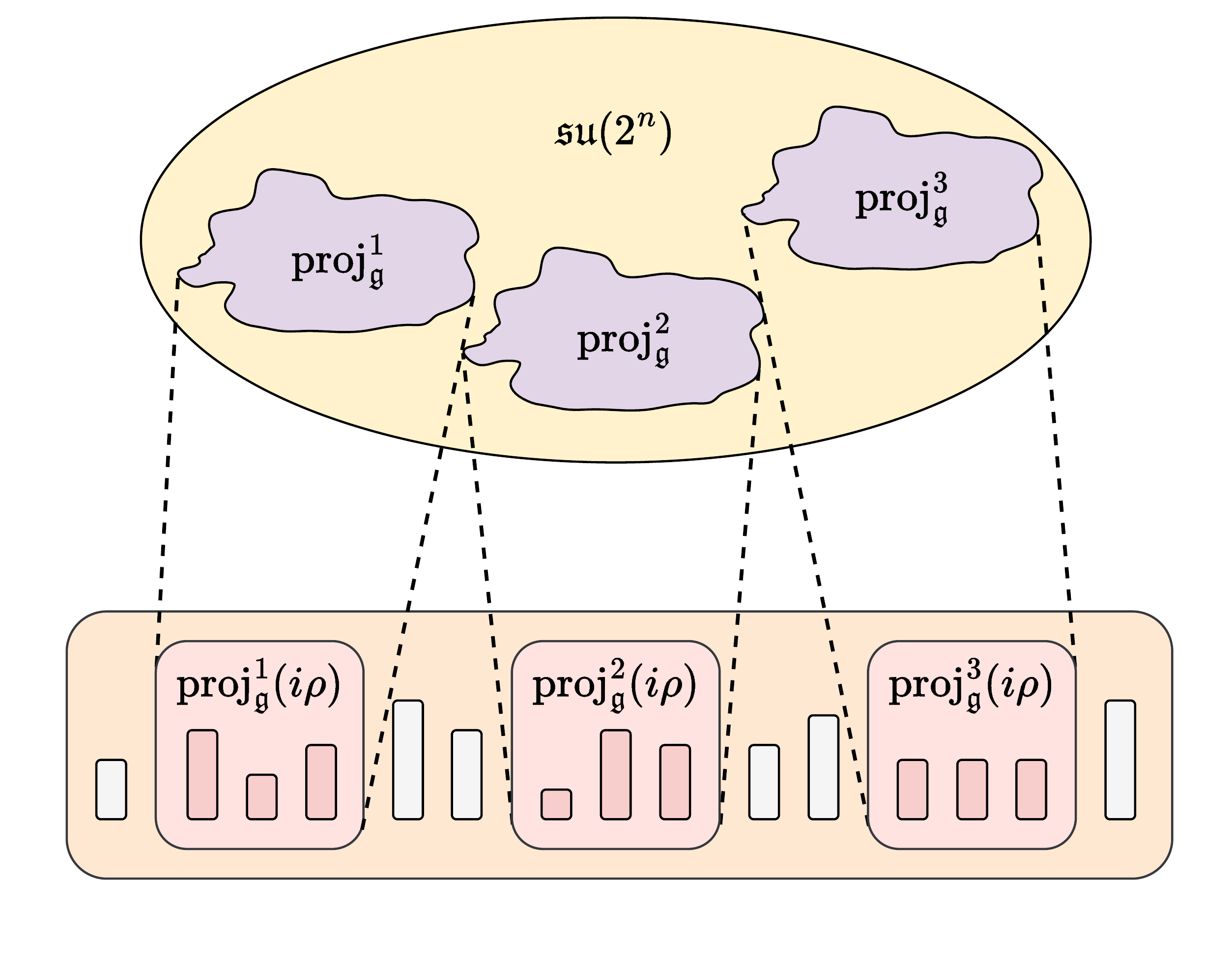}
 \caption{\tsbf{Decomposition into Multiple Dynamic Observable Subspaces} (reprint of \cref{fig:img_multi_dyn_sub}). Given an observable $O = v_1 \, O_1 + v_2 \, O_2 + v_3 \, O_3$, we can define the DOS for each $O_j$ and simulate them separately.}
\label{suppfig:img_multi_dyn_sub}
\end{suppfigure*}

As discussed in the main manuscript, a Hamiltonian $H = \sum_j v_j O_j$ can be decomposed into observable terms, each with its own DOS. Rather than constructing one large joint subspace, each $O_j$ can be simulated independently inside its own $\t{orb}_{\g}^{iO_j}$, keeping the subspace dimensions small. The following theorem formalizes this decomposition.

\begin{sthm}[Multiple DOS Representation, \cref{thm12:multidosrep}]\label{sthm:multidosrep}
Given $H = \sum_j v_j O_j$ and $\orbj$ for each $O_j$ (), let $\la \cdot, \cdot \ra_{\dos_j}$ be the projected trace inner product for orbital $j$ and let $i \vr_{j}(t)$ be equal to $ U_{1:t}(\bmt) \, \t{proj}_{\g}^{iO_{j}}\l( i \rho(0) \r) \, U_{1:t}^{\dg}(\bmt)$ from \cref{def:prj_op} for orbital $j$. Then for any $V \in e^{\g}$:
\begin{align}
\t{Tr}\l( V \rho(t) V^{\dg} \, H(t) \r) = \sum_{j} v_{j} \, \la \l( V \, i \vr_j(t) \, V^{\dg} \r), i O_j(t) \ra_{\dos_j} .
\end{align}
\end{sthm}

\begin{proof}
The proof follows directly from linearity. Since $H(t) = \sum_j v_j O_j(t)$,
\begin{align}
\t{Tr}\l( V \rho(t) V^{\dg} \, H(t) \r)
&= \sum_j v_j \t{Tr}\l( V \rho(t) V^{\dg} \, O_j(t) \r) \\
&= \sum_j v_j \t{Tr}\l( V (-i) \rho(t) V^{\dg} \, i O_j(t) \r) \\
&= \sum_j v_j \t{Tr}\l( (V i \rho(t) V^{\dg})^{\dg} \, i O_j(t) \r) \\
&= \sum_j v_j \la V \, i \rho(t) \, V^{\dg}, i O_j(t) \ra_{HS} \\
&= \sum_j v_j \la V \, i \vr_j(t) \, V^{\dg}, i O_j(t) \ra_{HS}
\qquad (\t{from \cref{sthm:dosrep}}) \\
&= \sum_j v_j \la V \, i \vr_j(t) V^{\dg}, i O_j(t) \ra_{\dos_j}
\qquad (\t{from \cref{sthm:heicut}}).
\end{align}
\end{proof}

In \results and \cref{part3:exp}, we discuss examples in which there are multiple orbitals for the Hamiltonian in question.

\begin{figure}[!t]
\begin{algobox}
\setlength{\intextsep}{0.5em}

\captionof{algorithm}{\tsf{Generate Orbital}}\label{alg:gen_orb}
\vspace{-0.9em}
\hrule
\vspace{0.2em}
\begin{algorithmic}[1]
\Statex \textbf{Inputs: } $ \tobs $, $ \tgens $
\vspace{0.2em}
\State $ \B_O \rgt \{ \tobs \} $ \Comment{Initialize open basis with $\tobs$}
\While{$ \B_O $ is non-empty}
    \For{$ig \in \B_O$}
    \For{$iG \in \tgens_{\tt{locs}(ig)}$} \Comment{Only generators with shared qubit location}
        \State $iT \rgt [iG,\, ig]$
        \State $iR \rgt \tt{Gram-Schmidt}\l( iT, \tt{gens}_{\tt{locs}(ig)} \r)$
        \If{$iR \neq \tbf{0} $}
        \State $\B_N \rgt \B_N \, \cup \, \{ iR \} $
        \EndIf
    \EndFor
    \EndFor
    \State $\B \rgt \B \, \cup \, \B_O$ \Comment{Add previously open basis to collection}
    \State $\B_O \rgt \B_N$ \Comment{Define open basis as the newly found basis}
\EndWhile
\State \textbf{return} $ \B $
\end{algorithmic}

\end{algobox}
\vspace{-1.0em}

\captionof*{algorithm}{{Alg.~\ref*{alg:gen_orb}}: \tsbf{Generating the Orbital of an Observable.} This breadth-first search algorithm generates the orbit of a specified observable $\tobs$ through generators $\tgens$.}

\setlength{\intextsep}{\ointextsep}
\end{figure}

\begin{figure}[!t]
\begin{algobox}
\setlength{\intextsep}{0.5em}

\captionof{algorithm}{\tsf{Build Adjoint Operator}}\label{alg:op_adj}
\vspace{-0.9em}
\hrule
\vspace{0.2em}
\begin{algorithmic}[1]
\Statex \textbf{Inputs: } $ \B $, $ iG $
\vspace{0.2em}
\State $ A \rgt \tbf{0}^{|\B| \times |\B|} $
\For{$ iB_j \in \B_{\tt{locs}(iG)} $} \Comment{Only basis elements sharing qubit location with $iG$}
    \State $ iT \rgt [iG,\, iB_j] $
    \If{$ iT \neq \tbf{0} $}
        \For{$ iB_k \in \B_{\tt{locs}(iT)} $} \Comment{Only basis elements sharing qubit location with $iT$}
            \State $ A[k,j] \rgt \t{Tr}\l( (iB_k)^\dg \, iT \r) $
        \EndFor
    \EndIf
\EndFor
\State \textbf{return} $ A $
\end{algorithmic}
\end{algobox}
\vspace{-1.0em}

\captionof*{algorithm}{{Alg.~\ref*{alg:op_adj}}: \tsbf{Building the Adjoint Operator.} Given an orthonormal basis $\B$ and gate $iG$, this algorithm builds the matrix representation $A$ of the adjoint map $\t{ad}_{iG} = [iG,\,\cdot\,]$ in $\B$. Columns corresponding to basis elements that $iG$ does not act on are identically zero.}

\setlength{\intextsep}{\ointextsep}
\end{figure}

\begin{figure}[!t]
\begin{algobox}
\setlength{\intextsep}{0.5em}

\captionof{algorithm}{\tsf{Build Projection Operator}}\label{alg:op_proj}
\vspace{-0.9em}
\hrule
\vspace{0.2em}
\begin{algorithmic}[1]
\Statex \textbf{Inputs: } $ \B $, $ iG $
\vspace{0.2em}
\State $ P \rgt \tbf{0}^{|\B| \times |\B|} $
\For{$ iB_j \in \B_{\tt{locs}(iG)} $} \Comment{Only basis elements sharing qubit location with $iG$}
    \State $ iT \rgt [iG,\, iB_j] $
    \If{$ iT \neq \tbf{0} $}
        \State $ P[j,j] \rgt 1 $
    \EndIf
\EndFor
\State \textbf{return} $ P $
\end{algorithmic}
\end{algobox}
\vspace{-1.0em}

\captionof*{algorithm}{{Alg.~\ref*{alg:op_proj}}: \tsbf{Building the Projection Operator.} Given an orthonormal basis $\B$ and gate $iG$, this algorithm builds the diagonal indicator matrix $P$, where $P[j,j] = 1$ if $iG$ acts non-trivially on $iB_j$ (i.e.\ $[iG, iB_j] \neq \tbf{0}$), and zero otherwise.}

\setlength{\intextsep}{\ointextsep}
\end{figure}

\section{Efficient Construction of the DLA and DOS}\label{sec:bfs_construction}

In this section, we discuss the efficient construction of an orthonormal basis for the DLA and the DOS. In \cref{alg:gram}, we describe the Gram-Schmidt algorithm. Given a matrix $iT$ and a set of orthonormal matrices $\B$ under the Hilbert-Schmidt inner product, the algorithm finds the orthogonal component $R$ of $T$. If $R$ is the null matrix, then $T$ is in the span of $\B$. Assuming that each matrix $iT$, $iB \in \B$ has $\Oc\l( 1 \r)$ Pauli words each of length at most $n$ then the computation runs in $\Oc\l( n |\B| \r)$. To do this, we must use string manipulation or Big integers to symbolically calculate commutation such that we never have to use the numerical form of each element in $\C^{2^{n} \times 2^{n}}$.

\Cref{alg:gen_orb} generates the orbital $\t{orb}_{\g}^{iO}$ given as inputs $iO$ and $\G$. In words, it does so using breadth first search beginning with the open set $iO$. In each round until the open set is empty, each element of the open set is commuted by each generator that does not commute with the element and orthonormalized. A simple check to improve runtime is to notice that if the open element and the gate generator share no qubit indices, then they must commute. Moreover, only elements that share qubit indices are nonorthogonal, which reduces the relevant basis for the Gram-Schmidt algorithm. These new elements are then placed into a new set. After this has occurred for every element in the open set, the new set is set to the open set and the open set is added to the closed set. Once the new set is the empty set, the last nonempty open set is added to the closed set and the closed set is an orthonormal basis $\B$ for $\t{orb}_{\g}^{iO}$.

In certain cases, even more aggressive pruning is possible for generating the orbital. For example, it is known that if $Z_{j} Z_{k}$ are part of Pauli word $P$, then $X_{j} X_{k}$ and $Y_{j}Y_{k}$ commute as well as $Z_{j}$ and $Z_{k}$. Tracking what terms will commute also is important for generating the \textit{numerical} representation of the projection superoperator and the adjoint superoperator over the matrix vector space, as described in \cref{alg:op_proj} and \cref{alg:op_adj} respectively.

\pagebreak

\part{\sffamily Representation and Computation of Dynamics inside the Dynamical Observable Subspace}\label{part2:repindos}

\section{Pauli Evolution in Observable Dynamic Subspace}

Recall that the Pauli basis over $n$ qubits is $\P = \{ \I, X, Y, Z \}^{ \otimes n} / \{ \I \}$. In this section, we delineate how to simulate the dynamics of the VQA with each gate as a Pauli word inside the DOS. We have a parameterized quantum circuit (PQC) defined by $ L $ layers of gates with Pauli generators $ \G = \{ G_{1}, \ldots, G_{k} \}, G_{j} \in \P $.

Then given $O \in \P$, it follows that nested commutation generates only other pauli terms such that $\orb = \t{span}\l( \B \r) $ with Pauli subbasis $\B \subseteq \P$.

We can write:
\begin{align}
i \vr(t) = \t{proj}_{\g}(i \rho) = \sum_{j=1}^{\t{dim}\l( \orb \r)} r_{j}(t) \, i Q_{j},
\end{align}
with Pauli $Q_{j} \in \B$.

We know that for two Pauli strings $ P, Q \in \P $, either they commute or they anticommute. Recall that $ P_{j} Q_{j} = (-1)^{c_{j}} Q_{j} P_{j} $ for $ P_{j}, Q_{j} \in \{ \I, X_j, Y_j, Z_j \} $ with
\begin{align}
c_{j} = (\t{Tr}(\I) - \t{Tr}(P_{j} Q_{j})) / \t{Tr}(\I).
\end{align}

Indeed, this can be generalized for two Pauli strings over $n$ qubits.

\begin{suppfigure*}[!t]
\centering
\includegraphics[width=0.8\textwidth]{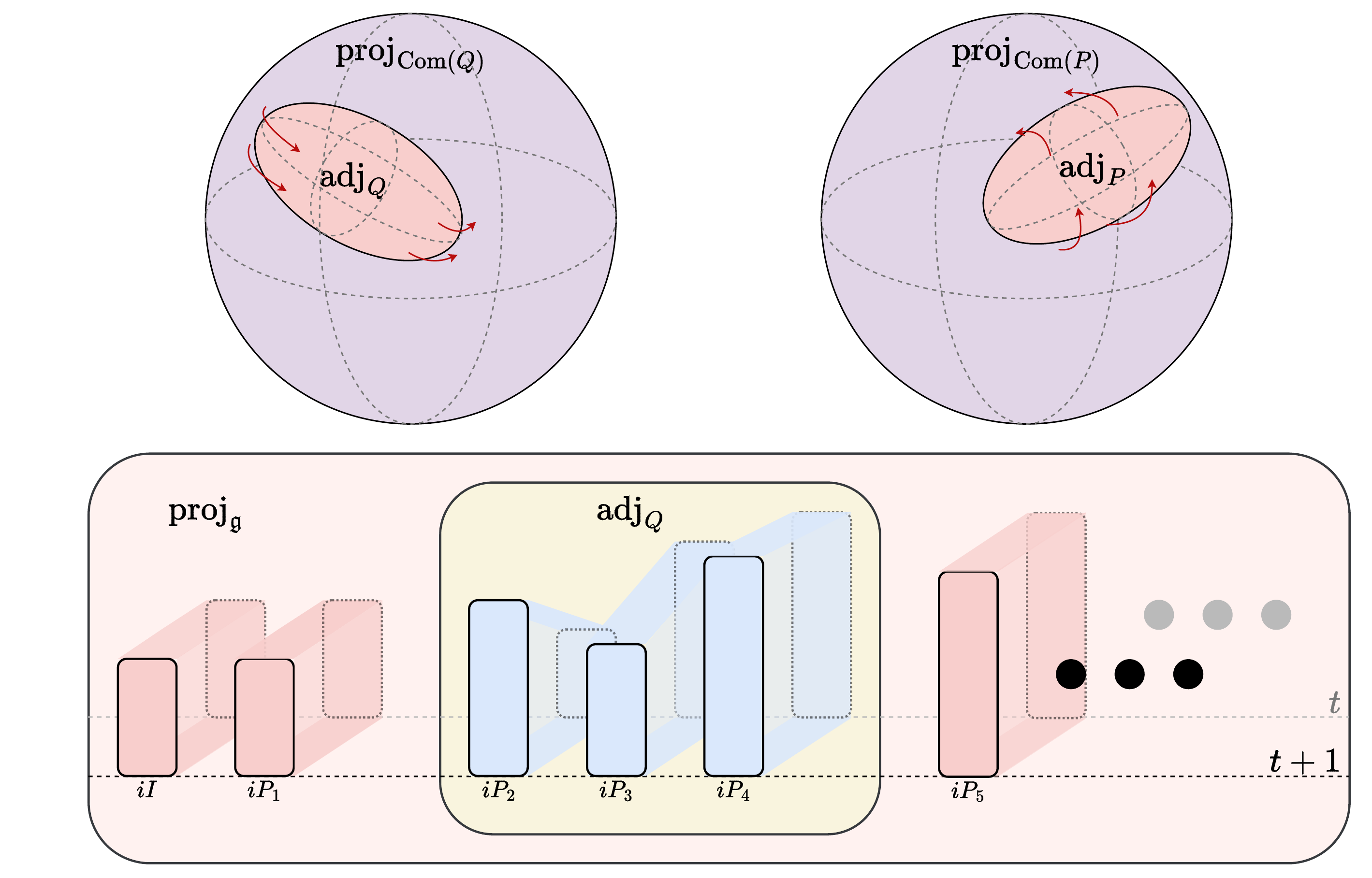}
 \caption{\tsbf{Dynamic Subspace of the Adjoint Superoperator} (reprint of \cref{fig:img_adjop}). Top depicts the dynamic action of two adjoint superoperators $\t{ad}_{iQ}$ and $\t{ad}_{iP}$ with $ \t{proj}_{\t{Com}(i Q)} $ and $ \t{proj}_{\t{Com}(i P)} $ as the outside space that is left stationary respectively. Bottom depicts $e^{i\theta Q}: \rho(t) \gt \rho(t+1) $ where support inside $ \t{proj}_{\t{Com}(Q)}^{\perp} $ colored blue is transformed while support outside (in $\t{proj}_{\t{Com}(i Q)}$) colored red is left stationary.}
\label{suppfig:img_adjop}
\end{suppfigure*}

\begin{sthm}[Pauli Commutation]\label{sthm:paulicomm}
Given two Pauli strings $ P, Q \in \P $,
\begin{align}
P \, Q = (-1)^{|c|} \, Q \, P,
\end{align}
where $ c \in \{0,1\}^{n} $ with $ c_{j} = (\t{Tr}(\I) - \t{Tr}(P_{j} \, Q_{j})) / \t{Tr}(\I) $.
\end{sthm}

\begin{proof}

Pauli operators on differing qubits commute, $ P_{j} Q_{k} = Q_{k} P_{j} $ for $ j \neq k $. Then:
\begin{align}
P Q &= P_1 \, \ldots \, P_n \, Q_1 \, \ldots \, Q_n \\
&= (P_1 \, Q_1) \, \ldots \, (P_n \, Q_n) \\
&= (-1)^{c_i} \, (Q_1 \, P_1) \ldots (-1)^{c_{n}} (Q_n \, P_n) \\
&= (-1)^{|c|} \, Q \, P,
\end{align}

where $ c \in \{0,1\}^n $ is a vector that counts the number of qubits for which $P,Q$ would fail to commute. Moreover if $P \, Q = -Q \, P$, then $ P \, Q = \frac{1}{2}(P \, Q + P \, Q) = \frac{1}{2}(P \, Q -Q \, P) = \frac{1}{2} \l[ P, Q \r] $.

\end{proof}

Because of this identity, it is valuable to define the set of Pauli strings inside the DOS that commute (and therefore fail to anticommute) with a specific Pauli string:
\begin{align}
\t{Com}\l( i P \r) = \t{span}\l( \l\{ iB : [P, B ] = 0, B \in \B \r\} \r) \subseteq \t{orb}_{\mf{g}}^{iO} .
\end{align}

Then the evolution superoperator can be written as given by the following theorem.

\begin{sthm}[Pauli Gate Evolution]\label{sthm:pauliupdate}

For a Pauli generator $P$ whose adjoint action preserves the DOS, the evolution superoperator on the DOS is
\begin{align}
\t{evo}_{iP}\l( \,\cdot\,;\theta \r) = e^{\theta \, \t{ad}_{iP}}
= \mc{I} + (\cos(2\theta)-1) \, \t{proj}_{\t{Com}(i P)}^{\perp} + \frac{\sin(2\theta)}{2} \, \t{ad}_{iP} .
\end{align}
Applying this superoperator to $i\vr \in \orb$ gives
\begin{align}
\t{evo}_{iP}\l( i\vr;\theta \r)
&= e^{i\theta P} \, i\vr \, e^{-i\theta P} \\
&= i\vr + (\cos(2\theta)-1) \, \t{proj}_{\t{Com}(i P)}^{\perp}\l( i\vr \r)
+ \frac{\sin(2\theta)}{2} \, \t{ad}_{iP}\l( i\vr \r) .
\end{align}
\end{sthm}

\begin{proof}

At point $ t $ the density operator is
\begin{align}
i\vr(t) &= U_{1:t}(\theta) \, i \vr(0) \, U_{1:t}^{\dg}(\theta) \\
&= \sum_{j=1}^{\t{dim}\l( \orb \r)} r_{j}(t) \, i Q_{j}
\end{align}

Then we apply $P$ at $t$:
\begin{adjustwidth}{-2em}{-2em}
\begin{align}
i \vr(t+1) &=  e^{i \theta_{t} P} \, i \vr(t) \, e^{-i \theta_{t} P} \\
&= (\cos(\theta_t) \, \I + i \, \sin(\theta_t) \, P ) \, i \vr(t) \, (\cos(\theta_t) \, \I - i \, \sin(\theta_t) \, P) \\
&= \cos^2(\theta_t) \, i\vr(t) + i\cos(\theta_t) \sin(\theta_t) (P \, i\vr(t) - i\vr(t) \, P) + \sin^2(\theta_t) \, P \, i\vr(t) \, P \\
&= \sum_{j \, : \, [P,Q_{j}] = 0 } r_j(t) \, i Q_{j} \\
&\peq + \sum_{j : \{ P,Q_{j} \} = 0} r_j(t) \l( \cos^2(\theta_t) \, i Q_j + \cos(\theta_t) \sin(\theta_t) \, [iP, iQ_{j}] - \sin^2(\theta_t) \, i Q_j \, P^2 \r) \\
&= \t{proj}_{\t{Com}(iP)}(i \vr(t)) \\
&\peq + \sum_{j : \{ P,Q_{j} \} = 0} r_j(t) \l( (\cos^2(\theta_t) - \sin^2(\theta_t)) \, i Q_{j} + \cos(\theta_t)\sin(\theta_t) \, \t{ad}_{iP}\l( iQ_{j} \r) \r) \\
&= i \vr(t) + (\cos(2\theta_t) - 1) \, \t{proj}_{\t{Com}(iP)}^{\perp} (i \vr(t)) + \frac{1}{2} \sin(2\theta_t) \, \t{ad}_{iP}(i \vr(t)) .
\end{align}
\end{adjustwidth}

In particular then the update to the density operator can be written:
\begin{align}
i\vr(t+1) &= i\vr(t) + (\cos(2\theta_t) - 1) \, \t{proj}_{\t{Com}(i P)}^{\perp} (i\vr(t)) + \frac{1}{2} \sin(2\theta_t) \, \t{ad}_{i P}(i\vr(t)),
\end{align}
which provides us with a framework for fast sparse simulation.
\end{proof}

\begin{figure}
\begin{algobox}
\setlength{\intextsep}{0.5em}

\captionof{algorithm}{\tsf{Subspace Simulation of Pauli Dynamics}}\label{alg:sim_dympauli}
\vspace{-0.9em}
\hrule
\vspace{0.2em}
\begin{algorithmic}[1]
\Statex \textbf{Inputs: } $ \tt{rho} $, $ \bmt $, $ \tt{projs} $, $ \tt{adjs} $
\vspace{0.2em}
\State $ t := 0 $
\While{$ t < LK $}
    \State $k := 0$
    \While{ $ k < K $ }
        \State $ t \rgt t + 1 $
        \State $ k \rgt k+1 $
        \State $\tt{rho}\_\tt{p} \rgt \tprojs[k] \, \tt{rho} $
        \State $\tt{rho}\_\tt{p} \rgt (\cos(2 \theta_{t})-1) \tt{rho}\_\tt{p} +  \frac{1}{2} \sin(2 \theta_{t}) \tt{adjs}[k] \tt{rho}\_\tt{p} $
        \State $\tt{rho} \rgt \tt{rho} + \tt{rho}\_\tt{p} $
    \EndWhile
\EndWhile
\State \textbf{return} $ \tt{rho} $
\end{algorithmic}

\end{algobox}
\vspace{-1.0em}

\captionof*{algorithm}{{Alg.~\ref*{alg:sim_dympauli}}: \tsbf{Simulating Subspace Dynamics with Pauli Operators.} Given $\rho$ and $ U(\bmt) = \prod_{\ell=1}^{L} \prod_{k=1}^{K} e^{i \theta_{t} G_{t \t{mod} K}} $ of parameterized Pauli gates, this algorithm simulates the dynamics of $ U(\bmt) \rho U^{\dg}(\bmt) $ through vector representation $\tt{rho}$ and matrix representations $\tt{projs}, \tt{adjs}$.}

\setlength{\intextsep}{\ointextsep}
\end{figure}

Using \cref{alg:gen_orb}, an orthonormal basis is generated. Since the gate generators are Pauli words, it follows that the orthonormal basis is of individual Pauli words.  Then the numerical representations $\t{proj}$ and $\t{ad}$ for each $iP \in \G$ are built. Then \cref{alg:sim_dympauli} simulates the dynamics of the Pauli operators inside a given DOS.

\section{Dynamics with Diffusor Mixers}\label{sec:dyn_diffmixers}

In this section, we consider the simulation of another class of important mixers, which build a foundation for QAOA on constrained problems~\cite{leipold2026imposing} and serve as a foundation for Eigendecomposable Hamiltonians. They are rotations of a projection operator $ F $ (such that $ F^2 = F $):
\begin{align}
U_F(\theta_t) &= e^{i \theta_t F} \\
&= \I + \sum_{k=1}^{\infty} (i \theta_t)^{k} F^{k} / k! \\
&= \I + \sum_{k=1}^{\infty} (i \theta_t)^{k} F / k! \\
&= \I + F \l( \sum_{k=1}^{\infty} (i \theta_t)^{k} / k! \r) \\
&= \I + (e^{i\theta_t}-1) F,
\end{align}

First, we prove a useful technical lemma about the simplification through double commutation.

\begin{slem}\label{slem:2com}
Given $ F^2 = F $, for any $ A $:
\begin{align}
\l[ F , \l[ F , A \r] \r] = F \, A  + A \, F - 2 \, F \, A \, F
\end{align}
\end{slem}

\begin{proof}
\begin{align}
\l[ F , \l[ F , A \r] \r] &= F \, \l[ F , A \r] - \l[ F , A \r] \, F \\
&= F \, (F \, A - \, A \, F) - (F \, A - A \, F) \, F \\
&= F \, F \, A - F \, A \, F - F \, A \, F + A \, F \, F \\
&= F \, A  + A \, F - 2 \, F \, A \, F
\end{align}
\end{proof}

\begin{sthm}[Diffusor Gate Evolution]\label{slem:diffevo}

Let $F^\dg = F = F^2$. The evolution superoperator on $\mf{u}\l( 2^n \r)$ is
\begin{align}
\t{evo}_{iF}\l( \,\cdot\,;\theta_t \r) = e^{\theta_t \, \t{ad}_{iF}}
= \mc{I} + \sin(\theta_t) \, \t{ad}_{iF} + (1-\cos(\theta_t)) \, \t{ad}_{iF}^{2} .
\end{align}
Applying this superoperator to $i\vr(t)$ gives
\begin{align}
\t{evo}_{iF}\l( i\vr(t);\theta_t \r)
&= e^{i\theta_t F} \, i\vr(t) \, e^{-i\theta_t F} \\
&= i\vr(t) + \sin(\theta_t) \, \t{ad}_{iF}\l( i\vr(t) \r)
+ (1-\cos(\theta_t)) \, \t{ad}_{iF}^{2}\l( i\vr(t) \r) .
\end{align}
These identities restrict to the DOS whenever $\t{ad}_{iF}$ preserves it.
\end{sthm}

\begin{proof}
\begin{align}
i\vr(t+1) &= \t{evo}_{iF}(i\vr(t); \theta_t) \\
&= e^{i \theta_t F} \, i \vr(t) \, e^{-i \theta_t F} \\
&= (\I +(e^{i \theta_t } - 1)F) \, i \vr(t) \, (\I +(e^{-i \theta_t } - 1)F) \\
&= i \vr(t) + (e^{i\theta_t} - 1) \,  F \, i \vr(t) + (e^{-i\theta_t} - 1) \, i \vr(t) \, F + (e^{i\theta_t} - 1)(e^{-i \theta_t} - 1) F \, i \vr(t) \, F \\
&= i \vr(t) + (e^{i\theta_t} - 1) \, F \, i \vr(t) + (e^{-i\theta_t} - 1) \, i \vr(t) \, F - 2(\cos(\theta_t)-1) \, F \, i \vr(t) \, F
\end{align}

Using the trigonometric identity $ e^{ \pm i\theta} = \cos{\theta} \pm i \sin{\theta}$, we can simplify:
\begin{align}
i\vr(t+1) &= i\vr(t) + \l( \cos(\theta_t)-1 + i\sin(\theta_t) \r) F i\vr(t) \\
&\phantom{=} + \l( \cos(\theta_t)-1 - i\sin(\theta_t) \r) i\vr(t) F \\
&\phantom{=} - 2(\cos(\theta_t) - 1) F i\vr(t) F \\
&= i\vr(t) + i\sin(\theta_t)\l( F i\vr(t) - i\vr(t) F \r) \\
&\phantom{=} + (\cos(\theta_t)-1)\l( F i\vr(t) + i\vr(t) F \r)
- 2(\cos(\theta_t) - 1) F i\vr(t) F \\
&= i\vr(t)
+ i\sin(\theta_t)\l( F i\vr(t) - i\vr(t) F \r) \\
&\phantom{=} + \l( \cos(\theta_t)-1 \r)\l( F i\vr(t) + i\vr(t) F - 2 F i\vr(t) F \r) \\
&= i\vr(t) + i\sin\l( \theta_t \r)\l( F\,i\vr(t) - i\vr(t)F \r) \\
&\phantom{=} + \l( \cos\l( \theta_t \r)-1 \r)\l( F\,i\vr(t) + i\vr(t)F - 2F\,i\vr(t)F \r) \\
&= i\vr(t) + i\sin\l( \theta_t \r)\com{F}{i\vr(t)} - \l( \cos\l( \theta_t \r)-1 \r)\com{iF}{\com{iF}{i\vr(t)}} & (\t{\cref{slem:2com}}) \\
&= i\vr(t) + \sin(\theta_t) \, \ad{iF}{i\vr(t)} - (\cos(\theta_t)-1) \ad{iF}{\ad{iF}{i\vr(t)}}
\end{align}
\end{proof}

\section{Dynamics with Eigendecomposable Hamiltonians}\label{sec:dyn_eigham}

\begin{figure}
\begin{algobox}
\setlength{\intextsep}{0.5em}

\captionof{algorithm}{\tsf{Subspace Simulation of Eigendecomposable Dynamics}}\label{alg:sim_dym}
\vspace{-0.9em}
\hrule
\vspace{0.2em}
\begin{algorithmic}[1]
\Statex \text{Given a }
\Statex \textbf{Inputs: } $ \tt{rho} $, $ \bmt $, $ \tt{adjs} $, $\tt{adjs2}$
\vspace{0.2em}
\State $ t := 0 $
\While{$ t < LK $}
    \State $k := 0$
    \While{ $ k < K $ }
        \State $ t \rgt t + 1 $
        \State $ k \rgt k+1 $
        \State $\tt{rho}\_\tt{p} \rgt \tprojs[k] \, \tt{rho} $
        \State $\tt{rho}\_\tt{p} \rgt (\cos(2 \theta_{t})-1) \trho +  \sin(2 \theta_{t}) \trho $
        \State $\tt{rho} \rgt \tt{rho} + \tt{rho}\_\tt{p} $
    \EndWhile
\EndWhile
\State \textbf{return} $ \tt{rho} $
\end{algorithmic}

\end{algobox}
\vspace{-1.0em}

\captionof*{algorithm}{{Alg.~\ref*{alg:sim_dym}}: \tsbf{Simulating Subspace Dynamics with Eigendecomposable Unitaries.} Given $\rho$ and $ U(\bmt) = \prod_{\ell=1}^{L} \prod_{k=1}^{K} e^{i \theta_{t} G_{t \t{mod} K}} $ of parameterized Pauli gates, this algorithm simulates the dynamics of $ U(\bmt) \rho U^{\dg}(\bmt) $ through vector representation $\tt{rho}$ and matrix representations $\tt{projs}, \tt{adjs}$.}

\setlength{\intextsep}{\ointextsep}
\end{figure}

The diffusor gate of \cref{sec:dyn_diffmixers} handles a single projector $F$. A natural and practically important generalization are Hamiltonians that admit a spectral decomposition into a small number $J$ of orthogonal projectors. An important case is that of local generators, which act nontrivially on a small number of qubits (with identity on the rest) and therefore have at most $2^k$ distinct eigenvalues for $k$ active.

We have a general Hermitian operator, $ H = \sum_{j=1}^{J} n_{j} \,  F_{j} $ with orthogonal eigenspace projectors satisfying
\begin{align}
F_{j} F_{k} &= \delta_{jk} F_{k}, \\
\sum_{j=1}^{J} F_{j} &= \I,
\end{align}
with $ n_{j} $ as the associated eigenvalue (energy) such that $ n_{j} = \t{Tr}\l( H F_{j} \r) $. The following well known fact enables the decomposition over the projectors.

\begin{slem}[Spectral Factorization of Matrix Exponential]\label{slem:specfact}
Let $H = \sum_{j=1}^{J} n_j F_j$ with orthogonal projectors satisfying $F_j F_k = \delta_{jk} F_j$ and $\sum_j F_j = \I$. Then
\begin{align}
e^{i\theta H} = \prod_{j=1}^{J} e^{i\theta n_j F_j} .
\end{align}
\end{slem}

\begin{proof}
We show both sides equal $\sum_j e^{i\theta n_j} F_j$. Since $F_j F_k = 0$ for $j \neq k$, cross terms vanish in any power, giving $\bigl(\sum_j n_j F_j\bigr)^m = \sum_j n_j^m F_j$. Substituting into the Taylor series,
\begin{align}
e^{i\theta H} = \sum_{m=0}^{\infty} \frac{(i\theta)^m}{m!} \sum_j n_j^m F_j = \sum_j e^{i\theta n_j} F_j .
\end{align}
For the right side, $F_j^m = F_j$ for $m \geq 1$ gives $e^{i\theta n_j F_j} = \I + (e^{i\theta n_j} - 1) F_j$. Since the $F_j$ have disjoint support, the product telescopes,
\begin{align}
\prod_j \bigl( \I + (e^{i\theta n_j} - 1) F_j \bigr) = \I + \sum_j (e^{i\theta n_j} - 1) F_j = \sum_j e^{i\theta n_j} F_j ,
\end{align}
where the last step uses $\sum_j F_j = \I$.
\end{proof}

Then the associated unitary can be decomposed over the projectors as
\begin{align}
U(\theta_t) &= e^{i \theta_t H} \\
&= e^{i \theta_t \sum_{j} n_{j} F_{j}} \\
&= \prod_{j} U_{F_j}\l( n_j \theta_t \r) ,
\end{align}
where the last step is \cref{slem:specfact}. Over the span $\mathcal{B}$, the associated transformation is applying each projector after one another.
\begin{align}
i\vr^{(0)} &= i\vr, \nonumber \\
i\vr^{(j)} &= \t{evo}_{iF_j}\l( i\vr^{(j-1)}; n_j \theta_t \r), \quad j = 1, \ldots, J, \\
\t{evo}_{iH}(i\vr; \theta_t) &= i\vr^{(J)}.
\end{align}

Each factor $\t{evo}_{iF_j}$ is a diffusor update as in \cref{slem:diffevo}, applied with rescaled angle $n_j \theta_t$. The total simulation cost is therefore $J$ sequential diffusor updates, scaling linearly in the number of distinct eigenvalues $J$.

\section{Dynamics of Digitalized Hamiltonians}\label{sec:dyn_digital}

Exponentiation of Hamiltonians that do not admit a compact eigendecomposition can be approximated by decomposing them into a sequence of simpler gates. A common approach is the Trotter Suzuki product formula~\cite{childs2021theory,yi_spectral_2022}, where for $H = \sum_{k} H_k$ with $H_{k}$ over few qubits the expansion is
\begin{align}
e^{i\theta H} \approx \prod_{k} e^{i\theta H_k / r}
\end{align}
for some repeated $r$ times per time step, with error bounded by the magnitude of the commutators $[H_j, H_k]$~\cite{haah2021quantum}. Randomized compilers such as qDRIFT~\cite{campbell2019random} sample terms probabilistically and can reduce gate counts for certain Hamiltonians~\cite{babbush2019quantum}. Exact fixed depth decompositions via Cartan decomposition~\cite{kokcu2022fixed} are also available for structured operators.

In all cases, each resulting gate $e^{i\theta H_k}$ is a Pauli rotation (\cref{sec:dyn_diffmixers}), a projector gate (\cref{sec:dyn_diffmixers}), or a local eigendecomposable operator (\cref{sec:dyn_eigham}), and the DOS simulation of each gate applies exactly. The only approximation error is that of the digitalization scheme itself, not the DOS simulation.

\pagebreak

\section{Dynamic Observable Subspace Adjoint Method: Efficient Gradients in Dynamic Subspace}\label{sec:dos_adjoint}

\begin{figure}[!t]
\begin{algobox}
\setlength{\intextsep}{0.5em}

\captionof{algorithm}{\tsf{Subspace Adjoint Method (Pauli)}}\label{alg:calc_grad}
\vspace{-0.9em}
\hrule
\vspace{0.2em}
\begin{algorithmic}[1]
\Statex \textbf{Inputs: } $ \trho $, $ \bmt $, $ \t{projs} $, $ \t{adjs} $, $ \tobs $
\vspace{0.2em}
\State $ \trho \rgt U(\bmt) \, \trho \, U^\dg(\bmt) $ \Comment{Prepare with Alg.~1}
\State $ \tsig \rgt \tobs $
\State $ \nabla_{\bmt} := (0,\ldots,0) $
\State $ t := LK $
\While{$ t > 0 $}
    \State $ k := K $
    \While{ $ k > 0 $ }
        \State $ \trho \rgt \trho + (\cos(-2 \bmt_{t}) - 1) \, \tprojs[k] \, \trho + \frac{1}{2} \sin(-2 \bmt_{t}) \, \tadjs[k] \, \trho $
        \State $ \tsig \rgt \tsig + (\cos(-2 \bmt_{t}) - 1) \, \tprojs[k] \, \tsig + \frac{1}{2} \sin(-2 \bmt_{t}) \, \tadjs[k] \, \tsig $
        \State $ \nabla_{\bmt}[t] \rgt \tsig \, \cdot \, (\tadjs[k] \, \trho) $
        \State $ t \rgt t - 1 $
        \State $ k \rgt k - 1 $
    \EndWhile
\EndWhile
\State \textbf{return} $ \nabla_{\bmt} $
\end{algorithmic}
\end{algobox}

\vspace{-1.0em}

\captionof*{algorithm}{{Alg.~\ref*{alg:calc_grad}}: \tsbf{Dynamic Observable Subspace Adjoint Method.} This algorithm computes $ \nabla_{\bmt} \, \L $ in a single back pass.}

\setlength{\intextsep}{\ointextsep}
\end{figure}

\subsection{Parameters on each Gate}\label{ssec:grad_pergate}

In this section, we describe a procedure to efficiently compute a gradient vector $ \nabla_{\bmt} \mathcal{L} = \l( \frac{\partial \L}{\partial \theta_1}, \ldots, \frac{\partial \L}{\partial \theta_T} \r) $. The method is based on the self-adjoint method, which was also used in Ref.~\cite{goh_lie-algebraic_2023} to efficiently compute gradients. As in the main manuscript, we begin by defining the loss function based on the energy associated with an observable $i O$.
\begin{align}
\L(\theta) &= -\t{Tr}\l( i O \, U(\bmt) \, i \rho \, U^\dg(\bmt) \r)
\end{align}

\noindent Define the unitary applied at time $t$ and the associated observable and density operator at time $t$ as:
\begin{align}
U_t(\theta_t) &= e^{i \theta_t G_t} \\
U(\bmt) &= \prod_{t=1}^{T} U_t(\theta_t) = \prod_{l=1}^{L} \prod_{k=1}^{K} U_{k+K(l-1)}(\theta_{k+K(l-1)}) \\
i O(t) &= U_{t+1:LK}^\dg(\bmt) \, i O \, U_{t+1:LK}(\bmt) \\
i \rho(t) &= U_{1:t}(\bmt) \, i \rho \, U_{1:t}^\dg(\bmt) ,
\end{align}
following the existing discussion in the main manuscript.

From \cref{sthm:dosrep}, conjugation by $U \in e^\g$ preserves the DOS, so $U^\dg \, iO \, U \in \orb$. From \cref{sthm:dosrho}, projection commutes with unitary conjugation. Together these give the following key lemma.

\begin{slem}\label{slem:traceindos}
Given any $ X $ and $ U \in e^\g $,
\begin{align}
\t{Tr}\l( U^\dg \, i O \, U \, X \r) = \t{Tr}\l( \t{proj}_{\g}^{iO}\l( U^\dg \, i O \, U \r) \, \t{proj}_{\g}^{iO}\l( X \r) \r)
\end{align}
\end{slem}

\begin{proof}
\begin{align}
\t{Tr}\l( U^\dg \, i O \, U \, X \r) &= \t{Tr}\l( U^\dg \, i O \, U \, \l( \t{proj}_{\g}^{iO}\l( X \r) + \t{proj}_{\g}^{O\perp} \l( X \r) \r) \r) \\
&= \t{Tr}\l( U^\dg \, i O \, U \, \t{proj}_{\g}^{iO}\l( X \r) \r) \nonumber \\
&\peq + \t{Tr} \l( U^\dg \, iO \, U \, \t{proj}_{\g}^{O\perp} \l( X \r) \r) \\
&= \t{Tr}\l( U^\dg \, i O \, U \, \t{proj}_{\g}^{iO}\l( X \r) \r) \\
&= \t{Tr}\l( \t{proj}_{\g}^{iO} \l( U^\dg \, iO \, U \r) \, \t{proj}_{\g}^{iO} \l( X \r) \r)
\end{align}
\end{proof}

From \cref{thm4:heicut} (\cref{sthm:heicut}), we show that the loss is computable at any time slice inside the DOS.

\begin{sthm}[Loss inside DOS]\label{sthm:lossdos}
Given a loss function $\L = -\t{Tr}\l( i O \, U(\bmt) \, i \rho \, U^\dg(\bmt) \r) $,
\begin{align}
\L &= \la i O(t) , i \vr(t) \ra_{\dos}
\end{align}
for any $t \in \{ 1, \ldots, LK \} $.
\end{sthm}

\begin{proof}
\begin{align}
\L &= -\t{Tr}\l( i O \, U(\bmt) \, i \rho \, U^\dg(\bmt) \r) \\
&= -\t{Tr}\l( i O \, U(\bmt_{t+1:LK}) U(\bmt_{1:t}) \, i \rho \, U^\dg(\bmt_{1:t}) U^\dg(\bmt_{t+1:LK}) \r) \\
&= -\t{Tr}\l( U_{t+1:LK}^\dg \l( \bmt \r) \, i O \, U_{t+1:LK} \l( \bmt \r) \,  U_{1:t}\l( \bmt \r) i \rho \, U_{1:t}^\dg \l( \bmt \r) \r) & (\t{trace cyclicity}) \\
&= -\t{Tr}\l( i O(t) \, i \vr(t) \r) & (\t{definition of } i O(t), i \rho(t)) \\
&= \la i O(t), i \vr(t) \ra_{\dos} & (\t{from \cref{sthm:heicut}})
\end{align}
\end{proof}

Then the gradient of the loss is the vector of inner products in the DOS.

\begin{suppfigure*}[!t]
\centering
\includegraphics[width=0.65\textwidth]{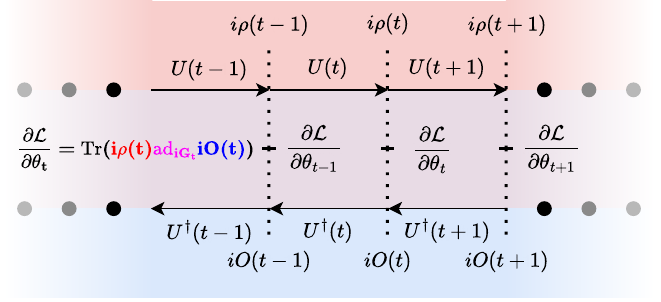}
 \caption{\tsbf{Meeting of Schr{\"o}dinger $ i\rho(t) $ and Heisenberg $ iO(t) $ at time $ t $} (reprint of \cref{fig:schrohei}). In the mixed Schr{\"o}dinger Heisenberg picture, $ i\rho(t) $ and $ iO(t) $ meet at time $ t $ through the action of generator $ G_t $, allowing us to compute each entry $  \frac{\partial}{\partial \theta_t} \mathcal{L} = \la i \vr(t) \, G_t , i O(t) \ra $ in a single pass of (reverse) evolving $ i \vr $ and $ i O $.}
\label{suppfig:schrohei}
\end{suppfigure*}

\begin{sthm}\label{sthm:grad}
Given $ \mathcal{L}(\bmt) = -\t{Tr}\l( i O \, U(\bmt) \, i \vr(0) \, U^\dg(\bmt) \r) $, then
\begin{adjustwidth}{-2em}{-2em}
\begin{align}
\nabla_{\bmt} \, \mathcal{L} &= \l( \la i O(1), \t{ad}_{i G_{1}}\l( i \vr(1) \r) \ra_{\dos},  \la i O(2), \t{ad}_{i G_{2}}\l( i \vr(2) \r) \ra_{\dos},    \ldots, \la i O(LK), \t{ad}_{i G_{LK}} \l( i \vr(LK) \r) \ra_{\dos} \r)
\end{align}
\end{adjustwidth}
\end{sthm}

\begin{proof}

Recall that $\frac{\partial }{\partial \theta_t} e^{ \pm i \theta_t G_t} = \pm i G_t \, e^{ \pm i \theta_t G_t} $.

\begin{align}
\frac{\partial \mathcal{L}}{\partial \theta_t}
&= -\frac{\partial}{\partial \theta_t} \t{Tr}\l( i O \, U(\bmt) \, i \vr \, U^{\dg}(\bmt) \r) \\
&= -\Tr{ i O \, U_{t+1:LK}(\bmt) \, \frac{\partial U_{t}(\theta_t)}{\partial \theta_t} \, U_{1:t-1}(\bmt) \, i \vr \, U^\dg(\bmt) } \nonumber \\
&\peq - \Tr{ i O \, U(\bmt) \, i \vr \, U_{1:t-1}^\dg(\bmt) \, \frac{\partial U_{t}^\dg(\theta_{t})}{\partial \theta_t} \, U_{t+1:LK}^\dg(\bmt) } \\
&= -\Tr{ U^\dg(\bmt_{t+1:LK}) \, i O \, U(\bmt_{t+1:LK}) \, \frac{\partial U_{t}(\theta_{t})}{\partial \theta_t} \, U(\bmt_{1:t-1}) \, i\vr \, U^\dg(\bmt_{1:t-1}) \, U_t^\dg } & \t{(trace cyclicity)} \nonumber \\
&\peq - \Tr{ U^\dg(\bmt_{t+1:LK}) \, i O \, U(\bmt_{t+1:LK}) \, U_t \, U(\bmt_{1:t-1}) \, i \vr \, U^\dg(\bmt_{1:t-1}) \, \frac{\partial U_{t}^\dg(\theta_{t})}{\partial \theta_t} } \\
&= -\Tr{ i O(t) \, \frac{\partial U_{t}(\theta_{t})}{\partial \theta_t} \, i \vr(t-1) \, U_t^\dg } - \Tr{ i O(t) \, U_t \, i \vr(t-1) \, \frac{\partial U_{t}^\dg(\theta_{t})}{\partial \theta_t} } \\
&= -\Tr{ i O(t) \, i G_{t} \, U_{t} \, i \vr(t-1) \, U_{t}^\dg } + \Tr{ i O(t) \, U_{t} \, i \vr(t-1) \, i G_{t} \, U_{t}^\dg } & \l( \tfrac{\partial U_t}{\partial \theta_t} = iG_t U_t \r) \\
&= -\Tr{i O(t) \, \l[ i G_{t} , i \vr(t) \r] } \\
&= -\Tr{i O(t) \, \t{ad}_{i G_{t}}\l( i \vr(t) \r) } \\
&= \la i O(t) , \t{ad}_{i G_{t}}\l( i \vr(t) \r) \ra_{\dos}.
\end{align}

\end{proof}

As a result, \cref{alg:calc_grad} computes the gradient as follows. We compute $i\rho(LK)$ through simulation as \cref{alg:sim_dympauli} and have $iO(LK) = iO$. Then we compute $i\rho(t)$ and $iO(t)$ from time slice $t=LK$ \textit{backwards} to $t=1$. Using \cref{sthm:grad}, we can compute the $\frac{\partial \L}{\partial \theta_t}$ at every time slice $t$ using $i\rho(t), iO(t)$ and the generator at this time $iG_{t}$.

\begin{sthm}[Gradient Complexity]\label{sthm:gradcomplexity}
\cref{alg:calc_grad} computes $\nabla_{\bmt} \, \L $ in $\Oc\l( LK \t{dim}\l( \orb \r) \r)$.
\end{sthm}

\begin{proof}
Since we can compute $i\rho(LK)$ in $LK$ steps. At each step, we compute $i\rho(t)$ using $i\rho(t-1)$ and evolution given by \cref{sthm:pauliupdate}. Then for these sparse matrices, each step is computable in $\Oc\l( \t{dim}\l( \orb \r) \r)$ or better. Then in total this require $\Oc\l( LK \t{dim}\l( \orb \r) \r)$. Each move to $i \rho(t)$ and $i O(t)$ as well as computing $\la i O(t), \t{ad}_{iG}\l( i \rho(t) \r) \ra_{\dos}$ is in $\Oc\l( \t{dim}\l( \orb \r) \r)$ leading to $\Oc\l( LK \t{dim}\l( \orb \r) \r)$ runtime total.
\end{proof}

\subsection{Shared Parameters on Time Slices of Gates}\label{ssec:grad_shared}

\begin{figure}[!t]
\begin{algobox}

\setlength{\intextsep}{0.5em}

\captionof{algorithm}{\tsf{Subspace Adjoint Method with Shared Angle Commutative Blocks}}\label{alg:calc_shared_grad}
\vspace{-0.9em}
\hrule
\vspace{0.2em}
\begin{algorithmic}[1]
\Statex \textbf{Inputs: } $ \trho $, $ \bmt $, $ \tt{projs} $, $ \tt{adjs} $, $ \tt{obs} $, $\tt{f} $
\State Prepare $ \trho \rgt U(\bmt) \, \trho \, U^\dg(\bmt) $ with Alg.~1
\State $ \tsig \rgt \tobs $
\State $ \nabla_{\bmt} := (0,\ldots,0) $
\State $ t := | \bmt | $
\While{$ t > 0 $}
    \State $ k := K $
    \While{ $ k > 0 $ }
        \State $ \trho \rgt \trho + (\cos(-2 \bmt_{t} \tt{f}[k]) - 1) \, \tprojs[k] \, \trho + \frac{1}{2} \sin(-2 \bmt_{t} \tt{f}[k]) \, \tadjs[k] \, \trho $
        \State $ \tsig \rgt \tsig + (\cos(-2 \bmt_{t} \tt{f}[k]) - 1) \, \tprojs[k] \, \tsig + \frac{1}{2} \sin(-2 \bmt_{t} \tt{f}[k]) \, \tadjs[k] \, \tsig $
        \State $ \nabla_{\bmt}[t] \rgt \nabla_{\bmt}[t] + \tt{f}[k] \tsig \, \cdot \, (\tadjs[k] \, \trho) $
        \State $ k \rgt k - 1 $
    \EndWhile
    \State $ t \rgt t - 1 $
\EndWhile
\State \textbf{return } $ \nabla_{\bmt} $
\end{algorithmic}

\end{algobox}

\vspace{-1.0em}
\captionof*{algorithm}{{Alg.~\ref*{alg:calc_shared_grad}}: \tsbf{Subspace Adjoint Method with Shared Angle Commutative Blocks.} This algorithm computes $ \nabla_{\bmt} \, \mathcal{L} $ in a single back pass. $\tt{f}$ holds fixed multiplicative parameters.}

\setlength{\intextsep}{\ointextsep}

\end{figure}

In many important settings, we do not have a single parameter for each gate, but instead share parameters amongst gates that are contiguous in time such that they can be lumped together into a single time slice.

We consider the case in which for each layer, we have a collection of blocks $B_{1}, \ldots, B_{K}$ such that each block is a subset of the gate generators $B_{j} \subseteq \G $. Then the unitary associated with this PQC is
\begin{align}
U(\theta) = \prod_{l=1}^{L} \prod_{k=1}^{K} \prod_{b=1}^{B_k} e^{i \theta_{l,k} G_{k,b} } .
\end{align}

An important example of this are the standard constructions of QAOA and HVA. The following theorem captures how to efficiently compute the gradient in this framework.

\begin{sthm}
Given $ \mathcal{L}(\bmt) = -\t{Tr}\l( i O \, U(\bmt) \, i \vr(0) \, U^{\dg}(\bmt) \r) $, then
\begin{adjustwidth}{-2em}{-2em}
\begin{align}
\nabla_{\bmt} \, \mathcal{L} &= \l( \sum_{b=1}^{B_1} \la i O(b), \t{ad}_{i G_{b}}\l( i \vr(b) \r) \ra_{\dos},  \sum_{b=1}^{B_2} \la i O(B_{\Sigma1}+b), \t{ad}_{i G_{B_{\Sigma1}+b}}\l( i \vr(B_{\Sigma1}+b) \r) \ra_{\dos}, \r. \\
&\qquad \l. \ldots, \sum_{b=1}^{B_K} \la i O(B_{\Sigma (K-1)} + b), \t{ad}_{i G_{B_{\Sigma(K-1)+b}}} \l( i \vr(B_{\Sigma(K-1)}+b) \r) \ra_{\dos} \r) \\
&= \l( \sum_{b=1}^{B_1} \la \t{proj}_{\g}^{iO}\l( i O(b) \r), \t{ad}_{i G_{b}}\l( \t{proj}_{\g}^{iO} \l( i \vr(b) \r) \r) \ra_{\dos}, \r. \\
&\qquad \l. \ldots, \sum_{b=1}^{B_K} \la \t{proj}_{\g}^{iO}\l( i O(B_{\Sigma(K-1)}+b) \r), \t{ad}_{i G_{B_{\Sigma(K-1)}+b}} \l( \t{proj}_{\g}^{iO} \l( i \vr(B_{\Sigma(K-1)}+b) \r) \r) \ra_{\dos} \r),
\end{align}
\end{adjustwidth}
with $ B_{\Sigma r} = \sum_{p=1}^{r} B_{q} $.
\end{sthm}

\begin{proof}
Since $\theta_{l,k}$ is shared across all gates $G_{k,b}$ in block $k$ of layer $l$, the chain rule gives
\begin{align}
\frac{\partial \L}{\partial \theta_{l,k}} = \sum_{b=1}^{B_k} \frac{\partial \L}{\partial \theta_{b'}}\bigg|_{\theta_{b'} = \theta_{l,k}},
\end{align}
where $b' = K(l-1) + B_{\Sigma(k-1)} + b$ is the flat time index of gate $(l,k,b)$. Each term is a well defined gradient for an individual parameter, and by \cref{sthm:grad},
\begin{align}
\frac{\partial \L}{\partial \theta_{b'}} = \la i O(b'), \t{ad}_{i G_{k,b}}\l( i \vr(b') \r) \ra_{\dos}.
\end{align}
Summing over $b = 1, \ldots, B_k$ for each block $k$ yields the result.
\end{proof}

\subsection{Variance from independent circuit pieces}\label{ssec:gradvar}

We separate the fixed subspace $V_{\t{fix}}$ first because it contributes zero to the gradient. Removing its constant contribution to the loss leaves the subspace $W$, which is where we compute the gradient variance.

To separate these components, take the conjugation action of $e^\g$ on the DOS $\orb$. Its fixed subspace and the remaining subspace are
\begin{align}
V_{\t{fix}} &= \l\{ iB \in \orb : U \, iB \, U^\dg = iB \ \t{for every } U \in e^\g \r\}, \\
W &= \orb \cap V_{\t{fix}}^\perp.
\end{align}
\providecommand{\gradientperpref}{ as in \cref{eq:spaceperp}}
The decomposition $\orb = V_{\t{fix}} \oplus W$ is orthogonal with respect to the Hilbert--Schmidt inner product\gradientperpref. For any $iB \in V_{\t{fix}}$ and $iC \in W$,
\begin{align}
\la iB, U \, iC \, U^\dg\ra_{\hs}  &= \la U^\dg \, iB \, U, iC\ra_{\hs} = \la iB, iC\ra_{\hs} = 0.
\end{align}
Thus both subspaces are preserved by conjugation. In the time-cut expression $\L = \la iO(t), i \vr(t)\ra_{\dos}$, the two mixed inner products vanish by orthogonality. The inner product of the two fixed components is independent of the gate parameters. The gradient therefore comes entirely from the inner product of the two $W$ components.

The definition of $W$ only removes the fixed directions. It may still contain several invariant subspaces. \Cref{cor:grad_var_bound} assumes that $W$ is nonzero and irreducible so that Haar averaging on $W$ gives a single scalar covariance.

To state the moment matching assumptions, let $\bmt \sim \lambda$, where $\lambda$ is the distribution used to sample the circuit angles. Fix the cut after gate $t$. For a measurable set $S$ of unitaries, the prefix and inverse suffix have distributions
\begin{align}
\nu_\rho(S) &= \Pr_{\bmt \sim \lambda}\l[ U_{1:t}(\bmt) \in S \r], \nonumber \\
\nu_O(S) &= \Pr_{\bmt \sim \lambda}\l[ U_{t+1:T}^{\dg}(\bmt) \in S \r]. \nonumber
\end{align}
The prefix evolves the state forward, and the inverse suffix evolves the observable backward. Let $\mu$ be the Haar probability measure on $e^\g$~\cite{haar1933massbegriff}. For fixed $V \in e^\g$ and measurable $S \subseteq e^\g$, write $VS = \l\{ VU:U \in S \r\}$. Left invariance of Haar measure gives the following identity (see also \cite[Sec.~7.2, Eq.~(7.129)]{watrous2018theory} for the unitary group)
\begin{align}
\mu(VS) &= \mu(S). \nonumber
\end{align}
Equivalently, if $U$ is sampled from $\mu$, then $VU$ also has distribution $\mu$. For a fixed $iB \in W$, write its conjugate as
\begin{align}
\t{evo}_{U}\l( iB \r) &= U \, iB \, U^\dg.
\end{align}
We use the superoperator outer product $\oask$ from \cref{sdef:oask}.
Unitary designs reproduce Haar moments up to a specified order~\cite{dankert2009exact}. Ragone et al.~\cite{ragone2024lie} use this property to evaluate circuit moments. Here we assume that each circuit piece matches the first two Haar moments of conjugation on $W$. For each $\eta\in\{\nu_\rho,\nu_O\}$ and every $iB\in W$,
\begin{align}
\E_{U \sim \eta}\l[ \t{evo}_{U}\l( iB \r) \r]
&= \E_{U \sim \mu}\l[ \t{evo}_{U}\l( iB \r) \r], \label{eq:haar_1st_mom} \\
\E_{U \sim \eta}\l[ \l( \t{evo}_{U}\l( iB \r) \oask \t{evo}_{U}\l( iB \r) \r) \r]
&= \E_{U \sim \mu}\l[ \l( \t{evo}_{U}\l( iB \r) \oask \t{evo}_{U}\l( iB \r) \r) \r]. \label{eq:haar_2nd_mom}
\end{align}
We also assume that $U_{1:t}$ and $U_{t+1:T}^{\dg}$ are independent under $\lambda$. Their joint law is then $\nu_\rho\otimes\nu_O$, so we can average the state and observable outer products separately.

\begin{scor}[Loss and gradient variance in one irreducible sector, \cref{cor:grad_var_bound}]
\DOSVarianceStatement
\end{scor}

\begin{proof}[Proof of \cref{cor:grad_var_bound}]
By definition, $iB \in V_{\t{fix}}$ means $U \, iB \, U^\dg = iB$ for every $U \in e^\g$. Since $iG_t \in \g$, we can take $U = e^{i\alpha G_t}$ for any real $\alpha$. Hence
\begin{align}
e^{i\alpha G_t} \, iB \, e^{-i\alpha G_t} &= iB \qquad\t{for every }\alpha \in \R.
\end{align}
Multiply on the right by $e^{i\alpha G_t}$ to move the inverse exponential to the other side:
\begin{align}
e^{i\alpha G_t} \, iB &= iB \, e^{i\alpha G_t}.
\end{align}
Differentiating both sides at $\alpha = 0$ gives
\begin{align}
iG_t \, iB &= iB \, iG_t.
\end{align}
Therefore
\begin{align}
[iG_t, iB] &= 0.
\end{align}
For $iC \in W$, it lies in $\orb$ and $V_{\t{fix}}^\perp$. Then the difference quotient
\begin{align}
\frac{e^{i\alpha G_t} \, iC \, e^{-i\alpha G_t}-iC}{\alpha} \qquad (\alpha \ne 0)
\end{align}
is in the linear subspace $W$. Specifically, its limit is $[iG_t, iC]$, which is still in $W$ because $W$ is finite-dimensional and closed.

Now decompose $i \vr(t)$ into its fixed component and its $W$ component. The commutator kills the fixed component and sends the $W$ component back into $W$. In the inner product with $iO(t)$, the fixed component of $iO(t)$ is orthogonal to this result. This leaves
\begin{align}
\frac{\partial \L}{\partial \theta_t}  &= \la iO_W(t), \l[ iG_t, i \vr_W(t) \r] \ra_{\hs} ,
\end{align}
with
\begin{align}
i \vr_W(t) &= \t{proj}_{W}\l( i \vr(t) \r) \\
iO_W(t) &= \t{proj}_{W}\l( iO(t) \r) .
\end{align}

Similarly, let $\mc{A}$ be the restriction of $\t{ad}_{iG_t}$ to $W$,
\begin{align}
\mc{A} &= \t{proj}_{W} \, \t{ad}_{iG_t} \, \t{proj}_{W}.
\end{align}

Conjugation commutes with $\t{proj}_{W}$ and preserves the Hilbert--Schmidt norm. Since $W \subseteq \orb$, projecting onto the DOS first also leaves the $W$ component unchanged. Therefore
\begin{align}
\l\| i \vr_W(t) \r\|_{\hs}^2
&= \l\| \t{proj}_{W}\l( i \rho(0) \r) \r\|_{\hs}^2 = q_\rho, \nonumber \\
\l\| iO_W(t) \r\|_{\hs}^2
&= \l\| \t{proj}_{W}\l( iO \r) \r\|_{\hs}^2 = q_O. \nonumber
\end{align}

For any $iB \in W$, define its Haar mean by
\begin{align}
i \overline{B} &= \E_{U \sim \mu}\l[ \t{evo}_{U}\l( iB \r) \r] \nonumber \\
&= \int_{e^\g} U \, iB \, U^\dg \, \t{d}\mu(U).
\end{align}
Every conjugate in this average is in $W$, so $i \overline{B} \in W$. For any fixed $V \in e^\g$,
\begin{align}
V \, i \overline{B} \, V^\dg
&= \int_{e^\g} (VU) \, iB \, (VU)^\dg \, \t{d}\mu(U) \nonumber \\
&= \int_{e^\g} U' \, iB \, (U')^\dg \, \t{d}\mu(U') \nonumber \\
&= i \overline{B}.
\end{align}
The substitution $U' = VU$ preserves the measure by left invariance of the Haar measure. This shows that $i \overline{B}$ also belongs to $V_{\t{fix}}$. Since $W \cap V_{\t{fix}} = \{0\}$, the Haar mean is zero.

The Haar average of the superoperator outer product of an operator with itself is self-adjoint and commutes with the group action. Its eigenspaces are invariant, and so irreducibility forces it to be a scalar multiple of the identity superoperator. The trace of the Haar average is the squared norm of the operator, so the scalar is that squared norm divided by $d = \t{dim}\l( W \r)$.

In each $\E_\mu$ below, the same evolution is evaluated with a Haar draw in place of the corresponding circuit draw. Second moment matching (\cref{eq:haar_2nd_mom}) means the circuit average of each superoperator outer product equals the Haar average:
\begin{align}
\E_{\nu_\rho}[\l( i \vr_W(t) \oask i \vr_W(t) \r)]
&= \E_\mu[\l( i \vr_W(t) \oask i \vr_W(t) \r)] \nonumber \\
&= \frac{q_\rho}{d} \, \mc{I}, \label{eq:rhomoment} \\
\E_{\nu_O}[\l( iO_W(t) \oask iO_W(t) \r)]
&= \E_\mu[\l( iO_W(t) \oask iO_W(t) \r)] \nonumber \\
&= \frac{q_O}{d} \, \mc{I}, \label{eq:omoment}
\end{align}
with $\mc{I}$ as the identity superoperator on $W$. For convenience, let $\mc{M}_\rho$ be the second moment of the projected state and let $\mc{M}_O$ be the second moment of the projected observable:
\begin{align}
\mc{M}_\rho &= \E_{\nu_\rho}\l[ i \vr_W(t) \oask i \vr_W(t) \r], \\
\mc{M}_O &= \E_{\nu_O}\l[ iO_W(t) \oask iO_W(t) \r].
\end{align}
These are superoperators on $W$ that will enter the trace for the squared gradient. First moment matching gives
\begin{align}
\E_{\nu_\rho}\l[ i \vr_W(t) \r] &= \E_\mu\l[ i \vr_W(t) \r] = 0, \\
\E_{\nu_O}\l[ iO_W(t) \r] &= \E_\mu\l[ iO_W(t) \r] = 0. \label{eq:omean}
\end{align}

Now, take the expectation over the derivative of the loss for a specific parameter. Hold $U_{1:t}$ fixed and average over the inverse suffix $U_{t+1:T}^{\dg}$ first. Note that these unitaries are drawn independent of each other.
\begin{align}
\E\l[ \frac{\partial \L}{\partial \theta_t} \r]
&= \E_{\nu_\rho}\l[ \E_{\nu_O} \l[ \la iO_W(t), \mc{A}\l( i \vr_W(t) \r)\ra_{\hs} \r] \r] \nonumber \\
&= \E_{\nu_\rho}\l[ \la \E_{\nu_O}[iO_W(t)], \mc{A}\l( i \vr_W(t) \r)\ra_{\hs} \r] \nonumber \\
&= \E_{\nu_\rho}\l[ \la 0, \mc{A}\l( i \vr_W(t) \r) \ra_{\hs} \r] & \t{(from \cref{eq:omean})} \nonumber \\
&= 0.
\end{align}
At the cut, the loss is the ordinary matrix trace $\L = -\t{Tr}\l( iO(t) \, i\rho(t) \r)$. Let $\L_{\t{fix}}$ be the constant contribution from the fixed subspace. The rest of the loss can be written using either trace:
\begin{align}
\L-\L_{\t{fix}}
&= -\t{Tr}\l( iO_W(t) \, i \vr_W(t) \r) \nonumber \\
&= \la iO_W(t), i \vr_W(t)\ra_{\hs} \nonumber \\
&= \t{Tr}_W\l( i \vr_W(t) \oask iO_W(t) \r). \nonumber
\end{align}
The first trace is over the quantum Hilbert space. The last trace is over the space $W$ of skew-Hermitian operators and uses \cref{slem:oasktrace}.

The derivative of the loss for a specific parameter depends on the adjoint superoperator $\mc{A}$ for that specific gate. Composing $\mc{A}$ with the outer product gives
\begin{align}
\mc{A} \, \l( i \vr_W(t) \oask iO_W(t) \r)  &= \mc{A}\l( i \vr_W(t) \r) \oask iO_W(t).
\end{align}
Applying \cref{slem:oasktrace} on $W$ gives
\begin{align}
\frac{\partial \L}{\partial \theta_t}
&= \la iO_W(t), \mc{A}\l( i \vr_W(t) \r)\ra_{\hs} \nonumber \\
&= \t{Tr}_W\l( \mc{A} \, \l( i \vr_W(t) \oask iO_W(t) \r) \r), \label{eq:gradtrace}
\end{align}

which is for superoperators the same content as $\bra{\psi} H \ket{\psi} = \t{Tr}\l( H \ketbra{\psi}{\psi} \r)$ for wavefunctions.

Squaring the scalar trace and using \cref{slem:oasktrace} gives
\begin{align}
&\l[ \t{Tr}_W\l( \mc{A} \, \l( i \vr_W(t) \oask iO_W(t) \r) \r) \r]^2 & \t{(from \cref{eq:gradtrace})} \nonumber \\
&= \la i \vr_W(t), \t{Tr}_W\l( \mc{A} \, \l( i \vr_W(t) \oask iO_W(t) \r) \r) \, \mc{A}^\dg\l( iO_W(t) \r)\ra_{\hs} \nonumber \\
&= \la i \vr_W(t), \la iO_W(t), \mc{A}\l( i \vr_W(t) \r)\ra_{\hs} \, \mc{A}^\dg\l( iO_W(t) \r)\ra_{\hs} \nonumber \\
&= \la i \vr_W(t), \mc{A}^\dg \, \l( iO_W(t) \oask iO_W(t) \r) \, \mc{A}\l( i \vr_W(t) \r) \ra_{\hs} \nonumber \\
&= \t{Tr}_W\l( \mc{A}^{\dg} \, \l( iO_W(t) \oask iO_W(t) \r) \, \mc{A} \, \l( i \vr_W(t) \oask i \vr_W(t) \r) \r). \label{eq:gradsq}
\end{align}
The first equality uses the adjoint relation to write one copy of the trace as $\la i \vr_W(t), \mc{A}^\dg\l( iO_W(t) \r)\ra_{\hs}$. The last line groups the state and observable factors so that we can average their outer products separately.
Independent sampling lets us average the observable and state outer products separately. Their expectations are $\mc{M}_O$ and $\mc{M}_\rho$, respectively. Using \cref{eq:rhomoment,eq:omoment} gives
\begin{align}
\E\l[ \l( \frac{\partial \L}{\partial \theta_t} \r)^2 \r]
&= \t{Tr}_W\l( \mc{A}^{\dg} \, \mc{M}_O \, \mc{A} \, \mc{M}_\rho \r) \\
&= \frac{q_\rho \, q_O}{d^2} \, \t{Tr}_W\l( \mc{A}^{\dg} \, \mc{A} \r) \\
&= \frac{q_\rho \, q_O}{d^2} \, \l\| \mc{A} \r\|_{\t{F}}^2 \\
&= \frac{q_\rho \, q_O \, \kappa_t}{d}, \label{eq:gradmoment}
\end{align}
where we introduced $\kappa_t = \l\| \mc{A} \r\|_{\t{F}}^2/d$. This agrees with the corollary because $\t{proj}_{W} \, \t{ad}_{iG_t} \, \t{proj}_{W}$ restricts to $\mc{A}$ on $W$ and vanishes on $W^\perp$. Since the mean derivative is zero, its second moment is precisely its variance.

Now we consider the variance of the loss. Since $\L_{fix}$ is a constant, it does not contribute to the variance:
\begin{align}
\t{Var}\l( \L \r) &= \t{Var}\l( \L_W \r) \\
&= \E\l[ \L_W^2 \r] - \E\l[ \L_W \r]^2 .
\end{align}

The second term is zero from the first Haar moment:
\begin{align}
\E\l[ \L_W \r] &= \E_{\nu_\rho\otimes\nu_O}\l[ \la iO_W(t), i \vr_W(t) \ra_{\hs} \r] \\
&= \la \E_{\nu_{O}}\l[ iO_W(t) \r], \E_{\nu_{\rho}}\l[ i \vr_W(t) \r] \ra_{\hs} \\
&= \la 0, 0 \ra_{\hs} \\
&= 0 .
\end{align}
Then moving the expectations into the expansion of the variance expression as before gives:
\begin{align}
\t{Var}\l( \L \r) &= \E\l[ \L_W^2 \r] \\
&= \t{Tr}_W\l( \mc{M}_O \, \mc{M}_\rho \r)  \\
&= \frac{q_\rho \, q_O}{d^2} \, \t{Tr}_W\l( \mc{I} \r) \\
&= \frac{q_\rho \, q_O}{d} . & (\l\| \mc{I} \r\|_{\t{F}}^2 = d)
\end{align}

If $iO$ lies in $\Qc_\g^k$, commutation with a generator preserves that product span by the Leibniz rule. Hence $d \leq \t{dim}\l( \orb \r) \leq \t{dim}\l( \g \r)^k$, which gives the lower bounds:
\begin{align}
\t{Var}\l[ \L \r] &\geq \frac{ q_\rho q_O }{ \t{dim}\l( \g \r)^k } , \\
\t{Var}\l[ \frac{\partial \L}{\partial \theta_t} \r] &\geq \frac{ q_\rho q_O \kappa_t }{ \t{dim}\l( \g \r)^k } .
\end{align}
\end{proof}

\pagebreak

\part{\sffamily Example Quantum System Simulation inside the Dynamic Observable Subspace}\label{part3:exp}

\section{TFXY DLA, Observable Orbitals, and Initial DOS States}\label{sec:tfxy_pauli}

This section develops the TFXY DLA, its DLA and $\Qc_{\g}^{2}$ observable orbitals, and the corresponding initial DOS vectors. The construction applies to any Hamiltonian with the same Pauli support, including $H_{SK,XX,YY}$ and $H_{SKZ,XX,YY}$ from \results.

The gate generators are
\begin{align}
\G_{XX,YY,Z} = \l\{ X_{j} X_{j+1} , Y_{j} Y_{j+1} \r\}_{j=1}^{n-1} \cup \l\{ X_{1} X_{n} , Y_{1} Y_{n} \r\} \cup \{ Z_{j} \}_{j=1}^{n} \label{eq:tfxy_ring_generators}
\end{align}

The corresponding DLA $\g_{XX,YY,Z}$ has dimension $\Oc(n^2)$. Its basis $\B^{\g_{XX,YY,Z}}$ consists of single site $Z$ terms and, for each pair $j < k$, four families of Z string operators with $X$/$Y$ endpoints:
\begin{align}
\B^{\g_{XX,YY,Z}} &= \l\{ Z_j \r\}_{j} \nonumber \\
& \cup \l\{ X_j \l( \prod_{l=j+1}^{k-1} Z_l \r) X_k ,\;\; X_j \l( \prod_{l=j+1}^{k-1} Z_l \r) Y_k \r\}_{j < k} \nonumber \\
& \cup \l\{ Y_j \l( \prod_{l=j+1}^{k-1} Z_l \r) X_k ,\;\; Y_j \l( \prod_{l=j+1}^{k-1} Z_l \r) Y_k \r\}_{j < k} \label{eq:tfxy_endpoint_strings} , \\
\g_{XX,YY,Z} &= \t{span}\l( \B^{\g_{XX,YY,Z}} \r) \nonumber
\end{align}
The Z string $Z_{j+1} \cdots Z_{k-1}$ is empty (identity) when $k = j+1$, giving the nearest neighbor terms $X_j X_{j+1}$, $X_j Y_{j+1}$, $Y_j X_{j+1}$, $Y_j Y_{j+1}$. Crucially, $\g_{XX,YY,Z}$ contains no all $Z$ terms of weight $\geq 2$.

\subsection{DLA and $\Qc_{\g}^{2}$ Observable Orbitals}\label{sec:tfxy_zz_sector}

The following observable orbital construction is independent of parameter sharing. The TFXY Hamiltonian is composed of observables in the DLA:
\begin{align}
H_{TFXY} = \sum_{j} X_{j} X_{j+1} + \sum_{j} Y_{j} Y_{j+1} + \sum_{j} Z_{j}
\end{align}
Its DOS is therefore the DLA orbital $\t{orb}_{\g}^{iH_{TFXY}}$. For an observable containing $ZZ$ terms, such as the disordered models studied in \results, each pair of sites instead gives
\begin{align}
 i \, Z_j Z_k &\in \Qc_{\g}^{2}=\t{span}\l\{ iB_aB_b:B_a,B_b\in\B^{\g} \r\}.
\end{align}
The $ZZ$ target is therefore simulated in a separate orbital inside $\Qc_{\g}^{2}$. Since the DLA contains no all $Z$ strings of weight two or greater, this $ZZ$ orbital contains no all $Z$ strings of weight three or greater.

\subsection{Initial DOS Vectors}

In \results, we considered systems initialized to the all zero state:
\begin{align}
\ket{\psi(0)} &= \underbrace{\ket{0} \ket{0} \ldots \ket{0}}_{n}, \\
\rho(0) &= \ketbra{\psi(0)}{\psi(0)} = \prod_{j=1}^{n} \ketbra{0}{0}_{j} \\
&= \prod_{j=1}^{n} (\I+Z_{j})/2 \\
&= \l( \I+\sum_{j} Z_{j} + \sum_{j<k} Z_{j} Z_{k} + \ldots \r) / 2^{n}.
\end{align}
Then projecting the initial density operator onto the active DOS orbitals gives
\begin{align}
 i\vr_{\g}^{iZ}(0) &= \t{proj}_{\g}^{iZ}\l( i\rho(0) \r) = \sum_{j=1}^{n} i \, Z_{j} / 2^{n}, \label{eq:tfxy_input_z} \\
 i\vr_{\g}^{iZZ}(0) &= \t{proj}_{\g}^{iZZ}\l( i\rho(0) \r) = \frac{1}{2^n}\sum_{j<k} i \, Z_j Z_k. \label{eq:tfxy_input_zz}
\end{align}
The DLA orbital $\t{orb}_\g^{iH_{TFXY}}$ has dimension $\Oc(n^2)$, spanned by the $n$ single site $Z$ terms and $\Oc(n^2)$ operators of the form $X_j Z_{j+1} \cdots Z_{k-1} Y_k$. Its initial vector has support only on the $n$ single site $Z$ terms. The $ZZ$ initial vector similarly has support only on the two qubit $ZZ$ terms of its $\Qc_{\g}^{2}$ orbital. Thus a $ZZ$ target uses the two qubit components of the same product state. For a target with both $Z$ and $ZZ$ terms, the simulator evaluates the corresponding DOS orbitals and adds their energies and gradients as in \cref{thm12:multidosrep}.

\subsection{Bounding the Variance for the XY Open Chain}\label{ssec:tfxy_grad_var}

For this example, omit the $X_1 X_n$ and $Y_1 Y_n$ bonds from the generator set in \cref{eq:tfxy_ring_generators}. Applying \cref{cor:grad_var_bound} to the resulting open chain, we first bound the gate factor $\kappa_t$ and then compute the two projection weights in the variance formula.

We use the irreducible open-chain sectors for $n > 4$~\cite{diaz2023beyond}. Set $r = 1$ for $O = Z_j$ and $r = 2$ for $O = Z_j \, Z_k$. The open-chain DLA is isomorphic to $\mf{so}(2n)$ by \cref{eq:opendla}, and $iO \in \Qc_\g^r$. We only need the dimension bound
\begin{align}
d \le \l( \t{dim}\l( \g \r) \r)^r  &= \binom{2n}{2}^{r} = O(n^{2r}).
\end{align}

\begin{slem}[Open-chain commutator bound]\label{slem:kappa}
For either open-chain sector above and any single $Z$, $XX$, or $YY$ gate,
\begin{align}
\kappa_t \in \Omega(n^{-1}).
\end{align}
\end{slem}
\begin{proof}
For $\kappa_t$, it suffices to count a subset of the basis strings that anticommute with the gate generator. Let $Q_{ab}$ denote the string in \cref{eq:tfxy_endpoint_strings} with $X$ at both endpoints $a<b$. The $Z_j$ sector contains the $Q_{ab}$, and the $Z_j \, Z_k$ sector contains the $Q_{ab} \, Q_{ce}$ for $a < b < c < e$. In either family, a string anticommutes with $Z_\ell$ if $\ell$ is one of its endpoints. Fixing that endpoint leaves $n-1$ strings in the first case and $\binom{n-1}{3}$ in the second. Each anticommuting normalized Pauli string contributes $4$ to $\l\| \mc{A} \r\|_{\t{F}}^2$. Thus this squared norm is $\Omega(n^{2r-1})$ for a $Z$ gate.

\quad The same bound holds for $P = X_\ell \, X_{\ell+1}$ and $P = Y_\ell \, Y_{\ell+1}$. Take $R = Y_\ell \, X_{\ell+1}$ in the first case and $R = X_\ell \, Y_{\ell+1}$ in the second. Since $iR$ is proportional to $[iZ_\ell, iP]$, the unitary $V = e^{i\pi R/4}$ belongs to $e^\g$ and preserves $W$. It satisfies $V \, Z_\ell \, V^\dg = \pm P$, so the two commutators restricted to $W$ have the same Frobenius norm.

Then
\begin{align}
\kappa_t &= \frac{\| \mc{A} \|_{\t{F}}^2 }{ d } = \Omega\l( n^{2r - 1 - 2r} \r) = \Omega\l( n^{-1} \r) .
\end{align}

\end{proof}

\begin{trivlist}
\item\relax\sffamily
\textbf{\Cref{cor:tfxy_grad_var} (Gradient variance for the open chain).}
Initialize the open chain in $\ket{0}^{\otimes n}$. For any individual $XX$, $YY$, or $Z$ gate and under the sampling assumptions of \cref{cor:grad_var_bound},
\begin{align*}
\t{Var}\l( \frac{\partial \L}{\partial \theta_t} \r) \in \begin{cases} \Omega(n^{-2}), & O = Z_j, \\
\Omega(n^{-3}), & O = Z_j \, Z_k . \end{cases}
\end{align*}
\end{trivlist}

\begin{proof}[Proof of \cref{cor:tfxy_grad_var}]
In the open chain, the all-$Z$ strings in the sector are precisely the products of $r$ distinct $Z$ operators~\cite{diaz2023beyond}. Thus the projected input is the vector in \cref{eq:tfxy_input_z} for $r=1$ or \cref{eq:tfxy_input_zz} for $r=2$. Its $\binom{n}{r}$ Pauli strings are mutually orthogonal, each with coefficient $2^{-n}$ and squared Hilbert--Schmidt norm $2^n$. The observable is a single Pauli string in that sector. Hence
\begin{align}
q_\rho &= 2^{-n} \, \binom nr, \nonumber \\
q_O &= 2^n.
\end{align}

Using \cref{slem:kappa} and the dimension bound gives
\begin{align}
\t{Var}\l( \frac{\partial \L}{\partial \theta_t} \r)  &= \frac{ q_{\rho} \, q_{O} \, \kappa_t }{ d } \\
&= \frac{\binom{n}{r} \, \kappa_t}{d} \\
&= \Omega\l( n^{r-1-2r} \r) \\
&= \Omega\l( n^{-r-1} \r),
\end{align}
where $r$ is a low constant, such as $r\in \{ 1, 2 \}$ as in the corollary.
\end{proof}

\subsection{Shared TFXY Parameterization}\label{sec:tfxy_shared_app}

The shared parameter TFXY circuit assigns one parameter to each layer. It is related to Hamiltonian variational ans\"atze and QAOA circuits~\cite{wecker_progress_2015,wiersema_exploring_2020}.
\begin{align}
U_{S}(\bmt) = \prod_{\ell=1}^{L} \prod_{j=1}^{n} e^{i \theta_{\ell} Z_j } e^{i \theta_{\ell} X_j X_{j+1} } e^{i \theta_{\ell} Y_j Y_{j+1} } .
\end{align}
The shared parameter circuit has the same generator set, DOS, and initial DOS vectors as the individual parameter circuit. \Cref{fig:tfxy_indv_share} compares the two parameterizations and shows better preparation quality with individual parameters. The gradient formula for the shared parameters is given in \cref{ssec:grad_shared}.
\begin{figure*}[!t]
\centering
\includegraphics[width=0.72\textwidth]{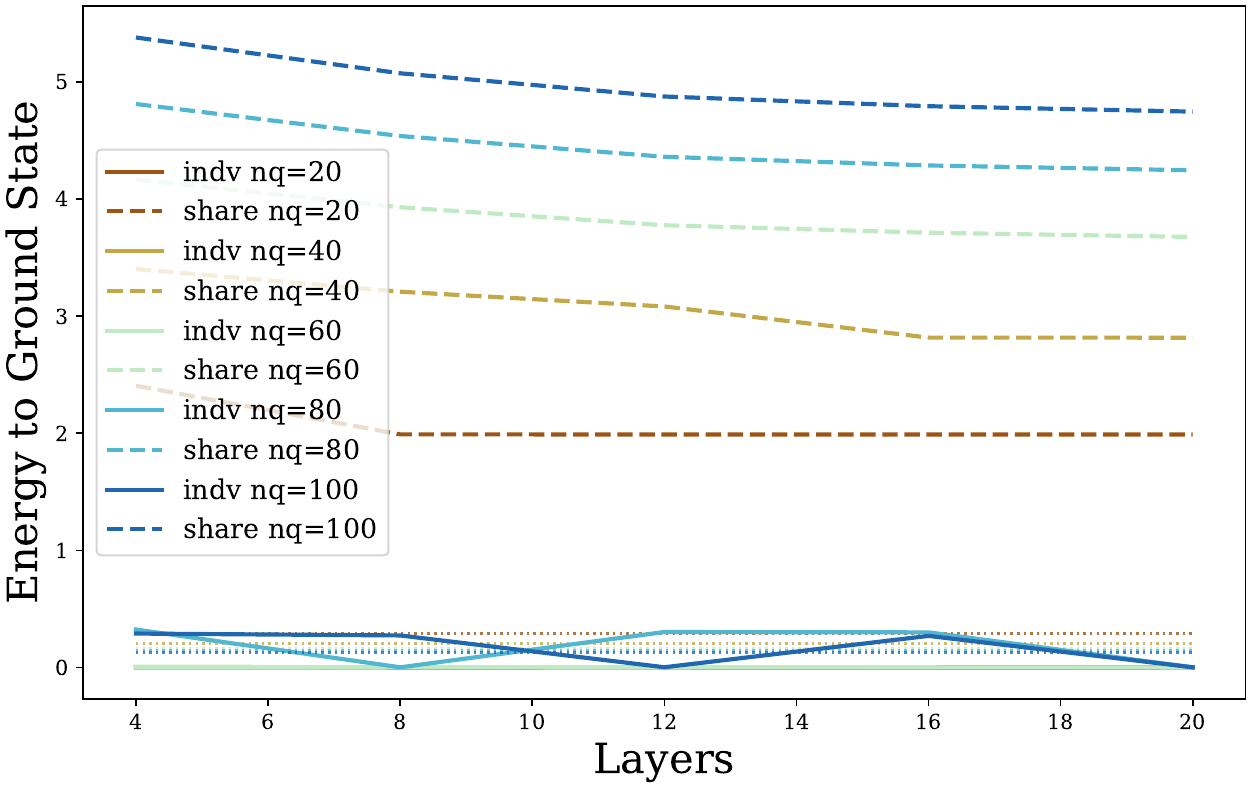}
 \caption{\tsbf{TFXY Individual and Shared Parameterizations.} Energy above the ground state for the staggered field target with amplitude $D=1.0$. Solid curves use individual parameters. Dashed curves use one shared parameter per layer.}
\label{fig:tfxy_indv_share}
\end{figure*}

\section{XY Ring Dynamics and Cardinality Constrained Optimization}

For $1\leq j<k\leq n$, define the normalized two site terms

\begin{align}
XY_{jk} &= \sigma_j^+\sigma_k^-+\sigma_j^-\sigma_k^+=\frac{1}{2}\l( X_jX_k+Y_jY_k \r), \\
YX_{jk} &= i\l( \sigma_j^+\sigma_k^- - \sigma_j^-\sigma_k^+ \r)=\frac{1}{2}\l( X_jY_k-Y_jX_k \r).
\end{align}

Before constructing the ring, consider the XY path generator set
\begin{align}
\G_{XY,Z}^{P} &= \l\{ XY_{j,j+1} \r\}_{j=1}^{n-1} \cup \l\{ Z_j \r\}_{j=1}^{n}, & \g_{XY,Z}^{P} &\cong \mf{u}\l( n \r).
\end{align}

\Cref{fig:xy_circ_su_rep} gives the $\mf{u}\l( n \r)$ representation of the XY path DLA. It makes the diagonal $Z$ directions and the off diagonal XY generators explicit.

Adding the single periodic coupling $XY_{1,n}$ gives the XY ring with local $Z$ rotations,

\begin{align}
\G_{XY,Z} &= \G_{XY,Z}^{P} \cup \l\{ XY_{1,n} \r\}.
\end{align}
where the indices are periodic. The following basis makes the two paths joining sites $j$ and $k$ explicit
\begin{align}
D_j &= Z_j, & \overline D_j &= \prod_{\ell\ne j}Z_\ell, \\
P_{jk} &= XY_{jk}\prod_{\ell=j+1}^{k-1}Z_\ell, & \overline P_{jk} &= XY_{jk}\l( \prod_{\ell=1}^{j-1}Z_\ell \r)\l( \prod_{\ell=k+1}^{n}Z_\ell \r), \\
Q_{jk} &= YX_{jk}\prod_{\ell=j+1}^{k-1}Z_\ell, & \overline Q_{jk} &= YX_{jk}\l( \prod_{\ell=1}^{j-1}Z_\ell \r)\l( \prod_{\ell=k+1}^{n}Z_\ell \r).
\end{align}
Thus $P_{jk}$ and $Q_{jk}$ use the path from $j$ to $k$ through the intermediate sites, while $\overline P_{jk}$ and $\overline Q_{jk}$ use the complementary path. With the independent elements of this family,
\begin{align}
\g_{XY,Z} &= \t{span}\l( \B^{\g_{XY,Z}} \r), \\
\B^{\g_{XY,Z}} &= \l\{ D_j,\overline D_j \r\}_j \cup \l\{ P_{jk},\overline P_{jk},Q_{jk},\overline Q_{jk} \r\}_{j<k}, \\
\g_{XY,Z} &\cong \mf{u}(1) \oplus \mf{su}(n) \oplus \mf{su}(n).
\end{align}

The displayed elements span the complete XY ring DLA, obtained by enlarging the XY path DLA with the periodic coupling while retaining quadratic dimension in the number of qubits~\cite{kordonowy2025lie}. For a chosen fixed Hamming weight initial state, we obtain the initial DOS vector by projecting $i\rho(0)$ onto the relevant orbital basis and evolve that vector under the complete ring generator set. This representation as a sum of Pauli strings is used for DOS construction.

\begin{suppfigure}[!t]
\centering
\begin{tabular}{c c}
\includegraphics[height=0.35\textheight]{imgs/circuit_z_xy.pdf} & \includegraphics[height=0.35\textheight]{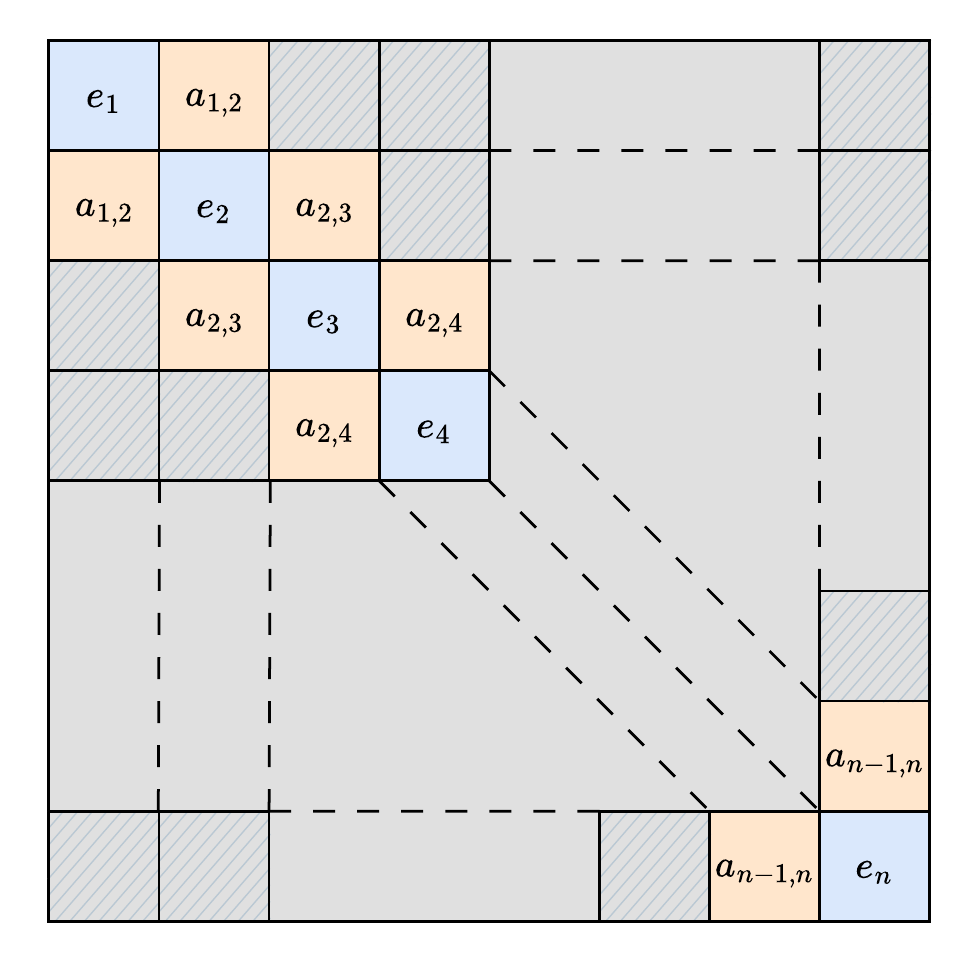} \\
(a) & (b)
\end{tabular}
 \caption{\tsbf{Representation of the XY Path with Single Spin $Z$ Rotations.} (a) depicts the circuit ansatz associated with the gate set $\G_{XY,Z}^{P}$. (b) depicts each generator as an element of $\mf{u}\l( n \r)$. Each evolution acts as a linear map determined by the commutation relations in $\mf{u}\l( n \r)$. The initial density operator has the diagonal representation shown in (b), where $Z_j\rightarrow e_j$.}
\label{fig:xy_circ_su_rep}
\end{suppfigure}

The XY ring target lies in $\g_{XY,Z}$, so its DOS is a DLA orbital. The cardinality constrained construction follows the same DLA and DOS logic as TFXY. A constrained cost with $Z$ and $ZZ$ terms is evaluated through separate $Z$ and $ZZ$ orbitals, with the $ZZ$ terms in $\Qc_{\g}^{2}$. The projected initial DOS states are
\begin{align}
i\vr_{\g_{XY,Z}}^{iZ}(0) &= \t{proj}_{\g_{XY,Z}}^{iZ}\l( i\rho(0) \r), \\
i\vr_{\g_{XY,Z}}^{iZZ}(0) &= \t{proj}_{\g_{XY,Z}}^{iZZ}\l( i\rho(0) \r),
\end{align}
obtained from the chosen fixed Hamming weight initial state using the same projection construction as above.

The evolution under $XY_{jk}$ follows from the Bell eigendecomposition. Let $P_{+} = \ketbra{B_{3}}{B_{3}}$ and $P_{-} = \ketbra{B_{4}}{B_{4}}$ where
\begin{align}
\ket{B_{3}} &= \frac{1}{\sqrt{2}} \l( \ket{10} + \ket{01} \r)  , \\
\ket{B_{4}} &= \frac{1}{\sqrt{2}} \l( \ket{10} - \ket{01} \r) .
\end{align}
Then $XY_{jk} = P_{+} - P_{-}$, since:
\begin{align}
XY_{jk} &= \ketbra{10}{01} + \ketbra{01}{10} \\
&= \frac{1}{2} \l( \ketbra{10}{01} + \ketbra{01}{10} \r) + \frac{1}{2} \l( \ketbra{10}{01} + \ketbra{01}{10} \r) \\
&\peq + \frac{1}{2} \l( \ketbra{10}{10} + \ketbra{01}{01} \r) - \frac{1}{2} \l( \ketbra{10}{01} + \ketbra{01}{10} \r) \\
&= \frac{1}{2} \l( \ket{01} + \ket{10} \r) \l( \bra{01} + \bra{10} \r) - \frac{1}{2} \l( \ket{01} - \ket{10} \r) \l( \bra{01} - \bra{10} \r)  \\
&= \ketbra{B_3}{B_3} - \ketbra{B_4}{B_4} .
\end{align}
Since $\la B_3 | B_4 \ra = 0$, $P_{+}$ and $P_{-}$ are orthogonal projectors with $P_{+} P_{-} = P_{-} P_{+} = 0$, so they commute and:
\begin{align}
e^{i \theta XY_{jk}} &= e^{i \theta \l( P_{+} - P_{-} \r)} = e^{i \theta P_{+}} e^{-i \theta P_{-}} . \label{eq:xy_bell_factor}
\end{align}

\begin{sremark}
The individual projector evolutions $e^{i \theta P_{+}}$ and $e^{-i \theta P_{-}}$ do not preserve the DOS associated with $\g_{XY,Z}$. Treated separately, they lead to an exponentially sized DLA. In particular,
\begin{align}
\t{proj}_{\g_{XY,Z}}^{iO}\l( \t{evo}_{iP_{-}}\l( \t{evo}_{iP_{+}}\l( i \vr ; \theta \r) ; -\theta \r) \r)
&\neq \t{evo}_{iP_{-}}\l( \t{proj}_{\g_{XY,Z}}^{iO}\l( \t{evo}_{iP_{+}}\l( i \vr ; \theta \r) \r) ; -\theta \r).
\end{align}
This failure does not occur for the complete XY evolution. The product $e^{i\theta XY_{jk}}=e^{i\theta P_{+}}e^{-i\theta P_{-}}$ lies in $e^{\g_{XY,Z}}$, so \cref{thmX:dosrho} allows projection to commute with this evolution. By evaluating the complete composed map and retaining every cross term, we maintain a fully valid representation of $e^{i \theta XY_{jk}}$ as a function of $\theta$ completely inside the DOS.
\end{sremark}

\begin{sthm}
For $i \vr \in \t{orb}_{\g_{XY,Z}}^{iO}$,
\begin{align}
\t{evo}_{i XY_{jk}}\l( i \vr ; \theta \r) = \t{evo}_{i P_{-}}\l( \t{evo}_{i P_{+}}\l( i \vr ; \theta \r) ; -\theta \r) \in \t{orb}_{\g_{XY,Z}}^{iO} .
\end{align}
\end{sthm}
\begin{proof}
Since $P$ is a projector, $\t{evo}_{iP}(X;\theta) = X + \sin\theta \, \t{ad}_{iP}(X) + (1-\cos\theta) \, \t{ad}_{iP}^{2}(X)$ exactly. Applying $\t{evo}_{iP_{+}}$ with $\theta$ and then applying $\t{evo}_{iP_{-}}$ with $-\theta$ and expanding
\begin{align}
\t{evo}_{iP_{-}}\l( \t{evo}_{iP_{+}}(i\vr;\theta);-\theta \r)
&= i\vr \nonumber \\
&\peq + \sin\theta \, \t{ad}_{iP_{+}}(i\vr) - \sin\theta \, \t{ad}_{iP_{-}}(i\vr) \nonumber \\
&\peq + (1-\cos\theta) \, \t{ad}_{iP_{+}}^{2}(i\vr) + (1-\cos\theta) \, \t{ad}_{iP_{-}}^{2}(i\vr) \nonumber \\
&\peq - \sin^{2}\theta \, \t{ad}_{iP_{-}} \t{ad}_{iP_{+}}(i\vr) \nonumber \\
&\peq - \sin\theta(1-\cos\theta) \, \t{ad}_{iP_{-}} \t{ad}_{iP_{+}}^{2}(i\vr) + \sin\theta(1-\cos\theta) \, \t{ad}_{iP_{-}}^{2} \t{ad}_{iP_{+}}(i\vr) \nonumber \\
&\peq + (1-\cos\theta)^{2} \, \t{ad}_{iP_{-}}^{2} \t{ad}_{iP_{+}}^{2}(i\vr) .
\end{align}
We construct the adjoint representation inside the DOS for each of the eight static operators and retain the identity contribution, following the logic presented in \cref{alg:op_adj}.

\end{proof}

\end{document}